\documentclass[11pt]{article}
\usepackage{amsmath,amssymb,amsfonts,amsthm,epsfig}
\usepackage[usenames,dvipsnames]{xcolor}
\usepackage{bm,xspace}
\usepackage{tcolorbox}
\usepackage{cancel}
\usepackage{fullpage}
\usepackage{guy}
\usepackage{framed}
\usepackage{verbatim}
\usepackage{enumitem}
\usepackage{array}
\usepackage{multirow}
\usepackage{afterpage}
\usepackage{mathrsfs}
\usepackage{pifont} 
\usepackage{chngpage}
\usepackage[normalem]{ulem}
\usepackage{boxedminipage}
\usepackage{caption}
\usepackage{forest}
\usepackage{svg}
\usepackage{microtype}
\usepackage{setspace}

\usepackage{algorithm}
\usepackage{algorithmicx}
\usepackage{algpseudocode}

\usepackage{pgfplots}
\pgfplotsset{width=8cm,compat=newest}

\def\colorful{1}

\ifnum\colorful=1

\fi
\ifnum\colorful=0

\fi

\crefname{claim}{claim}{claims}

\newcommand{\pparagraph}[1]{\bigskip \noindent {\bf {#1}}}

\begin{document}

\title{Differential privacy from axioms \vspace{10pt}}

\author{ 
Guy Blanc \vspace{6pt} \\ 
 {\sl Stanford University} \and 
William Pires \vspace{6pt} \\ 
 { {\sl Columbia University}} \and 
 Toniann Pitassi \vspace{6pt} \\
 { {\sl Columbia University}} \vspace{10pt}
}

\date{\small{\today}}

\maketitle
\begin{abstract}

Differential privacy (DP) is the de facto notion of privacy both in theory and in practice. However, despite its popularity, DP imposes strict requirements which guard against strong worst-case scenarios. For example, it guards against seemingly unrealistic scenarios where an attacker has full information about all but one point in the data set, and still nothing can be learned about the remaining point. While preventing such a strong attack is desirable, many works have explored whether average-case relaxations of DP are easier to satisfy \cite{HWR13,WLF16,BF16,LWX23}.

In this work, we are motivated by the question of whether   alternate, weaker  notions of privacy are possible: can a weakened privacy notion still guarantee some basic level of privacy, and on the other hand, achieve privacy more efficiently and/or for a substantially broader set of tasks? Our main result shows the answer is no: even in the statistical setting, any reasonable measure of privacy satisfying nontrivial composition is equivalent to DP. To prove this, we identify a core set of four axioms or desiderata: pre-processing invariance, prohibition of blatant non-privacy, strong composition, and linear scalability. Our main theorem shows that any privacy measure satisfying our axioms is equivalent to DP, up to polynomial factors in sample complexity. We complement this result by showing our axioms are minimal: removing any one of our axioms enables ill-behaved measures of privacy.

\end{abstract}
 \thispagestyle{empty}
 \newpage

 \setcounter{tocdepth}{2}
 \tableofcontents
 \thispagestyle{empty}
 \newpage

 \setcounter{page}{1}

\newcommand{\bTp}{\mathbf{T'}}
\newcommand{\calJ}{\mathcal{J}}
\newcommand{\calA}{\mathcal{A}}
\newcommand{\calM}{\mathcal{M}}
\newcommand{\calS}{\mathcal{S}}
\newcommand{\overlap}{\mathrm{overlap}}
\newcommand{\Unif}{\mathrm{Unif}}
\newcommand{\disjoint}{\mathrm{disjoint}}
\newcommand{\stabtv}{\mathrm{stab}_{\mathrm{tv}}}
\newcommand{\Ber}{\mathrm{Ber}}

\newcommand{\FindElement}{\textsc{FindElement}}
\newcommand{\FindLightElement}{\textsc{FindLightElement}}

\thispagestyle{empty}

\section{Introduction}
\label{sec:intro}
Differential privacy has won. It is the de facto formalization of privacy both in theory (see, e.g., the textbooks \cite{DR14book,Vad17book,NH21book}) and in practice (see, e.g., its use in the U.S. Census \cite{census} and by various technology companies \cite{AppleDP,Google23,DJY17}).

\begin{definition}[$(\eps,\delta)$-Differential Privacy, \cite{DMNS06,DKMMN06}]
    \label{def:approx-DP}
    A randomized algorithm $\mcM:X^n \to Y$ is $(\eps,\delta)$-DP if, for every $S, S' \in X^n$ differing in only one of the $n$ coordinates and $Y' \subseteq Y$,
    \begin{equation*}
        \Pr[\mcM(S) \in Y'] \leq e^{\eps} \cdot \Pr[\mcM(S') \in Y'] + \delta.
    \end{equation*}
\end{definition}
A large part of the reason that differential privacy (DP) has been so successful is the extensive toolkit of DP algorithms for a variety of basic primitives \cite{DR14book}. This toolkit can then be combined with \emph{strong composition}: The sequential combination of $k$-many of these primitives has a privacy loss ($\eps$ in \Cref{def:approx-DP}) that scales sublinearly in $k$ \cite{DRV10,KOV15}. This allows for efficient and simple construction of DP algorithms for a variety of tasks (see e.g. \cite{ACGMMTZ16} for how strong composition enables differentially private deep learning). This work is motivated by the following question.
\begin{question}
    \label{question:main}
    How inevitable was \Cref{def:approx-DP}? Is it possible to construct a materially different formulation of privacy that still satisfies strong composition?
\end{question}
A natural reason to suspect alternative definitions of privacy may be useful is that \Cref{def:approx-DP} guards against an incredibly strong, and in some cases unrealistic, attack. Even if the attacker is able to freely manipulate all but one point in the dataset, corresponding to the $n-1$ points $S$ and $S'$ agree on, they must still learn almost nothing about the one unknown point. In statistical settings, we model the entire dataset as being drawn from some unknown distribution $\bS \sim \mcD^n$, in which case the attacker is not nearly as strong as \Cref{def:approx-DP} suggests. That observation has motivated a number of  
relaxations of DP in which privacy must only be preserved on more ``typical" datasets \cite{HWR13,WLF16,BF16,LWX23}.

Our main result shows that we may as well use the worst-case definition of differential privacy.
\begin{quote}
     \textsl{Even in the statistical setting, any reasonable measure of privacy that satisfies strong composition is equivalent to \Cref{def:approx-DP} up to polynomial factors in the sample complexity.}
\end{quote}
To formalize this, we define the following four privacy axioms that we posit should be satisfied by any measure of privacy that is both  reasonable and useful.  
\footnote{By a privacy measure, we mean a scalar quantity ${\cal P}({\cal M})$ associated with an algorithm ${\cal M}$, and we say that ${\cal M}$ is ${\cal P}$-private if ${\cal P}({\cal M})$ is at most 1.} 
 
\begin{enumerate}
\item {\it Preprocessing:} Privacy is preserved under preprocessing.
Specifically, privacy should hold regardless of the ordering of the dataset, and regardless of the ordering of the domain.
 
\item {\it Prohibits blatant non-privacy:}
 
a private algorithm should not reveal almost all of the dataset. 
\item {\it Strong composition:} the privacy measure should grow sublinearly under composition. I.e., the composition of $\ell$-many  $\epsilon$-private  algorithms should be $O(\epsilon \ell^\delta)$-private, for some  $\delta<1$.
\item {\it Linear scalability:} the privacy measure should decrease linearly with the number of samples.
\end{enumerate}  
See Section \ref{sectionOurFramework} for a more detailed description of these axioms, and justification for why we view these axioms as both reasonable and usable.
 
With these axioms in place, our main results are captured by the following three theorems.
The first and most important theorem states that any algorithm satisfying our axioms is also differentially private:
\begin{theorem}[Our axioms imply differential privacy]
    \label{thm:axioms-imply-dp-intro}
    Let $\mcP$ be any privacy measure that satisfies \Cref{axiom:preprocessing,axiom:blalant,axiom:strong-composition,axiom:amplification} and $\mcM:X^n \to Y$ be any algorithm that is $\mcP$-private. For any $\eps, \delta> 0$ and $m \coloneqq \poly(n, 1/\eps, \log(1/\delta))$ there is an ($\eps$,$\delta$)-DP algorithm $\mcM':X^m \to Y$ that is equivalent (as in \Cref{def:equivalent}) to $\mcM$.
\end{theorem}
The exact polynomial in \Cref{thm:axioms-imply-dp-intro} depends on the constant $c$ in our our strong composition axiom (\Cref{axiom:strong-composition}). The best known constant for strong-composition is $c = 1/2$, in which case the sample-size in \Cref{thm:axioms-imply-dp-intro} would be $m \approx n^2$, provided the domain $X$ is not too small.\footnote{Our actual analysis case splits on the size of the domain, and gets a worse polynomial on very small domains.} We refer the reader to \Cref{sec:axioms-to-DP} for the full version of \Cref{thm:axioms-imply-dp-intro}.
We note that given the known equivalences between DP, replicability, and various notions of stability, 
\Cref{thm:axioms-imply-dp-intro} shows that these other notions are also implied by our axioms.

Second, we show that differential privacy satisfies our axioms. 
 
\begin{theorem}[Informal]\label{thm:DP-fits-axioms-intro}
Approximate differential privacy (\Cref{def:approx-DP}) satisfies \Cref{axiom:preprocessing-intro,axiom:blalant-intro,axiom:strong-composition-intro,axiom:amplification-intro}.

\end{theorem}

Lastly, we show that removal of any one of our axioms would allow for measures of privacy that do not intuitively align with any reasonable notion of privacy. To keep this overview succinct, we defer an in-depth discussion of what those nonsensical privacy measures are to \Cref{sec:axioms-min} (with a briefer overview in \Cref{subsec:overview-axioms-necessary}). We do have the following simple implication of those results.
\begin{theorem}[Minimality of our axioms]
    \label{thm:minimal-intro}
    \Cref{thm:axioms-imply-dp-intro} does not hold if $\mcP$ is allowed to not satisfy any one of \Cref{axiom:preprocessing-intro,axiom:blalant-intro,axiom:strong-composition-intro,axiom:amplification-intro}.
\end{theorem}

\medskip

\noindent{\bf Organization of Paper.} 
In \Cref{sectionOurFramework}, we present our framework and our axiomatic formulation of privacy;
in \Cref{sec:technical-overview}, we give high level overviews of the proofs of our main theorems, and in \Cref{sec:discussion-and-related-work} we discuss related work and several open problems raised by our framework.
After a brief preliminaries (\Cref{sec:prelims}), the remaining sections (\Cref{sec:axioms-to-DP,sec:DP-fits-axioms,sec:axioms-min}) give formal proofs of \Cref{thm:axioms-imply-dp-intro,thm:DP-fits-axioms-intro,thm:minimal-intro} respectively.

\section{Our Framework}
\label{sectionOurFramework}
All of our equivalences will hold with respect to algorithms that solve statistical tasks.

\begin{definition}[Statistical task, \cite{F17,BGHILPSS23}]
    \label{def:statistical-task}
    A statistical task is defined by a set of distributions $\mrD$ over data domain $X$, an output space $Y$ and a mapping $\mcT$ from distributions $\mcD \in \mrD$ to valid responses $\mcT(\mcD) \subseteq Y$. An algorithm $\mcM:X^n \to Y$ solves $\mcT$ with failure probability $\beta$ if, for all $\mcD \in \mrD$,
    \begin{equation*}
        \Prx_{\bS \sim \mcD^n}[\mcM(\bS) \in \mcT(\mcD)] \geq 1-\beta.
    \end{equation*}
\end{definition}  
Statistical tasks capture essentially any setting where the algorithm is learning from i.i.d. data. We note that in many such tasks, there is an error parameter $\eps$. This parameter is implicit in \Cref{def:statistical-task} as we can define $\mcT(\mcD)$ to only consist of outputs that are ``$\eps$-good." For example, if we aim to capture realizable PAC learning of a concept class $\mcC$ to error $1-\eps$, then $\mrD$ would consist of all distributions over labeled pairs $(\bx, \by)$ where $\by = f(\bx)$ for some single $f \in \mcC$ with probability $1$. The valid responses $\mcT(\mcD)$ would be any hypothesis $h$ satisfying $\Pr_{\bx, \by \sim \mcD}[h(\bx) = \by] \geq 1-\eps$. Our notion of equivalence will be agnostic to the particular statistical task an algorithm wishes to solve, and hence, automatically applies to all goals and error parameters.
\begin{definition}[Equivalent algorithm]
    \label{def:equivalent}
    We say an algorithm $\mcM':X^m \to Y$ is $(\beta,\beta')$-equivalent to $
    \mcM:X^n \to Y$ if, any statistical task that $\mcM$ solves with failure probability $\beta$, $\mcM'$ solves with failure probability $\beta'$.  
\end{definition}

\subsection{Privacy measures and our axioms}

To formalize \Cref{thm:axioms-imply-dp-intro,thm:DP-fits-axioms-intro,thm:minimal-intro} we define a series of axioms that any reasonable and useful \emph{privacy measure} should satisfy.  
\begin{definition}[Privacy measure]
    \label{def:privacy-measure}
    A \emph{privacy measure} is a mapping $\mcP$ from (possibly randomized) algorithms $\mcM:X^n \to Y$ to their level of privacy, parametrized as a number on $\R_{\geq 0}$. We adopt the convention that a lower values for $\mcP(\mcM)$ indicate that $\mcM$ is more private. It will often be useful to succinctly say that $\mcM$ is $\mcP$-private if $\mcP(\mcM) \leq 1$.
\end{definition}
We remark upon a few basic properties about \Cref{def:privacy-measure}. First, as is typical of previous definitions of privacy, a single privacy measure $\mcP$ must provide privacy levels for algorithms taking in samples of all sizes $n \in \N$. Later, our scaling axiom (\Cref{axiom:scaling}) will enforce some amount of consistency between how $\mcP$ behaves on different sample sizes.

Second, \Cref{def:privacy-measure} is a single parameter definition of privacy, in contrast to the two-parameters of DP (\Cref{def:approx-DP}). This single parameter was a deliberate choice. A guiding philosophy in the development of our axioms was to not directly enforce specific meaning to the privacy value $\mcP(\mcM)$, as we did not want our axioms to be biased by the meaning of $\eps$ and $\delta$ in DP. If we had a two (or more) parameter definition of privacy, we would need our axioms to somehow encode the distinction between those parameters, contradicting that guiding philosophy.

Furthermore, despite DP having two parameters, they are not of equal importance. Typical applications of DP simply set $\delta$ small enough to ignore and focus on $\eps$. Indeed, following the intuition that one only needs $\delta$ ``small enough," we show in \Cref{subsec:overview-DP-fits-axioms} how to collapse \Cref{def:approx-DP} into a single parameter in a way that respects all of our axioms.

\subsubsection{Axioms any reasonable definition of privacy should satisfy}
We now proceed to define our axioms, beginning with those that any ``reasonable" definition of privacy should satisfy. The first axioms encodes some basic operations that should maintain privacy.
\begin{axiom}[Preprocessing maintains privacy]
    \label{axiom:preprocessing}
    \label{axiom:preprocessing-intro}
    \label{axiom:preprocessing-formal}
    We say a privacy measure $\mcP$ satisfies the \emph{preprocessing axiom} if the following is true.
    \begin{enumerate}
        \item \textbf{Reordering the input maintains privacy:} For any algorithm $\mcM:X^n \to Y$ and permutation $\pi:[n] \to [n]$, defining
        \begin{equation*}
            \mcM \circ \pi(S) \coloneqq \mcM(S_{\pi(1)}, \ldots ,S_{\pi(n)}),
        \end{equation*}
        we have that $\mcP(\mcM \circ \pi) \leq \mcP(\mcM)$.
        \item \textbf{Remapping the domain maintains privacy:} For any mapping $\sigma:X \to X$ and algorithm $\mcM:X^n \to  Y$, defining
        \begin{equation*}
            \mcM \circ \sigma(S) \coloneqq \mcM(\sigma(S_{1}), \ldots ,\sigma(S_{n})),
        \end{equation*}
        we have that $\mcP(\mcM\circ \sigma) \leq \mcP(\mcM)$.
    \end{enumerate}
\end{axiom}

The first criteria, that reordering the input maintains privacy, says that under $\mcP$ it is equally bad to leak information about the $i^{\text{th}}$ point and $j^{\text{th}}$ point for any $i,j \in [n]$. The second criteria similarly says that it is equally bad to leak information about users $x$ and $x'$ for any $x,x' \in X$.

While both of these criteria are intuitively reasonable, we also provide more formal justification for their inclusion as axioms. In \Cref{sec:axioms-min}, we will show that removing any one of our four axioms would allow for ill-behaved privacy measures, illustrating why these axioms are necessary (see \Cref{subsec:overview-axioms-necessary} for a briefer overview). Since \Cref{axiom:preprocessing-intro} has two criteria, we will furthermore show that removing either of them would similarly result in ill-behaved privacy measures, helping to justify why both are necessary.

Our second axiom requires private mechanisms to not reveal (essentially) the entire dataset. This is the only axiom that directly enforces that $\mcP$ measures some notion of privacy.

\begin{definition}[Blatantly non-private]
    \label{def:blatant}
    A mechanism $\mcM:X^n \to Y$ is \emph{blatantly non-private} if there is a ``high-entropy" distribution $\mcD$ (formally $\mcD(x) \leq 1/(100n^2)$ for all $x \in X$) 
    and adversary $A$ mapping mechanism outputs $y \in Y$ to datasets $S' \in X^n$ for which\footnote{This constant of $0.9$ could be replaced with any $c < 1$.}
    \begin{equation*}
        \Ex_{\substack{\bS \iid \mcD^n\\ \bS' \leftarrow A(\mcM(\bS))}}\bracket*{\,\sum_{x \in \bS}\Ind[x \in \bS']} \geq 0.9n.
    \end{equation*}
\end{definition}
The ``high-entropy" requirement of \Cref{def:blatant} is designed to ensure the adversary's task is not too easy. In particular, it means that if the adversary were not able to see the $\mcM$'s output, it would not even be able to guess a single point in $\bS$. This stands in sharp contrast to the adversary being able to guess nearly all of $\bS$ upon seeing $\mcM$'s output.

\begin{axiom}[Prohibits blatant non-privacy]
    \label{axiom:blalant}
    \label{axiom:blalant-intro}
    \label{axiom:blalant-formal}
    We say a privacy measure $\mcP$ satisfies the \emph{prohibits blatant non-privacy} axiom if any $\mcP$-private algorithm is not blatantly non-private.
\end{axiom}

\subsubsection{Strong-composition axioms}

While the first two axioms were meant to be minimal requirements of any privacy definition to capture some reasonable notion of privacy, our next two axioms together formalize the notion of \emph{strong composition}. As discussed earlier, the fact that the privacy costs of \Cref{def:approx-DP} scale sublinearly with composition is crucial to the widespread adoption of differential privacy. Our next axiom encodes that the composition of $\ell$ many algorithms each of which have privacy level $\eps$ results in an algorithm with privacy level $\eps' \coloneqq \eps \cdot \ell^c$. We will state this in the minimal form we need: In particular, we only need that the composed algorithm is $\mcP$-private whenever $\eps' \leq 1$.

\begin{axiom}[Strong composition]
    \label{axiom:strong-composition}
    \label{axiom:strong-composition-intro}
    \label{axiom:strong-composition-formal}
    For $c < 1$, we say a privacy measure $\mcP$ satisfies $c$-\emph{strong composition} if for any algorithms $\mcM^1,\ldots, \mcM^\ell:X^n \to Y$ all satisfying $\mcP(\mcM^i) \leq \eps$ and
    \begin{equation*}
        \eps' \coloneqq \tilde{O}(\eps \cdot \ell^c) = O(\eps \cdot \ell^c \cdot \polylog(n)),
    \end{equation*}
    if $\eps' \leq 1$, then the composed algorithm $\mcM':X^n\to Y^{\ell}$ that takes in a sample $S \in X^n$ and outputs the $\ell$ responses $(\mcM^1(S), \ldots, \mcM^\ell(S))$ is $\mcP$-private.

\end{axiom}

Interestingly, we are able to define \Cref{axiom:strong-composition} to be qualitatively weaker than the strong composition DP satisfies. DP satisfies \emph{adaptive} strong composition, where the choice of $\mcM_i$ may depend adaptively on the outputs of $\mcM_1, \ldots, \mcM_{i-1}$. In contrast, \Cref{axiom:strong-composition} only requires strong composition to hold when $\mcM_1,\ldots, \mcM_\ell$ are fixed in advance. Yet, we are still able to show that our axioms imply DP. This shows, in some sense, that non-adaptive strong composition is enough to derive adaptive strong composition.

\Cref{axiom:strong-composition} on its own is not enough to enforce any reasonable notion of strong composition because it does not enforce any notion of scaling. For example, suppose we had some privacy measure $\mcP$ that only satisfied linear composition\footnote{As in the case for the average-case variants of DP defined in \cite{HWR13,WLF16,LWX23}} (\Cref{axiom:strong-composition} with $c = 1$). Then, we could simply define a new privacy measure $\mcP'$ as $\mcP'(\mcM) \coloneqq \sqrt{\mcP(\mcM)}$. This new measure would satisfy \Cref{axiom:strong-composition} with $c = 1/2$. Our last axiom rectifies this.

\begin{axiom}[Linear scalability]
 
    \label{axiom:amplification}
    \label{axiom:amplification-intro}
    \label{axiom:amplification-formal}
    \label{axiom:scaling}
    We say a privacy measure $\mcP$ satisfies \emph{linear scaling}, if for some polynomial $p:\R^2 \to \R$, any $\mcP$-private algorithm $\mcM:X^n \to Y$, any failure probability  $\beta >0 $, and any large enough $k \geq p(n,1/\beta)$, there exists a $(\beta, \beta' \coloneqq O(\beta))$-equivalent algorithm $\mcM'$ taking in $m \coloneqq kn$ samples that satisfies $\mcP(\mcM') \leq O(1/k)$.
\end{axiom}
Roughly speaking, linear scalability says that the privacy level can be improved by a factor of $1/k$ by increasing the sample size by a factor of $k$. For example, one common way to amplify privacy is \emph{subsampling}, meaning $\mcM'$ is the randomized algorithm which runs $\mcM$ on a uniform size-$n$ subsample of its size-$m$ input dataset. Indeed, for \Cref{def:approx-DP}, subsampling an $(\eps,\delta)$-DP algorithm by a factor of $k$ leads to an $(\eps/k, \delta/k)$-DP algorithm, though we will need a slightly more complicated amplification algorithm after we collapse $\eps$ and $\delta$ to a single parameter (see \Cref{lem:DP-to-DP}).

\Cref{axiom:strong-composition-intro,axiom:amplification-intro} are best viewed as together enforcing the following notion of strong composition. If the goal is to do a sequence of $\ell$ operations that each require a sample of size $n$ to perform privately, then only need a single sample size of $n \cdot \ell^{1-\Omega(1)}$. {That is, we require some non-trivial improvement over a strategy that, for example, uses $n$ separate samples for each of the $\ell$ operations.} We prefer this definition of strong composition in terms of the sample size required for $\ell$ many operations over explicit definitions that enforce a particular meaning to the value of $\mcP(\mcM)$ in lieu of \Cref{axiom:amplification-intro}.

\section{Technical Overview}
\label{sec:technical-overview}
\subsection{Overview of \Cref{thm:axioms-imply-dp-intro}: Our axioms imply DP}

Given any privacy measure $\mcP$ satisfying our axioms and $\mcP$-private algorithm $\mcM$, we wish to construct an equivalent $\mcM'$ that is $(\eps, \delta)$-DP. To do so, we use the following intermediate notion of stability.
 
\begin{definition}[TV-Stability, also called TV-indistinguishability by \cite{KKMV23})]
    \label{def:TV-stability}
    The \emph{TV-stability} of an algorithm $\mcM:X^n \to Y$ under distribution $\mcD$ is defined as
    \begin{equation*}
        \stabtv(\mcM, \mcD) \coloneq \Ex_{\bS, \bS' \iid \mcD^n}\bracket*{\dtv(\mcM(\bS), \mcM(\bS'))}.
    \end{equation*}
    We simply say $\mcM$ is \emph{$\rho$-TV-stable} if $\stabtv(\mcM, \mcD) \leq \rho$ for all distributions $\mcD$ over $X$.
\end{definition}
This definition is useful because (slight modifications) of the results of \cite{BGHILPSS23} allow us to convert any TV-stable algorithm into an equivalent DP algorithm (see \Cref{lem:TV-to-DP} for a formal statement of that conversion). Most of our effort goes into converting a $\mcP$-private algorithm into a TV-stable algorithm.

 \begin{theorem}[Our privacy axioms imply TV-stability]
      \label{thm:axioms-to-tv-overview}
    Let $\mcP$ be any privacy measure that satisfies \Cref{axiom:preprocessing,axiom:blalant,axiom:strong-composition,axiom:amplification} and $\mcM:X^n \to Y$ be any algorithm that is $\mcP$-private (that is, $\mcP ({\cal M}) \leq 1$.)
    For any constant $\rho > 0$ and $m \coloneqq \poly_{\rho}(n)$, there is a TV-stable algorithm $\mcM':X^m \to Y$ that is equivalent to $\mcM$.
 \end{theorem}

To prove \Cref{thm:axioms-to-tv-overview}, we show, roughly speaking, that for any non-TV-stable algorithm $\mcM:X^m \to Y$, there exists algorithms $\mcM_1, \ldots, \mcM_{\ell}$ for $\ell \approx m$ satisfying,
\begin{enumerate}
    \item Each $\mcM_i$ can be formed by preprocessing $\mcM$, and therefore, by the \Cref{axiom:preprocessing} (preprocessing), should have the same privacy.
    \item The composed algorithm $\mcM_{\mathrm{comp}}$ that takes as input $S$ and outputs the tuple $(\mcM_1(S), \ldots, \mcM_{\ell}(S))$ is blatantly non-private.
\end{enumerate}
By \Cref{axiom:blalant} (prohibition of blatant non-privacy), we can conclude that $\mcM_{\mathrm{comp}}$ is not $\mcP$-private. Then, \Cref{axiom:strong-composition} (strong composition) says that at least one $\mcM_i$ must satisfy $\mcP(\mcM_i) \geq \tilde{\Omega}(\ell^{-c})$. By \Cref{axiom:preprocessing} (preprocessing) this in fact means that $\mcP(\mcM) \geq \tilde{\Omega}(\ell^{-c}) = \tilde{\Omega}(m^{-c})$.

By contrapositive, this allows us to prove something just short of our goal: Any $\mcM$ satisfying sufficiently strong privacy, $\mcP(\mcM) \leq \tilde{O}(m^{-c})$, then $\mcM$ itself must be TV-stable.\footnote{We note that there are some caveats to this statement: Briefly, it only holds for \emph{symmetric} algorithms, those whose output does not depend on the order of its input, and assumes the domain is not too small. Both details are handled in the body.} In contrast, \Cref{thm:axioms-to-tv-overview} only assume that $\mcM$ is $\mcP$-private. Here, we can exploit linear scalability: Using \Cref{axiom:amplification}, we can convert any $\mcM:X^n \to Y$ that is $\mcP$-private to an $\mcM':X^m \to Y$ satisfying $\mcP(\mcM') \leq (1/m^{-c})$ with only a polynomial increase in the sample size. This is the step where we crucially utilize the combined power of linear scalability and strong composition: Ultimately, we want to convert any $\mcP$-stable algorithm using $n$ samples into one using $O(m)$ samples with the additional property that it can be composed $m$ times and still be $\mcP$-stable. \Cref{axiom:amplification,axiom:strong-composition} together allow us to do this.

\subsubsection{Exploiting TV-unstable algorithms}
\label{subsec:overview-tv-unstable}
The key step in proving \Cref{thm:axioms-to-tv-overview} is to show that if we compose $\approx m$ many preprocessed copies of a non-TV-stable algorithm $\mcM:X^m \to Y$, we will obtain a blatantly non-private algorithm. To prove this, we show a single random preprocessing reveals much information about the sample. It will be most convenient to state this lemma in terms of algorithms that take as input an unordered size-$m$ set as input, and we will use $\binom{X}{m}$ to denote all such sets.
\begin{restatable}[Key lemma, uniform permutations distinguish far samples]{lemma}{keylemma}
    \label{lem:key-lemma-overview}
    For any $\mcM:\binom{X}{m} \to Y$ where $|X| \geq 2m$, define
    \begin{equation}
        \label{eq:TV-disjoint-overview}
        \rho \coloneqq \Ex_{\bS, \bS' \sim \Unif\paren*{\binom{X}{m}}}\bracket*{\dtv(\mcM(\bS),\mcM(\bS')) \,\Big|\, \abs*{\bS \cap \bS'} =0}.
    \end{equation}
    Then, for any $S,S' \in \binom{X}{m}$ and $\bsigma:X \to X$ a uniform permutation,
    \begin{equation*}
        \Ex\bracket*{\dtv(\mcM \circ\bsigma(S),\mcM \circ \bsigma(S'))} \geq \frac{\rho}{2}\cdot \dist(S, S')/m,
    \end{equation*}
    where $\dist(S,S') \coloneqq m - |S \cap S'|$ is the number of points $S$ and $S'$ differ on.
\end{restatable}
Since we start with a $\mcM$ that is not TV-stable, the quantity $\rho$ in \Cref{eq:TV-disjoint-overview} is promised to be somewhat large. \Cref{lem:key-lemma-overview} says that, if we draw just one $\bsigma$, the algorithm $\mcM \circ \bsigma$ provides roughly ``$\Omega(1)$ bit" of useful information in distinguishing any $S$ and $S'$ that are somewhat far, satisfying $\dist(S, S') \geq 0.01n$. Since the number of possible datasets $S$ is $\binom{|X|}{m}$, it is possible to determine a dataset close to $S$ by observing $\mcM \circ \sigma_1, \ldots, \mcM \circ \sigma_{\ell}$ for $\ell \coloneqq O(\log \binom{|X|}{m}) = O(m \log |X|)$. We furthermore show in the body of the paper how to reduce to the case where $|X| = O(m^2)$, in which case $\ell = O(m \log m)$ suffices.

The key step in proving \Cref{lem:key-lemma-overview} is constructing the following random walk.
\begin{restatable}[Random walk to disjoint samples]{lemma}{randomwalk}
    \label{claim:random-walk-overview}
    For any $S, S' \in \binom{X}{m}$, setting $d \coloneqq \dist(S,S')$ and $k \coloneqq \ceil{m/d}$, there exists random variables $\bT^0, \ldots, \bT^k$ with the following properties:
    \begin{enumerate}
        \item For any $i \in [k]$ the marginal distribution of $(\bT^{i-1}, \bT^i)$ is equal to the distribution of $(\bsigma(S), \bsigma(S'))$ when $\bsigma:X \to X$ is a uniform permutation.
        \item The marginal distribution of $(\bT^0, \bT^k)$ is equal to the distribution of $\bU, \bU' \sim \Unif\paren[\big]{\binom{X}{m}}$ conditioned on $|\bU \cap \bU'| = 0$.
    \end{enumerate}
   
\end{restatable}

The intuition behind \Cref{claim:random-walk-overview} is that $\bT^i$ can be formed by ``rerandomizing" exactly $d$ many of the elements in $\bT^{i-1}$. As long as we have at least $m/d$ steps, we can ensure all elements get rerandomized. The actual proof of \Cref{claim:random-walk-overview} is a bit precise. In particular we need to use a non-Markovian walk (in that the distribution of $\bT^i$ is not independent of $\bT^1,\ldots, \bT^{i-2}$ conditioned on $\bT^{i-1}$) for the following reasons:
\begin{enumerate}
    \item In order to ensure all elements get rerandomized, the steps of the random walk cannot be independent. Instead, we enforce that the elements rerandomized in each step are different, while still ensuring that all the pairwise marginal $(\bT^{i-1}, \bT^i)$ have the right distribution.
    \item When $m/d$ is not exactly an integer some elements will be rerandomized twice. In this case, we need to ensure that no element accidentally gets rerandomized back into an element appearing in $\bT^0$ as that would cause $\dist(\bT^0, \bT^k) < m$.
\end{enumerate}
Nonetheless, we show with a careful construction that \Cref{claim:random-walk-overview} holds.

\subsection{Overview of \Cref{thm:DP-fits-axioms-intro}: DP satisfies the axioms}
\label{subsec:overview-DP-fits-axioms}

Since \Cref{def:approx-DP} has two parameters, $\eps$ and $\delta$, we must collapse them to one parameter for our framework. We do this by defining,
\begin{equation}
    \label{eq:DP-one-param}
    \mcP_{\mathrm{DP}}(\mcM) \coloneqq \argmin_v \set*{\mcM:X^n \to Y \text{ is }(\eps =\Theta(v^{4/5}), \delta=\Theta(v^{8/5}/n^3)\text{-DP}}.
\end{equation}
There is just one of many ways to collapse $\eps$ and $\delta$ into a single parameter in a way that respects our axioms. The exponents $4/5$ and $8/5$ could be replaced with $\alpha$ and $2\alpha$ for any $\alpha \in (0.5, 1)$. Furthermore, the $n^{3}$ factor could be replaced with any $n^{\beta}$ for $\beta > 3$. With this privacy measure, we can state the formal version of \Cref{thm:DP-fits-axioms-intro}.
We first state the formal version of \Cref{thm:DP-fits-axioms-intro}.

\setcounter{theorem}{1}
\begin{restatable}[DP implies our axioms, formal version]{theorem}{DPImplyAxioms}
    \label{thm:dp-fits-axioms-formal}

The privacy measure
 \begin{equation*}
      \mcP_{\mathrm{DP}}(\mcM) \coloneqq \argmin_v \set*{\mcM:X^n \to Y \text{ is }(\eps =\Theta(v^{4/5}), \delta=\Theta(v^{8/5}/n^3)\text{-DP}}
    \end{equation*}
     satisfies \Cref{axiom:preprocessing-intro,axiom:blalant-intro,axiom:strong-composition-intro,axiom:amplification-intro}.
\end{restatable}

The reason we want $\delta$ to be much smaller than $\eps$ is because that's the regime in which differential privacy satisfies strong composition.
The following well-known theorem shows that DP has strong composition.

\medskip
\noindent{\bf Theorem }(DP-Strong-Composition, Theorem 3.20 in \cite{DR14book}).

    {\it For all} $\epsilon, \delta \geq 0$, {\it if }$\mcM^1, \ldots, \mcM^\ell$ {\it are} $(\epsilon, \delta)$-{\it differentially private, then the composed algorithm }$(\mcM^1, \ldots, \mcM^\ell)$ {\it is} $(\epsilon', \delta')$-{\it differentially private where} $\delta' \coloneqq 2\ell\delta$ {\it and } $$\epsilon' \coloneqq \epsilon\sqrt{\ell \ln(1/(\ell\delta))}+ \ell\epsilon(e^{\epsilon}-1).$$

\bigskip

Given that we have forced $\delta$ to be small, we cannot simply use subsampling \cite{BBG18} to ensure that $\mcP_{\mathrm{DP}}$ satisfies linear scalability, as subsampling's effect on $\delta$ is too mild. Instead, we use (a small modification of) a recent result of \cite{BGHILPSS23} to prove $\mcP_{\mathrm{DP}}$ satisfies linear scalability. See \Cref{lem:DP-to-DP} for this amplification procedure and the surrounding discussion for comparison to \cite{BGHILPSS23}'s result. Given the well-known strong-composition theorem for DP and this amplification procedure, showing that $\mcP_{\mathrm{DP}}$ satisfies all our axioms is straightforward.

\subsection{Overview of \Cref{thm:minimal-intro}: Necessity of our axioms}
\label{subsec:overview-axioms-necessary}

Here, we explain why all four of our axioms are necessary. For each axiom, we exhibit ill-behaved notions of privacy that would be allowed if we removed the axiom. In the case of \Cref{axiom:preprocessing}, we even show this is true if only one of the two parts of it are removed, and in the case of \Cref{axiom:strong-composition}, we will show it is true even if we replace strong composition with linear composition (i.e. setting $c = 1$). The proof of \Cref{thm:minimal-intro} will build on these ill-behaved privacy measures by showing that they allow algorithms solving statistical tasks that no differentially private algorithm can solve. (see \Cref{sec:axioms-min} for details).

If we remove just the first part of \Cref{axiom:preprocessing}, that reordering the input maintains privacy, then there is a privacy measure satisfying the remaining axioms ($\mcP_{\mathrm{half}}$ from \Cref{def:first-half-privacy}) which deems the algorithm $\mcM:X^n \to X^{\floor{n/2}}$ that outputs the first half of its dataset perfectly private, satisfying $\mcP_{\mathrm{half}}(\mcM) = 0$. 

If we remove the second half of \Cref{axiom:preprocessing}, that remapping the domain maintains privacy, then there is a privacy measure satisfying the remaining axioms ($\mcP_{\text{heavy}}$ in \Cref{def:heavy-privacy}) that deems the following algorithm perfectly private: Let $\mcM:X^n \to X^n \cup \set{\varnothing}$ be the algorithm with the following behavior:
\begin{equation*}
    \mcM(S) = \begin{cases}
        S&\text{if there is some $x$ appearing at least $0.6n$ times in $S$.}\\
        \varnothing&\text{otherwise.}
    \end{cases}
\end{equation*}
Essentially, $\mcM$ is allowed to leak the entire dataset if there is any element appearing frequently enough. Despite this leakage, $\mcP_{\text{heavy}}(\mcM) = 0$, indicating that $\mcM$ should have ``perfect" privacy. We further observe (see \Cref{remark:permute-vs-remap}) that $\mcP_{\text{heavy}}$ still satisfies that \emph{permuting} the domain maintains privacy. This shows that we could not have replaced the arbitrary \emph{mappings} $\sigma:X \to X$ in \Cref{axiom:preprocessing} with arbitrary \emph{permutations} without allowing this ill-behaved notion of privacy.

If we remove \Cref{axiom:blalant} (prohibition of blatant non-privacy), then a privacy measure $\mcP_{\mathrm{all}}$ that deems \emph{all} algorithms perfectly private, i.e. $\mcP_{\mathrm{all}}(\mcM) = 0$ for all $\mcM$, satisfies the remaining axioms.

If we relax strong composition to linear composition, i.e. allow $c = 1$ in \Cref{axiom:strong-composition}, then there is a privacy measure,  $\mcP_{\mathrm{junta}}$ (see \Cref{def:privacy-junta}) with the following behavior: The algorithm $\mcM:X^n \to X^k$ which outputs the first $k$ points in its dataset satisfies $\mcP_{\mathrm{junta}}(\mcM) = \frac{k}{2n}$. For example, an algorithm which outputs the first half of its dataset is still $\mcP_{\mathrm{junta}}$-private. 

If we remove \Cref{axiom:amplification} (linear scaling), then there is a rescaling, $\mcP_{\sqrt{\mathrm{junta}}}$, of the above privacy measure that satisfies the remaining axioms. The algorithm $\mcM:X^n \to X^k$ which outputs the first $k$ points in its dataset is satisfies $\mcP_{\sqrt{\mathrm{junta}}}(\mcM) = \sqrt{\frac{k}{2n}}$. This still has essentially the same consequences as if we weakened \Cref{axiom:strong-composition}. For example, we still have that the algorithm which outputs the first half of its dataset is $\mcP_{\sqrt{\mathrm{junta}}}$-private.

\section{Discussion and Related Work}
\label{sec:discussion-and-related-work}
 
\textbf{Computational efficiency.} In \Cref{thm:axioms-imply-dp-intro}, we guarantee that any sample efficient $\mcP$-private algorithm $\mcM$ can be transformed into an equivalent DP algorithm $\mcM'$ with approximately the same sample complexity. {While our transformation is constructive, it does not necessarily preserve computational efficiency}. Part of the reason is that \Cref{axiom:amplification-intro} does not require the {scaling} to preserve computational efficiency, and we utilize a {scaled} version of $\mcM$ to construct $\mcM'$. This choice to allow for non-computationally efficient amplification is crucial to \Cref{thm:DP-fits-axioms-intro} as we utilize the following (computationally inefficient) procedure to prove that DP fits our axioms:

\bigskip

\noindent{\bf Theorem }(DP-Amplification, \cite{BGHILPSS23}.) \label{fact:DP-amp-intro}{\it For any} $(\eps = O(1), \delta = O(1/n^3))$-DP {\it algorithm} $\mcM:X^n \to Y$, {\it there exists an equivalent} $(\eps', \delta')$-DP {\it algorithm} $\mcM':X^m \to Y$ {\it using} $m \coloneqq 1/\eps' \cdot \poly(n, \log 1/\delta')$ {\it samples}.

\bigskip 

We remark that there is a computationally efficient way to amplify an $(\eps, \delta)$-DP algorithm to $(\eps/k, \delta/k)$-DP at the cost of a $O(k)$ increase in the sample size, via subsampling \cite{BBG18}. While subsampling's linear amplification of $\eps$ is as good as DP-Amplification, the linear amplification of $\delta$ is not sufficient for our purposes, and so we need to utilize the computationally inefficient amplification of
DP-Amplification.

As far as we are aware, despite it being of independent interest, it is unknown whether a computationally efficient analogue of DP-Amplification  
exists. More broadly, we leave open the possibility that it is possible to obtain a computationally efficient analogue of our results, possibly by adjusting the axioms appropriately.

\pparagraph{Other formalizations of differentially privacy.} We focused on the well-studied $(\eps,\delta)$-DP formulation of \Cref{def:approx-DP} (often called \emph{approximate DP}). One popular alternative, \emph{pure} DP, is equivalent to \Cref{def:approx-DP} where $\delta$ is fixed to be $0$. We did not focus on pure-DP because it does not satisfy strong composition, which makes it more difficult to utilize in practice and also that it does not fit our axioms. That said, it would be interesting to come up with an alternative set of axioms that characterize pure DP in the same sense as our axioms characterize approximate DP. One tempting solution is to simply remove our strong composition axiom (\Cref{axiom:strong-composition-intro}). However, as we show in \Cref{sec:axioms-min}, removing \Cref{axiom:strong-composition-intro} allows for a degenerate privacy measure which is much weaker than pure DP, so a different approach is needed.

A second popular generalization of approximate DP follows from the simple observation that algorithms are not $(\eps, \delta)$-DP for a single fixed choice of $\eps$ and $\delta$. Rather, for any algorithm $\mcM$, there is an entire ``curve" $\eps:\R_{\geq 0} \to \R_{\geq 0}$ for which $\mcM$ is $(\eps(\delta), \delta)$-DP for all choices of $\delta > 0$. There are a variety of formulations of DP that bound the behavior of this curve (e.g. through bounding appropriately defined ``moments") such as R\'eyni DP and concentrated DP \cite{M17,DR16,BS16}. These variations are popular precisely because they allow for easy (and often strong) composition, and in appropriate parameter regimes, also are amplified by subsampling. We refer the reader to \cite{steinke22} for an excellent overview.

Given this, it's natural to expect these variants would play nicely with our axioms. Indeed, we show in \Cref{sec:Reyni} that the privacy measure that assigns to $\mcM$ the smallest privacy value $v$ s.t. $\mcM$ is $(2,\sqrt{v})$-Réyni DP respects all our axioms with an even more straightforward analysis than the proof of \Cref{thm:DP-fits-axioms-intro} (see \Cref{subsec:Reyni-fits-axioms}). In the other direction, we show a variant of \Cref{thm:axioms-imply-dp-intro}, that our axioms imply R\'eyni DP. One distinction between that statement (\Cref{thm:axioms-to-RDP}) and \Cref{thm:axioms-imply-dp-intro} is that the R\'eyni DP algorithm has a sample size that depends on $\log \log |Y|$, which we show is necessary in \Cref{lem:sep-DP-RDP}. 

\bigskip

\noindent{\bf Related Work.} Perhaps most in the spirit of our results is recent work on reproducibility \cite{ILPS22}, and in particular the followup paper of Bun et al. \cite{BGHILPSS23} (see also \cite{KKMV23}). 
That work examines the broader context of \emph{algorithmic stability}, which are various ways of formalizing that an algorithms output does not depend too much on its input. They show that some of these measures of stability, replicability, max-information, and perfect generalization, are equivalent to differential privacy using the same formalization of equivalence as us. Measures of algorithmic stability and privacy share many of the same basic properties. In some sense, the only distinction between algorithmic stability and privacy is simply that measures of algorithmic stability were designed for applications other than privacy. Indeed, one could just as easily view our axioms as desirable properties for any measure of algorithmic stability. From this perspective, our work is a natural evolution of \cite{BGHILPSS23} as we show all measures of stability satisfying our axioms are equivalent to privacy.  We also utilize some of their techniques to prove our results.

More broadly, there have been several works formalizing axioms that any ``reasonable" definition of privacy should satisfy. Often this includes an axiom or assumption that privacy should be some measure of distance between the distributions $\mcM(S)$ and $\mcM(S')$ for worst-case $S$ and $S'$ (as in \Cref{def:approx-DP}). This includes \cite{KL10,S24}, which both investigate what measures of distance satisfy other reasonable axioms. Also in this spirit is the central limit theorem of \cite{DRS22}. Roughly speaking, it says that if we consider only privacy definitions based on some distance between $\mcM(S)$ and $\mcM(S')$, in the limit of many compositions, we may as well define ``Gaussian differential privacy." The key distinction between {all of these works} and ours is that we aim to justify why the most successful privacy definitions are measures of distance between $\mcM(S)$ and $\mcM(S')$ for worst-case $S$ and $S'$, {whereas previous works take that as an assumption or axiom.}

\grayblock{ 
}
 







\section{Preliminaries}
\label{sec:prelims}
Throughout this work, we will assume all domains are finite. For a domain $X$, we denote,
\begin{enumerate}
    \item All ordered tuples of $n$ elements as $X^n$.
    \item All ordered tuples of $n$ distinct elements as $\Xn$.
    \item All unordered sets of $n$ distinct elements as $\binom{X}{n}$.
    \item All permutations $\sigma:X \to X$ as $\mathfrak{S}(X)$.
\end{enumerate}
For a natural number $n$, we use $[n]$ as shorthand for $\set{1, \ldots, n}$. We use \textbf{boldface} letters, e.g.~$\bx,\bS$. For brevity, we will sometimes use $\bx \sim S$ as shorthand for $\bx \sim \Unif(S)$. For example, $\bsigma \sim \mathfrak{S}(X)$ indicates that $\bsigma$ is a uniform permutation mapping $X$ to itself.

It will sometimes be useful to convert any algorithm into its symmetric counterpart.
\begin{definition}[Symmetrization of an algorithm]
    \label{def:symmetrization}
    For a randomized algorithm $\mcM:X^n \to Y$, its \emph{symmetrization}, which we denote $\widetilde{\mcM}:X^n \to Y$, is defined as follows. On input $S \in X^n$, $\widetilde{\mcM}(S)$ first draws a uniform permutation $\bpi \sim \mathfrak{S}([n])$ and then outputs $\mcM(S_{\bpi(1)}, \ldots, S_{\bpi(n)})$.
\end{definition}

In particular, it's easy to see that the symmetrized algorithm $\widetilde{\mcM}$ doesn't depend on the order of its input. As such we will often abuse notation and view $\widetilde{\mcM}$ as both taking ordered sets and unordered sets as input. We also observe that for statistical tasks, permuting the input uniformly at random doesn't impact the correctness of the algorithm.

\begin{fact}\label{fact:symmetric}
    Let $\mcM:X^n \to Y$ be an algorithm with symmetrized version $\widetilde{\mcM}$. For any $1 \geq \beta \geq 0$, we have $\widetilde{\mcM}$ is $(\beta,\beta)$-equivalent to $\mcM$. 
\end{fact}
\begin{proof}
    Note that for any set $S \in X^n$ and permutation $\pi : [n] \to [n]$ and distribution $\mcD$ over $X$ we have:
    $\Prx_{\bS \sim \mcD^n}[\bS = S]=\Prx_{\bS \sim \mcD^n}[\bS = \pi(S)]$. So, let $\mcT$ be a statistical task that $\mcM$ solves with failure probability at most $\beta$, recalling that $\mcT(\mcD)$ is the set of correct answers to the task under distribution $\mcD$, we have:
\begin{align*}
    \Prx_{\bS \sim \mcD^n}[\widetilde{\mcM}(\bS) \in \mcT(\mcD)]&=\Prx_{\substack{\bS \sim \mcD^n\\ \bpi \sim \mathfrak{S}(n)}}[\mcM(\bpi(\bS)) \in \mcT(\mcD)] \\
    &=\Prx_{\substack{\bS \sim \mcD^n}}[\mcM(\bS) \in \mcT(\mcD)] \\
    &\geq 1-\beta.
\end{align*}
So $\widetilde{\mcM}$ also solves task $\mcT$ with failure probability at most $\beta$.
\end{proof}

We also recall the definition of TV-distance.
\begin{definition}
    Let $\mcD, \mcD'$ be two distributions over $Y$. The TV-distance between $\mcD$ and $\mcD'$ is defined as:
    $$\dtv(\mcD, \mcD'):=\max_{Y' \subseteq Y}\big| \Prx_{\by \sim \mcD}[\by \in Y] - \Prx_{\by \sim \mcD'}[\by \in Y] \big|.$$

    It will sometimes be convenient to work with the following equivalent definition:
     $$\dtv(\mcD, \mcD'):= \frac{1}{2}\lone{\mcD-\mcD'}=\frac{1}{2}\sum_{y \in Y}\big| \mcD(y)-\mcD'(y) \big|.$$
\end{definition}

Finally, we recall the postprocessing property of differential privacy. 
\begin{proposition}[Postprocessing, see Proposition 2.1 in \cite{DR14book}]
    Let $\mcM:X^n \to Y$ and $\mcA : Y \to Z$ be algorithms. If $\mcM$ is $(\epsilon, \delta)$-differentially private then the composed algorithm $\mcA \circ \mcM : X^n \to Z$ is $(\epsilon,\delta)$-differentially private. 
\end{proposition}
\begin{definition}[Distance between two sets]
     Let $S,S' \in X^n$ (or alternatively $S,S' \in \binom{X}{n}$), we define $\dist(S,S') := \sum_{i \in [n]}\mathbbm{1}[S_i \not \in S']$. 
\end{definition}

\section{Proof of \Cref{thm:axioms-imply-dp-intro}: Our axioms imply DP}
\label{sec:axioms-to-DP}

In this section, we prove that for any privacy measure $\mcP$ satisfying our axioms, any $\mcP$-private algorithm can be converted to a differential private one with only a polynomial increase in the sample size. We begin with a formal version of \Cref{thm:axioms-imply-dp-intro}.

\setcounter{theorem}{0}
\begin{restatable}[Our axioms imply DP, formal version]{theorem}{AxiomsImplyDP}
    \label{thm:axioms-to-dp-formal}
    Let $\mcP$ be any privacy measure satisfying \Cref{axiom:amplification-formal,axiom:blalant-formal,axiom:preprocessing-formal,axiom:strong-composition-formal} and $\mcM : X^n \to Y$ be any $\mcP$-private algorithm. For any $\eps, \delta, \beta > 0$ and, $c$ the constant in \Cref{axiom:strong-composition-formal} and $p$ the polynomial in \Cref{axiom:amplification-formal}, define
    \begin{equation*}
        m' \coloneqq \tilde{O}\left(\frac{r^2\cdot n^2 \cdot \log(1/\delta)}{\beta^2 \cdot \epsilon}\right) \text{ where } r= \max \left(n \cdot p(n,1/\beta), n^{\frac{1}{1-c}} \right).
    \end{equation*}
    Then there is an $(\eps,\delta)$-DP algorithm $\mcM'$ using $m'$ samples that is $(\beta, \beta' \coloneqq O(\beta))$-equivalent to $\mcM$.
\end{restatable}
\setcounter{theorem}{3}

This conversion relies on the two following lemmas. The first allows us to go from our axioms to TV-Stability. While the later one allows us to go from TV-Stability to differential privacy. 
\begin{restatable}[Our privacy axioms imply TV-Stability, formal version]{theorem}{AxiomsToTV}
    \label{lem:axioms-to-tv-formal}
    Let $\rho>0$ be a constant, $\mcP$ be any privacy measure satisfying \Cref{axiom:amplification-formal,axiom:blalant-formal,axiom:preprocessing-formal,axiom:strong-composition-formal} and $\mcM : X^n \to Y$ be any $\mcP$-private algorithm. Let be $c$ the constant in \Cref{axiom:strong-composition-formal} and $p$ the polynomial in \Cref{axiom:amplification-formal}, define
    \begin{equation*}
        m' \coloneqq \tilde{O}\left(\frac{r^2 \cdot n^2}{\beta^2}\right) \text{ where } r= \max \left(n \cdot p (n,1/\beta), n^{\frac{1}{1-c}} \right).
    \end{equation*}
    Then there is a $\rho$-TV stable $\mcM'$ using $m'$ samples that is $(\beta, \beta' \coloneqq O(\beta))$-equivalent to $\mcM$.
\end{restatable}
\setcounter{theorem}{4}

\begin{restatable}{lemma}{tvToDP}\label{lem:TV-to-DP}
    There is a universal constant $1>\rho^\star>0$ such that if $\mcM:X^n \to Y$ is an $\rho^\star$-TV-Stable algorithm, then there exists a $(\beta,5\beta)$-equivalent algorithm $\mcM':X^m \to Y$ which is $(\epsilon,\delta)$-differentially private using
    $$m=n \cdot O\Bigg(\log(1/\beta) \cdot \frac{\log(1/\beta)+\log(1/\delta)}{\epsilon}\Bigg) \text{ samples.}$$
\end{restatable}

With the above lemmas in mind, the proof of \Cref{thm:axioms-to-dp-formal} follows easily.
\begin{proof}[Proof of \Cref{thm:axioms-to-dp-formal}]
     We first apply \Cref{lem:axioms-to-tv-formal} to $\mcM$ to get a $(\beta,\beta')$-equivalent, $\beta'=O(\beta)$, and $\rho^\star$-TV-Stable algorithm $\mcM': X ^{m'} \to Y$ algorithm and using $$m'=\tilde{O}(m^2 \cdot n^2/\beta^2) \text{ where }m=n \cdot \max\left( n \cdot p(n, 1/\beta), n^{\frac{1}{1-c}} \right)\text{ samples. }$$

    By applying \Cref{lem:TV-to-DP} to $\mcM'$, we have that there exists a $(\beta',5\beta')$-equivalent algorithm $\mcM^*:X^{m^*} \to Y$ which is $(\epsilon,\delta)$-differentially private and using 
    $$m^*=m'\cdot O\left(\log(1/\beta) \frac{\log(1/\beta)+\log(1/\delta)}{\epsilon}\right)=\tilde{O}\left(\frac{m^2 \cdot n^2 \cdot \log(1/\delta)}{\beta^2 \cdot \epsilon}\right) \text{ samples.}$$

    In particular, $\mcM$ and $\mcM'$ are $(\beta,O(\beta))$-equivalent.  
\end{proof}

We prove \Cref{lem:TV-to-DP} in \Cref{sec:TV-to-DP}. In particular \Cref{lem:TV-to-DP} follows almost identically from the work of \cite{BGHILPSS23} who gave a transformation from perfect generalization (which is related to our notion of TV-Stability) to differential privacy. The main difference is that the  transformation of \cite{BGHILPSS23} makes the error probability of the algorithm go from $\beta$ to $O(\log(1/\beta) \beta)$. While our error only grows to $O(\beta)$, this also comes at the cost of a worse dependency on $\log(1/\beta)$ in sample complexity compared to the result of \cite{BGHILPSS23}. 

The rest of this section is dedicated to the proof of \Cref{lem:axioms-to-tv-formal}.

\pparagraph{Useful notation:} Note that throughout this section, we often use $\gamma$ to denote a pair of $(\sigma,\pi)$ where $\sigma \in \mathfrak{S}(X)$ and $\pi \in \mathfrak{S}([n])$. For ease of notation, we will use $\bgamma \sim \mathfrak{S}(X,n)$ to denote drawing such a pair uniformly at random. Finally, given $\mcM: X^n \to Y$, we denote by $\mcM \circ \gamma$, the algorithm such that $\mcM \circ \gamma(S)= \mcM \circ \sigma \circ \pi(S)$ (i.e. the algorithm reorders the elements according to $\pi$ and then remaps them according to $\sigma$ and then runs $\mcM$).

\subsection{A first step toward \Cref{lem:axioms-to-tv-formal}}

We begin with a simplified version of \Cref{lem:axioms-to-tv-formal}. While \Cref{lem:axioms-to-tv-formal} guarantees TV-Stability on \emph{all distributions}, this simplified setting will only guarantee we obtain an algorithm that is TV-stable on uniform distributions with small support. It will also assume $\mcM$ satisfies a stronger privacy guarantee than in \Cref{lem:axioms-to-tv-formal}. We will do away with these assumptions in \Cref{sec:remove-assumption}.  
\begin{restatable}[TV-Stability in a simplified setting]{lemma}{AxiomsToTVSimplerFormal}\label{lem:axioms-to-tv-simpler-formal}
    Fix any constant $\rho > 0$. For any domain $X^\star$ of size at least $\frac{100}{\rho} \cdot n^2$, privacy measure $\mcP$ satisfying \Cref{axiom:blalant-intro,axiom:preprocessing-intro,axiom:strong-composition-intro}, and algorithm $\mcM:(X^\star)^n \to Y$. There exists constant $\alpha>0$ and $t>0$ such that if $$\mcP(\mcM) \leq \frac{1}{\alpha \cdot n^c \log^{t}(n)},$$ where $c$ is the constant in \Cref{axiom:strong-composition-formal}, and $t$ depends on the $\polylog$ factors in \Cref{axiom:strong-composition-formal}, then the ``symmetrized" version of $\mcM$ (See \Cref{def:symmetrization}) denoted $\widetilde{\mcM}$

    satisfies $\stabtv(\widetilde{\mcM}, \mcD \coloneqq \Unif(X))) \leq \rho$ for every $X \subseteq X^\star$ of size $\ceil{\frac{100}{\rho}} \cdot n^2$. 
\end{restatable}

The domain size of $|X| \geq \Omega(n^2)$ ensures that $\bS \sim \Unif(X)^n$ is very likely to take on $n$ unique values. Furthermore, we work with the symmetrized version of $\mcM$ because the output behavior of $\widetilde{\mcM}(S)$ does not depend on the order of elements in $S$. 
The proof of \Cref{lem:axioms-to-tv-simpler-formal} proceeds by contradiction: We show that if $\widetilde{\mcM}$ is \emph{not} TV-stable under $\Unif(X)$, we can compose many ``preprocessed'' copies of $\widetilde{\mcM}$ to create an algorithm $\mcM'$ that should be $\mcP$-private according to \Cref{axiom:preprocessing-formal,axiom:strong-composition-formal}. However, using \Cref{lem:key-lemma-overview} we will show, in the following subsections, the existence of an adversary which given the output of $\mcM'(\bS)$ for $\bS \iid \Unif(X)^n$ can guess more than $90\%$ of the dataset. In particular, this implies $\mcM'$ is blatantly non-private, violating \Cref{axiom:blalant-formal}.

\subsubsection{Exploiting \Cref{lem:key-lemma-overview} and hypothesis selection}

Throughout this section we fix an algorithm $\mcM:X^n \to Y$, where the reader should think that $|X|=\Theta(n^2)$. We first show how to reduce the adversary's task to classical \emph{hypothesis selection}.
\begin{fact}[Hypothesis selection, Ch.~6 of \cite{DLbook01} or \cite{ABS24}]
    \label{fact:hypothesis-selection}
    For any distributions $\mcD_1, \ldots, \mcD_M$ there is an algorithm $\mcH$ with the following guarantee: Given
    \begin{equation*}
        \ell \coloneqq O\paren*{\frac{\log (M/\delta)}{\eps^2}}
    \end{equation*}
    i.i.d. samples from some unknown $\mcD_i$, $\mcH$ outputs a $j$ satisfying $\dtv(\mcD_i, \mcD_j) \leq \eps$ with probability at least $1-\delta$.
\end{fact}

Having fixed the algorithm $\mcM$, we will be interested in the following set of distributions.
\begin{definition}\label{def:distribution} 
For $S \in \Xn$, we define the distribution $\mcD_S$ has  the distribution $(\bgamma, {\mcM} \circ \bgamma(S))$ for $\bgamma \sim \mathfrak{S}(X,n)$. We denote by $\mcF$ the set of all such distribution $\mcD_S$. 
\end{definition}

As sketched in the introduction, we will use \Cref{lem:key-lemma-overview} on $\widetilde{\mcM}$ to argue that for any $S,S' \in \Xn$, if $S$ and $S'$ differ on $\Omega(n)$ coordinates, the distributions $\mcD_S$ and $\mcD_{S'}$ will be far. We present this idea formally in the next lemma. Since the number of possible datasets $S$ is at most $|X|^n$, the adversary, using hypothesis selection, can find a dataset close to an unknown set $S \in \Xn$ by observing $O(n \log |X|)$ many samples from $\mcD_S$.

\begin{restatable}{lemma}{corDistr}\label{cor:key_lemma_to_distribution}
Fix $\mcM:X^n \to Y$ with $|X| \geq 2n$, if for its ``symmetrized'' version $\tilde{\mcM}$ we have: 
\begin{equation*}
        \rho \coloneqq \Ex_{\bS, \bS' \sim {\binom{X}{n}}}\bracket*{\dtv(\widetilde{\mcM}(\bS),\widetilde{\mcM}(\bS')) \,\Big|\, \abs*{\bS \cap \bS'} =0}.
    \end{equation*}
Then, for any $S,S' \in \Xn$ we have
\begin{equation*}
    \dtv(\mcD_S, \mcD_{S'}) \geq \frac{1}{2}\cdot \dist(S, S')/n \cdot\rho.
\end{equation*}
\end{restatable}
\begin{proof}
    Let $S,S' \in \Xn$ and let $d=\dist(S,S')$. We let $\widetilde{S}, \widetilde{S'}$ denote the unordered version of $S,S'$. Note that $d=\dist(\widetilde{S},\widetilde{S'})$. By viewing $\widetilde{\mcM}$ as taking unordered sets as input and using \Cref{lem:key-lemma-overview}, we have:
    \begin{align*}
     \frac{\rho \cdot d}{2n}
     &\leq \Ex_{\bsigma}\bracket*{ \dtv( \widetilde{M} \circ \bsigma(\widetilde{S}),  \widetilde{M} \circ \bsigma{(\widetilde{S'}))}} \\
     &=\frac{1}{2}\Ex_{\bsigma}\bracket*{\lone{\widetilde{M} \circ \bsigma(\widetilde{S})-  \widetilde{M} \circ \bsigma{(\widetilde{S'})}}}.
\end{align*}

By the definition we have $\widetilde{\mcM}(S)=\sum_{\pi}\Pr[\pi]\mcM(S)$, hence:
\begin{align*}
 \frac{\rho \cdot d}{n}&=\Ex_{\bsigma}\bracket*{\lone{\sum_{\pi} \Pr[\pi] \mcM \circ \pi \circ \bsigma(S) -  \sum_{\pi} \Pr[\pi] \mcM \circ \pi \circ \bsigma(S')} }
        \\&=\Ex_{\bsigma}\bracket*{\lone{\sum_{\pi}\Pr[\pi]\left(\mcM \circ \pi \circ \bsigma({S}) -  \mcM \circ \pi \circ \bsigma({S}')\right)}} \\
        &\leq \Ex_{\bsigma}\bracket*{ \sum_{\pi}\Pr[\pi]\lone{\left(\mcM \circ \pi \circ \bsigma({S}) -  \mcM \circ \pi \circ \bsigma({S}')\right)}}\\ 
        &= \Ex_{\bpi, \bsigma}\bracket*{\lone{\left(\mcM \circ \bpi \circ \bsigma({S}) -  \mcM \circ \bpi \circ \bsigma({S}')\right)}}\\
        &= \Ex_{\bgamma=(\bsigma, \bpi)}\bracket*{\lone{\left(\mcM \circ \bgamma({S}) -  \mcM \circ \bgamma({S}')\right)}}.
    \end{align*}

On the other hand, we also have:
    \begin{align*}
     \lone{\mcD_S-\mcD_{S'}}&=\sum_{\gamma, y} \lone{D_{S}(\gamma,y)-\mcD_{S'}(\gamma,y)}  \\
     &=\sum_{\gamma} \Pr[\gamma]\cdot \sum_{y}\Big|\left(\mcM \circ \gamma(S)\right)(y)-\left(\mcM \circ \gamma(S)\right)(y)\Big|\\
     &=\Ex_{\bgamma} \bracket*{\sum_{y}\Big|\left(\mcM \circ \gamma(S)\right)(y)-\left(\mcM \circ \gamma(S)\right)(y)\Big|}\\
     &=\Ex_{\bgamma}\bracket*{\lone{\left(\mcM \circ \bgamma({S}) -  \mcM \circ \bgamma({S}')\right)}}
    \end{align*}

Combining the above yields, $\dtv(\mcD_S,\mcD_{S'})=\frac{1}{2}\lone{D_S-\mcD_{S'}} \geq \frac{\rho \cdot d}{2n}$ as claimed.  
\end{proof}

\subsubsection{The adversary and proof of \Cref{lem:axioms-to-tv-simpler-formal}}

We consider $\mcM:X^n \to Y$ where $|X|=\Theta(n^2)$. Fix some unknown set $S \in \Xn$. In the above subsection, we explained how in this setting the adversary, upon seeing $\ell=O(n\log(|X|)$ samples from $\mcD_S$ could use hypothesis selection to output a set $S'$ that has many elements in common with $S$. In particular, consider an adversary that first samples $\overline{\bgamma}=(\bgamma_1, \ldots, \bgamma_\ell) \sim \mathfrak{S}(X,n)^\ell$, and then runs 
\begin{equation}
    \label{eq:def-M-gamma}
    \mcM_{\overline{\bgamma}}:=(\mcM\circ\bgamma_1(S), \ldots, \mcM\circ\bgamma_\ell(S)).
\end{equation} 
Using the output of the algorithm, they build samples $(\bgamma_i, \mcM\circ\bgamma_i(S))_{i \in [\ell]}$ (which are distributed iid. from $\mcD_S$) and then run hypothesis selection. It's easy to see that such an adversary would succeed with high probability in finding $S'$ which overlaps a lot with $S$. So why are we not done already ?
\Cref{axiom:preprocessing-formal} only gives us guarantees on the privacy of  $\mcM \circ \gamma$ for a fixed $\gamma$; we have no guarantee on an algorithm that samples first $\bgamma$ and then runs $\mcM \circ \bgamma(S)$. This is why we need to ``derandomize'' the above strategy to argue there is a fixed choice of $\overline{\gamma}^\star$ such that the adversary succeeds with high probability when seeing output of $\mcM_{\overline{\gamma}^\star}$.

\begin{lemma}\label{lem:adversary_power}
    Let $\rho>0$ be a constant. Assume that for every $S,S' \in X^{(n)}$ we have $$\dtv(\mcD_S, \mcD_{S'}) \geq {\frac{\rho}{2}}\cdot \dist(S, S')/n.$$

    Then, there exists $\overline{\gamma} \in \mathfrak{S}(X,n)^\ell$ where $\ell=O(n\log(|X|))$ and adversary $\mcA: Y^{\ell} \to \Xn$ such that:

    $$\Ex_{\bS \sim {\Xn}}\bracket*{\Ex_{\bS' \leftarrow \mcA(\mcM_{\overline{\gamma}}(S))}\bracket*{\dist(S,\bS')}} \leq 0.02n$$
\end{lemma}
\begin{proof}
    We denote by $\mcH$ the hypothesis selection algorithm of \Cref{fact:hypothesis-selection} which we will aim to run on the set of candidates distribution $\mcF$. In particular, we let $\ell=O(\log(|\mcF|))=O(n\log(|X|))$ be large enough so that for every $\mcD_S \in \mcF$, with probability at least $0.99$ we have that $\mcH$ returns a distribution $\mcD_{S'}$ with $$\dtv(\mcD_S, \mcD_{S'}) < \frac{1}{100} \cdot \frac{\rho}{2}.$$ 
    
    For every ${\overline{\gamma}} \in \mathfrak{S}(X,n)^\ell$, we define an adversary $\mcA_{\overline{\gamma}}$ described in \Cref{fig:adversary} who runs the algorithm $\mcM_{\overline{\gamma}}(S)$. We will prove the following:
\begin{equation}\label{eq:adversary_advantage}
\Ex_{\substack{{\overline{\bgamma}} \sim \mathfrak{S}(X,n)^\ell}}\bracket*{\Ex_{ \substack{\bS \sim {\Xn}, \\ \bS' \leftarrow \mcA_{\overline{\bgamma}}(\mcM_{\overline{\bgamma}}(\bS))}} \bracket*{\dist(\bS, \bS')}} \leq 0.02n.
\end{equation}

    In particular, this implies there exists a fixed choice of  ${\overline{\gamma}}^\star$, such that that $${\Ex_{ \substack{\bS \sim \Xn, \\ \bS' \leftarrow \mcA_{{\overline{\gamma}}^\star}(\mcM_{{\overline{\gamma}}^\star}(\bS))}} \bracket*{{\dist(\bS, \bS')}}} \leq 0.02n,$$
    proving the lemma. We now turn to the prove of \Cref{eq:adversary_advantage}, we will simply show that for every $S \in \Xn$ (unknown to the adversary), we have $$\Ex_{\substack{{\overline{\bgamma}} \sim \mathfrak{S}(X,n)^\ell}}\bracket*{\Ex_{ \bS' \leftarrow \mcA_{\overline{\bgamma}}(\mcM_{\overline{\bgamma}}(S))} \bracket*{\dist(S, \bS')}} \leq 0.02n.$$

  Fix $S \in \Xn$, by our choice of $\ell$ and the properties of $\mcH$ we have:
  $$\Prx_{\substack{\bz_1, \ldots, \bz_\ell \iid\mcD_S\\ \bS' \leftarrow \mcH(\bz_1, \ldots, \bz_\ell)}}[\dtv(\mcD_S, \mcD_{\bS'}) \geq 0.01\rho/2] \leq 0.01$$ 

By assumption, we have that for any $S'$ with $\dist(S,S') \geq 0.01n$ we have $\dtv(\mcD_S,\mcD_{S'}) \geq \frac{1}{100}\frac{\rho}{2}$, so we can conclude that:

$$\Ex_{\substack{\bz_1, \ldots, \bz_\ell \iid\mcD_S\\ {\bS'} \leftarrow \mcH(\bz_1, \ldots, \bz_\ell)}}\bracket*{\dist(S, \bS')} \leq n \cdot 0.01+0.01n\cdot (1-0.01) \leq 0.02n$$

Finally note that the output of $\bS' \leftarrow \mcM_{\overline{\gamma}}$ is distributed as follows: First run $\mcM_{\overline{\gamma}}(S)$ to get $(\by_1, \ldots, \by_\ell)$ and then run $\mcH((\bgamma_1,\by_1), \ldots, (\bgamma_\ell,\by_\ell))$ to get $\bS'$. With this in mind, we have:
\begin{align*}
    \Ex_{\substack{{\overline{\bgamma}} \sim \mathfrak{S}(X,n)^\ell}}\bracket*{\Ex_{ \bS' \leftarrow \mcA_{\overline{\bgamma}}(\mcM_{\overline{\bgamma}}(S))} \bracket*{{\dist(S, \bS')}}} &=  \Ex_{{\overline{\bgamma}} \sim \mathfrak{S}(X,n)^\ell} \left[
    \Ex_{\substack{(\by_1,\ldots, \by_\ell) \leftarrow \mcM_{\overline{\bgamma}}(S) \\
    \mcD_{\bS'} \leftarrow \mcH( (\bgamma_1,\by_1), \ldots, (\bgamma_\ell,\by_\ell) )}}
    \bracket*{{\dist(S, \bS')}}\right]\\
    &= \Ex_{\bgamma_1, \ldots, \bgamma_\ell \iid \mathfrak{S}(X,n)} \left[
    \Ex_{\substack{\by_1 \leftarrow \mcM \circ \bgamma_1(S), \ldots, \by_\ell \leftarrow \mcM \circ \bgamma_\ell(S)  \\
    \mcD_{\bS'} \leftarrow \mcH( (\bgamma_1,\by_1), \ldots, (\bgamma_\ell,\by_\ell) )}}
\bracket*{\dist(S, \bS')}\right] \vspace{5pt}\\
&=\Ex_{\substack{\bz_1, \ldots, \bz_\ell \iid\mcD_S  \\
    \mcD_{\bS'} \leftarrow \mcH( \bz_1, \ldots, \bz_\ell )}}
\bracket*{\dist(S, \bS')}\\
&\leq 0.02n. \qedhere
\end{align*}

\end{proof}

\begin{figure}[H]
\captionsetup{width=.9\linewidth}

    \begin{tcolorbox}[colback = white,arc=1mm, boxrule=0.25mm]
    \vspace{2pt}
 \begin{enumerate}[nolistsep,itemsep=2pt]
        \item Run the algorithm $\mcM_{{\overline{\gamma}}}=(\mcM \circ \gamma_1, \ldots, \mcM \circ {\gamma_\ell})$ on the unknown set $S$.
        \item Let $(\by_1, \ldots, \by_\ell)$ be the output of $\mcM_{\overline{\gamma}}$. 
        \item Build a set of samples $\bZ=(\bz_1, \ldots, \bz_\ell)$ where $\bz_i=(\gamma_i, \by_i)$. 
        \item Runs the hypothesis selection algorithm $\mcH$ on the sample set $\bZ$ with the set of candidate distributions $\mcF$ (See \Cref{def:distribution}). 
        \item If $\mcH$ returns $\bS'$, do the same. 
    \end{enumerate}
    \end{tcolorbox}
\caption{The adversary $\mcM_{\overline{\gamma}}$ where ${\overline{\gamma}}=(\gamma_1, \ldots, \gamma_\ell) \in \mathfrak{S}(X,n)^\ell$. }
\label{fig:adversary}
\end{figure}

 In the above we assumed that the hidden set $S$ was in $\Xn$. By our choice of $|X| = \Theta (n^2)$, we have with high probability that $\bS \iid \Unif(X)^n$ has no duplicates. As such, the above adversary {can just ignore the (small) probability that $\bS \iid \Unif(X)^n$ has duplicates, and} will still output a set $\bS'$ such that $\E_{\bS \iid \Unif(X)^n}[\dist(\bS,\bS')] < 0.1n$, guaranteeing that $\mcM_{\overline{\gamma}}$ is blatantly non-private. With this in mind, we now present the proof of \Cref{lem:axioms-to-tv-simpler-formal}:

\begin{proof}[Proof of \Cref{lem:axioms-to-tv-simpler-formal}]
   Let $X \subseteq X^\star$ of the size $\ceil{\frac{100}{\rho}} \cdot n^2$. Assume for contradiction that $\stabtv(\widetilde{\mcM}, \Unif(X)) > \rho$. Let $E$ be the event over the draw of $\bS, \bS' \sim X^n $ that $\bS,\bS'$ have no common elements nor any duplicates. We have:
$$\Prx_{\bS, \bS' \sim \Unif(X)^n}[E] \geq \prod_{i=1}^{2n}(1-\rho \frac{i-1}{100 \cdot n^2}) \geq (1-\rho \frac{2n}{100 \cdot n^2})^{2n} \geq 1- \frac{\rho}{25}.$$

In which case we have: 
\begin{align*}
     \Ex_{\bS, \bS' \sim {\binom{X}{n}}}\bracket*{\dtv(\widetilde{\mcM}(\bS), \widetilde{\mcM}(\bS')) \,\Big|\, \abs*{\bS \cap \bS'} =0} &= \Ex_{\bS, \bS' \iid \Unif(X)^n}\bracket*{\dtv(\widetilde{\mcM}(\bS), \widetilde{\mcM}(\bS')) \,\Big|\, E} \\
     &\geq \Ex_{\bS, \bS' \iid \Unif(X)^n}\bracket*{\dtv(\widetilde{\mcM}(\bS), \widetilde{\mcM}(\bS'))} - (1- \Prx [E]) \\
     &\geq {\rho}-\frac{\rho}{2}\\
     &\geq \rho/2.
\end{align*}

Now consider the algorithm $\mcM': X^n \to Y$ which is just $\mcM$ restricted to inputs in $X$, we denote its symmetrized version by $\widetilde{\mcM'}$. It's easy to see that we also have:

$$ \Ex_{\bS, \bS' \sim {\binom{X}{n}}}\bracket*{\dtv(\widetilde{\mcM'}(\bS), \widetilde{\mcM'}(\bS')) \,\Big|\, \abs*{\bS \cap \bS'} =0} \geq \rho/2. $$

Hence, by \Cref{cor:key_lemma_to_distribution} we have that for any $S,S' \in X^{(n)}$ with $\dist(S,S') \geq 0.01n$ we have \begin{equation}
    \dtv(\mcD_S,\mcD_{S'}) \geq 0.01 \rho/4.
\end{equation}

By \Cref{lem:adversary_power}, we have there exists ${\overline{\gamma}}=(\gamma_1, \ldots, \gamma_\ell) \in \mathfrak{S}(X,n)^\ell$ where $\ell=O(n\log(|X|))=O(n\log(n))$ and adversary $\mcA: Y^{\ell} \to \Xn$ such that:

    $$\Ex_{\bS \sim {\Xn}}\bracket*{\Ex_{\bS' \leftarrow \mcA(\mcM'_{\overline{\gamma}}(S))}\bracket*{\dist(\bS,\bS')}} \leq 0.02n,$$
where $\mcM'_{\overline{\gamma}}$ is as defined in \Cref{eq:def-M-gamma}. Finally, it follows by our bound on the probability of $E$ that whenever we draw $\bS \iid \Unif(X)^n$, we have that $\bS \in \Xn$ (i.e. $\bS$ has  no duplicates) with probability at least $1-1/25 \geq 0.96$. 

As such we have:
\begin{align*}
    \Ex_{\bS \iid \Unif(X)^n }\bracket*{\Ex_{\bS' \leftarrow \mcA(\mcM'_{\overline{\gamma}}(\bS))}\bracket*{\sum_{x \in \bS}^n \mathbbm{1}[x \in \bS']}} 
    &\geq 0.96 \Ex_{\bS \sim \Unif\paren*{\Xn}}\bracket*{\Ex_{\bS' \leftarrow \mcA(\mcM'_{\overline{\gamma}}(\bS))}\bracket*{ \sum_{x \in \bS}^n \mathbbm{1}[x \in \bS']  }} \\
    &= 0.96 \Ex_{\bS \sim \Unif\paren*{\Xn}}\bracket*{\Ex_{\bS' \leftarrow \mcA(\mcM'_{\overline{\gamma}}(\bS))}\bracket*{n-\dist(\bS,\bS')}} \\
    &\geq 0.94n.
\end{align*}

Right now, this means there is an adversary which can leak most of the dataset when seeing the output of $\mcM'_{\overline{\gamma}}$, but we have no guarantee for the privacy of $\mcM'$. Hence, we now want to translate this to an adversary working for $\mcM$. For each $\gamma_i=(\sigma_i, \pi_i) \in {\overline{\gamma}}$ we define $\sigma_i^\star : X^\star \to X^\star$ as:

$$\sigma^\star_i(x)= \begin{cases}
    & x \text{ if } x \in X^\star \setminus X\\
    & \sigma_i(x) \text{ otherwise.}
\end{cases}$$

We then let $\gamma_i^\star=(\sigma_i^\star, \pi_i)$ and define ${\overline{\gamma}}^\star=(\gamma_1^\star, \ldots, \gamma_\ell^\star)$. It's easy to see that that for any $S \in \Xn$ we have the following distributions are equivalent $$\bS' \leftarrow \mcA(\mcM'_{\overline{\gamma}}(S)) \text{ and }\bS' \leftarrow \mcA(\mcM_{{\overline{\gamma}}^\star}(S)).$$

Hence, $$\Ex_{\bS \iid X^n }\bracket*{\Ex_{\bS' \leftarrow \mcA(\mcM_{{\overline{\gamma}}^\star}(\bS))}\bracket*{\sum_{x \in \bS}^n \mathbbm{1}[x \in \bS']}} \geq 0.94n.$$

Recall that $\ell=O(n\log(n))$, hence by making $\alpha,t >0$ large enough we have that $\epsilon:=\mcP(\mcM) \leq \frac{1}{\alpha n \log^t(n)}$ is so that, in \Cref{axiom:strong-composition-formal} (strong composition), we have $\epsilon'=O(\epsilon \cdot  \ell^c \cdot \polylog(n) \leq 1$.

We have that $\mcM_{{\overline{\gamma}}^\star}=(\mcM \circ \gamma^\star_1, \ldots, \mcM \circ \gamma^\star_\ell)$. In particular, for each $1 \leq i \leq \ell$, by \Cref{axiom:preprocessing-formal} we have that $\mcP(\mcM \circ \gamma^\star_i) \leq \mcP(\mcM)$. \footnote{Recall that $\gamma^\star_i=(\sigma^\star_i, \pi_i)$ for some permutation $\pi:[n] \to [n]$ and $\sigma_i=X \to X$. So this follows by applying both points of the axioms to $\mcM\circ \gamma_i=\mcM \circ \sigma_i^\star \circ \pi_i$. } 

Hence by \Cref{axiom:strong-composition-formal} we must have that $\mcM_{{\overline{\gamma}}^\star}$ is $\mcP$-private. But combined with existence of the adversary $\mcA$ we have that \Cref{axiom:blalant-formal} can't hold, reaching a contradiction. Hence, it must be that  $\stabtv(\widetilde{\mcM}, \Unif(X)) \leq \rho$. 
\end{proof}

\subsection{Proof of \Cref{lem:axioms-to-tv-formal}: Getting TV-Stability on all distributions}\label{sec:remove-assumption}
\Cref{lem:axioms-to-tv-simpler-formal}  only guarantees that $\widetilde{\mcM}$ is $\rho$-TV-stable on distributions that are uniform over a small set elements. To achieve \Cref{lem:axioms-to-tv-formal} we actually need to guarantee $\rho$-TV-Stability on \emph{all} distributions. To get around this, we use two separate arguments depending on the size of the domain. In the first case, when $|X|$ is large we will show the following:

\begin{restatable}[TV-Stability on one distribution to TV-Stability on all distributions]{claim}{transferstab}
    \label{claim:transfer-stability}
   For any $\mcM:X^n \to Y$ where $X \geq 8 n^2$, let $X' \subseteq X$ be any subset of size at least $8n^2$. For any privacy measure $\mcP$ satisfying \Cref{axiom:preprocessing-intro} and distribution $\mcD$ on $X$, there exists an algorithm $\mcM':X^n \to Y$ for which $\mcP(\mcM') \leq \mcP(\mcM)$ satisfying,
    \begin{equation}
        \label{eq:stab-close}
        \frac{1}{2} \cdot \stabtv(\widetilde{\mcM}, \mcD) \leq \stabtv(\widetilde{\mcM'}, \Unif(X')).
    \end{equation}
    where $\widetilde{\mcM}$ is the symmetrized version of $\mcM$ defined in \Cref{def:symmetrization}.
\end{restatable}

In the case where the domain is small, we take an entirely different approach. Instead, we show that any algorithm, without any assumption that it is $\mcP$-private, can be converted into a TV-stable algorithm.
\begin{lemma}[Any algorithm can be made TV-stable over small domains]\label{lem:learn_when_X_small}
    Let $\mcM:X^n \to Y$ be an algorithm and $\rho>0$ be a constant. There exists a $(\beta,2\beta)$ equivalent algorithm $\mcM':X^m \to Y$ such that $\mcM'$ is $\rho$-TV-stable and 
    $$m=O\left(\frac{|X|+\log(1/\beta)}{(\beta/n)^2}\right).$$
\end{lemma}

We first prove \Cref{claim:transfer-stability}. 
\begin{proof}[Proof of \Cref{claim:transfer-stability}]
    Let $\mcD$ be an arbitrary distribution over $X$ and $X' \subseteq X$ with $|X'| \geq 8n^2$. Consider the random map $\bsigma:X \to X$ where $\bsigma(x) \iid \mcD$. Consider the two following distributions: 
    \begin{enumerate}
    \item Draw $\bT, \bT' \sim \Unif(X')^n$ and $\bsigma$ as above and output $(\bsigma(\bT),\bsigma(\bT'))$. 
    \item Draw $\bS,\bS' \iid \mcD^n$ and output $(\bS,\bS')$. 
\end{enumerate}

Let $E$ denote the event where $\bT,\bT'$ contain no duplicates and no element in common. Observe that conditioned on $E$, both $\bsigma(\bT),\bsigma(\bT')$  and $(\bS,\bS')$ follow the same distribution.

We now have:

$$\Prx_{\bT, \bT'\sim \Unif(X')^n}[E]\geq \prod_{i=1}^{2n}\left(1-\frac{i-1}{|X'|}\right)\geq \left(1-\frac{2n}{|X'|}\right)^{2n}\geq 1/2.$$
Where the last inequality follows by having $|X'| \geq 8n^2$. Hence, we have:

\begin{align*}
     \stabtv(\widetilde{M}, \mcD)&=\Ex_{\bS,\bS' \sim \mcD^n}\bracket*{\dtv(\widetilde{\mcM}(\bS), \tilde{\mcM}(\bS')} \\
     &= \Ex_{\substack{\bT,\bT' \sim \Unif(X')^n\\\bsigma \sim \mathfrak{S}(X)}}\bracket*{\dtv(\widetilde{\mcM}\circ \bsigma(\bT), \widetilde{\mcM}\circ(\bT') \,\big|\, E }\\
     &\leq \frac{1}{\Prx_{\bT, \bT' \sim \Unif(X')^n}[E]} \cdot \Ex_{\substack{\bT,\bT' \sim \Unif(X')^n\\\bsigma}}\bracket*{\dtv(\widetilde{\mcM}\circ \bsigma(\bT), \widetilde{\mcM}\circ(\bT')}  \\
     &\leq 2 \Ex_{\substack{\bT,\bT' \sim \Unif(X')^n\\\bsigma}}\bracket*{\dtv(\widetilde{\mcM}\circ \bsigma(\bT), \widetilde{\mcM}\circ(\bT')}.\\
     &\leq 2 \Ex_{\bsigma} \bracket*{\Ex_{\substack{\bT,\bT' \sim \Unif(X')^n}}\bracket*{\dtv(\widetilde{\mcM}\circ \bsigma(\bT), \widetilde{\mcM}\circ(\bT')}}. \\
     &\leq 2\Ex_{\bsigma}\bracket*{\stabtv(\widetilde{\mcM} \circ \bsigma, \Unif(X'))}
\end{align*}

Hence, there exists a choice of $\sigma^\star:X \to X$, such that $\stabtv(\widetilde{\mcM} \circ \sigma^\star, \Unif(X')) \geq \frac{1}{2}\stabtv(\widetilde{M}, \mcD)$. Now, let $\mcM'= \mcM \circ \sigma^\star$. It's easy to see that $\widetilde{\mcM'}=\widetilde{\mcM} \circ \sigma^\star$, meaning $\mcM'$ respects \Cref{eq:stab-close}. Furthermore, by \Cref{axiom:preprocessing-formal} (preprocessing), we must have $\mcP(\mcM') \leq \mcP(\mcM)$. 
\end{proof}

We now turn to the proof of \Cref{lem:learn_when_X_small}. First, we will need the following well known fact, see for instance Theorem 1 of  \cite{canonne2020shortnotelearningdiscrete}: 
\begin{lemma}\label{lem:learn_distribution}
    Let $X$ be a finite set. There exists an algorithm taking $O\left(\frac{|X|+\log(1/\delta)}{\epsilon^2}\right)$ samples from an unknown distribution $\mcD$ over $X$ and outputs a distribution $\Sim_D$ such that $\dtv(\mcD,\bSim_{\mcD}) \leq \epsilon$ with probability at least $1-\delta$. 
\end{lemma}

\begin{proof}[Proof of \Cref{lem:learn_when_X_small}]
     Let $\mcD$ be an arbitrary distribution over $X$ We consider the algorithm $\mcM'$ from \Cref{fig:brute_force}, which uses the claimed bound on the number of samples from the distribution $\mcD$. It remains to show that $\mcM'$ outputs a correct answer with probability at least $1-2\beta$ and that it is TV-stable, meaning:

     \begin{equation*}
         \Ex_{\bS^1, \bS^2 \iid \mcD^m}\left[\dtv\left(\mcM'(\bS^1),\mcM'(\bS^2)\right)\right] \leq \rho.
     \end{equation*}

    To do this, let $\mcO$ be the distribution over $Y$ induced by sampling $\bS \sim \mcD^n$ and then running $\mcM(\bS)$. We will show that the following equation holds:

     \begin{equation}\label{eq:to_show}
         \Ex_{\bS \iid \mcD^m}[\dtv(\mcM'(\bS),\mcO)] \leq \frac{1}{2} \cdot \beta \cdot \rho
     \end{equation}

    TV-Stability then follows from the above by the triangle inequality. For the correctness of $\mcM'$, let $\mcT(\mcD)$ be the set of good answer to the statistical task $\mcT$ under distribution $\mcD$.  By the correctness of $\mcM$ we have that $\Prx_{\bA \sim \mcO}[\bA \in \mcT(\mcD)]=\Prx_{\bS \iid \mcD^n}[\mcM(\bS) \in \mcT(\mcD)] \geq 1-\beta$. Hence, assuming \Cref{eq:to_show} holds, we have: 
    
    \begin{align*}
        \Prx_{\bS \iid \mcD^m }[\mcM'(\bS) \in \mcT(\mcD)]&=\Ex_{\bS \iid D^m }\bracket*{\Pr[\mcM'(\bS) \in \mcT(\mcD)]} \\
        &\geq \Ex_{\bS \iid \mcD^m}\bracket*{\Prx_{\bA \sim \mcO}[\bA \in \mcT(\mcD)]-\dtv(\mcM'(\bS), \mcO)} \\
        &=\Prx_{\bA \sim \mcO}[\bA \in \mcT(\mcD)] - \Ex_{\bS \iid \mcD^n}[\dtv(\mcM'(\bS), \mcO)]\\
        &\geq 1-\beta-\beta.
    \end{align*}
    So $\mcM'$ has failure probability at most $2\beta$. We now prove \Cref{eq:to_show}. Let $\tau=\frac{\beta \cdot \rho}{4}$. By setting the hidden constant in $m$ to be large enough, we have by \Cref{lem:learn_distribution} that for any distribution $\mcD$, with failure probability $1-\tau$ the distribution $\bSim_{\mcD}$ obtained by the algorithm is such that $\dtv(\mcD,\bSim_{\mcD}) \leq \tau/n$. Conditioned on $\dtv(\mcD,\bSim_{\mcD}) \leq \tau$ we have:
    \begin{align*}
        \tau \geq \dtv\left(\bSim_\mcD^n, \mcD^n\right) \geq \dtv\left(\mcM(\bSim_{\mcD}^n\right), \mcM(\mcD^n))=\dtv\left(\mcM(\bSim_{\mcD}^n), \mcO\right)
    \end{align*}

    Hence, we have that:
    \begin{align*}
        \Ex_{\bS \iid \mcD^n}\left[\dtv\left(\mcM'(\bS), \mcO\right) \right] &\leq \tau +  \Ex_{\bS \iid \mcD^n}\left[\dtv\left(\mcM'(\bS), \mcO\right) \mid \dtv\left(\bSim_{\mcD}, D \right) \leq \tau/n \right] \\
        &= \tau + \Ex_{\bS \iid \mcD^n}\left[\dtv\left(\mcM(\bSim_\mcD^n), \mcO\right) \mid \dtv(\bSim_{\mcD}, D) \leq \tau/n \right] \\
        &\leq 2\tau.
    \end{align*}
By our choice of $\tau$, this proves \Cref{eq:to_show}    
\end{proof}

\begin{figure}[H]

  \captionsetup{width=.9\linewidth}

    \begin{tcolorbox}[colback = white,arc=1mm, boxrule=0.25mm]
    \vspace{2pt}
   
    \textbf{Input:} An algorithm $\mcM:X^n \to Y$, parameter $\beta>0$. Sample access to an unknown distribution $\mcD$ over $X$. 
 \begin{enumerate}[nolistsep,itemsep=2pt]
        \item Draw a set $\bS$ of $m$ samples $\mcD$ where $$m=O\Bigg(\frac{|X|+\log(1/\beta)}{\left({\beta}/{n}\right)^2}\Bigg).$$ 

        \item Run the algorithm of \Cref{lem:learn_distribution} to get a distribution $\bSim_{\mcD}$ over $X$. 
        \item Sample $\bS^\star \sim \bSim_\mcD^n$. 
        \item Output $\mcM(\bS^\star)$.
    \end{enumerate}
    \end{tcolorbox}
\caption{A procedure to turn $\mcM:X^n \to Y$ into a TV-Stable algorithm $\mcM'$ by learning the distribution $\mcD$. }
\label{fig:brute_force}
\end{figure} 

Finally, we turn to the proof of \Cref{lem:axioms-to-tv-formal}, which we recall bellow for convenience. 

\AxiomsToTV*

\begin{proof}
As mentioned in the overview, we will want to use \Cref{axiom:amplification-formal} (linear scalability). Let $t$ be the constant in \Cref{lem:axioms-to-tv-simpler-formal}. We will choose the smallest $k \geq p(n,1/\beta)$ satisfying, for $m = kn$, 
\begin{equation*}
   1/k \leq O\paren*{\frac{1}{ m^c \log^t(m)}}.
\end{equation*}
Ultimately, this will allow us to apply \Cref{lem:axioms-to-tv-simpler-formal}. Choosing any $k$ at least $\tilde{O}(n^{c/(1-c)})$ suffices for the above expression to hold. In which case $m=\tilde{O}(n^{1/(1-c)}$. 

We then set $m \coloneqq n\cdot k$ and case split on the size of the domain $X$.

\pparagraph{Small domains:} If $|X| \leq 100/\rho \cdot m^2$, we use \Cref{lem:learn_when_X_small} on $\mcM$ to get a $\rho$-TV-stable algorithm $\mcM':X^{m'} \to Y$, which is $(\beta, 2\beta)$-equivalent to $\mcM$ and where $$m'=O\left(\frac{m^2+\log(1/\beta)}{\beta^2} \cdot n^2\right).$$
Note that above fits the claimed bound on the sample complexity of the lemma. 

\pparagraph{Large domains:} If $|X| \geq 100/\rho \cdot m^2$. Since $k \geq p(n, 1/\beta)$, we can use \Cref{axiom:scaling} to obtain an algorithm to get an algorithm $\mcM^{\textsf{A}}: X^m \to Y$ which is $(\beta,\beta'=O(\beta))$-equivalent to $\mcM$ using $m$ samples satisfying
\begin{equation*}
    \mcP(\mcM^{\textsf{A}}) \leq O(1/k) \leq O\paren*{\frac{1}{m^{c} \cdot \log^t(m)}}.
\end{equation*}

We claim that the symmetrized version of $\mcM^{\textsf{A}}$, $\widetilde{\mcM^{\textsf{A}}}$ must be $\rho$-TV-Stable. Assume for contradiction it's not. Then, there exists a distribution $\mcD$ such that $\stabtv(\widetilde{\mcM^{\textsf{A}}}, D) > \rho$. We fix an arbitrary subset $X'$ of size $100/\rho \cdot m^2$ of $X$. By \Cref{claim:transfer-stability} we have that there exists an algorithm $\mcM':X^m \to Y$ with $\rho/2 < \stabtv(\widetilde{\mcM'}, \Unif(X'))$ and $\mcP(\mcM') \leq \mcP(\mcM)=O(1/m^{c}\log^t(m))$. 

By making the hidden constant in the $O(\cdot)$ small enough and using  $|X'|=\frac{100}{\rho} m^2$, we have by \Cref{lem:axioms-to-tv-simpler-formal}, that $\stabtv(\widetilde{\mcM'}, \Unif(X')) \leq \rho/2$. Hence, using \Cref{claim:transfer-stability}, it must have been the case that $\widetilde{\mcM^{\textsf{A}}}$ was $\rho$-TV-Stable. Furthermore, $\widetilde{\mcM^{\textsf{A}}}$ uses $m$ samples, which clearly matches the claim bounds on sample complexity of the lemma. Finally, by \Cref{fact:symmetric} we have that $\widetilde{\mcM^\textsf{A}}$ is $(\beta', \beta')$ equivalent to $\mcM^\textsf{A}$, and thus is $(\beta, O(\beta))$ equivalent to $\mcM$. 
\end{proof}

\subsection{The proof of \Cref{lem:key-lemma-overview}}\label{sec:key-lemma}

In this subsection, we prove \Cref{lem:key-lemma-overview}. Which we recall for convenience. 
\keylemma*

As mentioned in our overview, the first step toward the proof of \Cref{lem:key-lemma-overview} is \Cref{claim:random-walk-overview} which we recall bellow:

\randomwalk*

We first prove \Cref{lem:key-lemma-overview}, and delay the proof of \Cref{claim:random-walk-overview} to the end of this subsection.

    \begin{proof}
Let $S,S' \in \binom{X}{m}$ and let $d=\dist(S,S')$. By \Cref{claim:random-walk-overview} we have that there exists random variables $\bT^0, \ldots, \bT^k$ with $k=\ceil{m/d}$. Such that:
\begin{enumerate}
        \item For any $i \in [k]$ the marginal distribution of $(\bT^{i-1}, \bT^i)$ is the same as $(\bsigma(S),\bsigma(S'))$ where $\bsigma \sim \mathfrak{S}(X)$.
        \item The marginal distribution of $(\bT^0, \bT^k)$ is the same as drawing $\bT,\bT' \sim \Unif(\binom{X}{m})$ conditioned on $|\bT \cap \bT'|=0.$
    \end{enumerate}
We thus have:
    \begin{align*}
        \rho &=  \Ex_{\bT, \bT' \sim \Unif\paren*{\binom{X}{m}}}\bracket*{\dtv({\mcM}(\bT),{\mcM} (\bT')) \,\Big|\, \abs*{\bT \cap \bT'} =0 }\\
        &= \Ex_{\bT^0,\ldots, \bT^k}\bracket*{\dtv({\mcM}(\bT^0), {\mcM}(\bT^k))} \\
        &\leq \Ex_{\bT^0,\ldots, \bT^k}\bracket*{ \sum_{i=1}^{k}\dtv({\mcM}(\bT^{i-1}), {\mcM}(\bT^{i}))} \tag{Triangle inequality}\\
        &= \sum_{i=1}^{k} \Ex_{\bT^{i-1}, \bT^{i}}\bracket*{\dtv({\mcM}(\bT^{i-1}), {\mcM}(\bT^{i}))} \\
        &= \sum_{i=1}^{k} \Ex_{\bsigma \sim \mathfrak{S}(X) }\bracket*{\dtv({\mcM}( \bsigma(\bS)), {\mcM}(\bsigma(\bS'))} \\
         &= k \Ex_{\bsigma \sim \mathfrak{S}(X) }\bracket*{\dtv({\mcM} \circ  \bsigma(\bS), {\mcM}\circ \bsigma(\bS')}.
    \end{align*}

    Finally, since $k \leq 2m/d$ we have, 
    \begin{align*}
        \Ex_{\bsigma \sim \mathfrak{S}(X) }\bracket*{\dtv({\mcM}\circ \bsigma(\bS), {\mcM}\circ \bsigma(\bS'))} \geq \frac{\rho}{2}\cdot\frac{\dist(S,S')}{m}. &\qedhere
    \end{align*}
\end{proof}

\bigskip

\noindent{\bf Proof of \Cref{claim:random-walk-overview}.}  
To better understand the distributions of $(\bsigma(S), \bsigma(S'))$ when $\bsigma \sim \mathfrak{S}(X)$ we will work with the following distributions: 
\begin{definition}[Joint distribution at distance $d$]
    \label{def:joint-dists}
    For any distance $d \in [0,m]$ define $\mcJ_d$ to be the uniform distribution over all $T, T' \in \binom{X}{m}$ satisfying $\dist(T, T') = d$.
\end{definition}
In particular, we have the following:
\begin{restatable}[The joint distribution of uniform permutations]{proposition}{jointDistr}
    \label{prop:joint-perm}
    For any $S, S' \in \binom{X}{m}$ satisfying $\dist(S, S') = d$, if we draw $\bsigma \sim \mathfrak{S}(X)$, the joint distribution of $(\bsigma(S), \bsigma(S'))$ is exactly $\mcJ_{d}$.
\end{restatable}

\begin{proof}
    We observe that, for any permutation $\sigma:X \to X$ and set $T = \sigma(S)$, $T' = \sigma(S')$ then,
    \begin{equation*}
        \dist(T, T') = m - |\sigma(S) \cap \sigma(S')| = m - |S \cap S'| = \dist(S, S').
    \end{equation*}
    Furthermore, if we choose $\bsigma$ uniformly then, by symmetry, every choice of $\bT$ and $\bT'$ with distance $d$ are equally likely, giving the desired distribution.
\end{proof}

We now turn to the proof of \Cref{claim:random-walk-overview}.
\begin{proof}[Proof of \Cref{claim:random-walk-overview}]
We start by constructing $\bT^0, \ldots, \bT^k$. To begin, we pick $\bT^0 \sim \Unif(\binom{X}{m})$ uniformly among all size-$n$ samples. Then, for each $i \in [k]$ we will use the following procedure to form $\bT^i$:
\begin{enumerate}
    \item We choose $d$ many elements to remove from $\bT^{i-1}$. Here, there are two cases:
    \begin{enumerate}
        \item For the first $k-1$ steps (i.e. when $i < k$), there will be at least $d$ many elements remaining from $\bT^0$, and we choose the elements to remove uniformly from them:
        \begin{equation*}
            \bR^i \sim \binom{\bT^0 \cap \bT^{i-1}}{d}.
        \end{equation*}
        \item  In the last step, when $i = k$, the number of remaining elements in $\bT^0 \cap \bT^{i-1}$ will be $m - (k-1)d$. In particular, if $kd \neq m$ (which happens whenever $n/d$ is not exactly an integer), the number of remaining elements will be strictly less than $d$. In this case, we will remove all $m - (k-1)d$ elements plus $kd - m$ other uniform elements from $\bT^{i-1} \setminus \bT^0$,
        \begin{equation*}
            \bR^k = (\bT^0 \cap \bT^{i-1}) \cup \bE \quad\quad\text{where } \bE \sim \binom{\bT^{i-1} \setminus \bT^{0}}{kd - n}
        \end{equation*}
    \end{enumerate}
     \item We choose $d$ many elements to add to $\bT^{i-1}$. Here, we simply choose uniformly among all elements that have not appeared in the process yet,
    \begin{equation*}
        \bA^i \sim \binom{X \setminus(\bT^0 \cup \cdots \bT^{i-1})}{d}.
    \end{equation*}
    \item We then construct $\bT^i$ as:
    \begin{equation*}
        \bT^i = (\bT^{i-1} \setminus \bE^i) \cup \bA^i.
    \end{equation*}
\end{enumerate}
At each step, we construct $\bT^i$ by swapping exactly $d$ elements from $\bT^{i-1}$, so $\dist(\bT^{i-1}, \bT^i) = d$ with probability $1$. Furthermore, we eventually swap every element that started in $\bT^{0}$, meaning $\bT^k \subseteq \bA^1 \cup \cdots \cup \bA^k$. By construction, there is no overlap between $\bT^0$ and $\bT^k$, so $\dist(\bT^{0}, \bT^n) = m$ with probability $1$.

Finally, we observe this process is fully symmetric in the following sense: If we remapped the entire domain $X$ according to any permutation $\sigma:X \to X$, then all the probabilities remain the same:
\begin{equation*}
    \Pr[\bT^{0}, \ldots, \bT^{k} = T^{0},\ldots, T^k] =   \Pr[\bT^{0}, \ldots, \bT^{k} = \sigma(T^{0}),\ldots, \sigma(T^k)],
\end{equation*}
where $\sigma(T) = \set{\sigma(x) \mid x \in T}$. Due to this symmetry, $\bT^{i-1}$ and $\bT^i$ are equally likely to be any two sets with distance $d$, and so are distributed according to $\mcJ_d$. Similarly, $\bT^0$ and $\bT^k$ are equally likely to be any two sets with distance $m$, and as $\bT^0, \bT^k \sim \Unif\paren*{\binom{X}{m}}$ conditioned on $\bT^0 \cap \bT^k =\emptyset$. The proof then follows by \Cref{prop:joint-perm}. 
\end{proof}

\section{Proof of \Cref{thm:DP-fits-axioms-intro}: DP satisfies the axioms}
\label{sec:DP-fits-axioms}
First, we will need the following theorem to amplify $(\epsilon,\delta)$-differential privacy. This theorem is a small modification of Theorem 6.2 of \cite{BGHILPSS23} (mentioned in \Cref{fact:DP-amp-intro}) and is proved in \Cref{sec:DP-to-DP}.
The result of \cite{BGHILPSS23} makes the error probability of the algorithm go from $\beta$ to $O(\beta\log(\beta))$, whereas our modified Theorem  achieves error probability  $O(\beta)$, albeit with a slightly worse dependency on $\log(1/\beta)$ {in the sample complexity}. 

\begin{restatable}{lemma}{DPtoDP}\label{lem:DP-to-DP}
There is a universal constant $0.1 \geq \alpha > 0$ such that the following holds. Let $\mcM:X^n \to Y$ be $(0.1, \alpha^2/n^3)$-differentially private. Then for every $\epsilon, \delta> 0$; there exists an $(\epsilon,\delta)$-differentially private algorithm $\mcM':X^m \to Y$ solving $T$ using

$$m=O\Big(\frac{\log(1/\beta)^2+\log(1/\beta)\log(1/\delta)}{\epsilon}\Big) \cdot n^2 $$
samples and which is $(\beta, 5\beta)$-equivalent to $\mcM$. 
\end{restatable}

We now define our stability measure $\mcP_{\mathrm{DP}}$ as follows:
\begin{definition}
    For an algorithm $\mcM$ taking $n$ samples let $\epsilon$ be the smallest value such that $\calM$ is $(\epsilon,\epsilon^2/n^{3})$-DP. Recalling that $\alpha$ is the constant of \Cref{lem:DP-to-DP} we set $$\mcP_{\mathrm{DP}}(\mcM)= \left(\frac{ \epsilon}{\alpha}\right)^{5/4}.$$ 
\end{definition}   

We now prove that $\mcP_{\mathrm{DP}}$ fits all four of axioms. Before doing so we will need the following fact.
\begin{fact}
    Let $\mcM$ be an algorithm taking $n$ samples. If $\mcP_{\mathrm{DP}}(\mcM) \leq 1$, then $\mcM$ is $(0.1, \frac{\alpha^2}{n^{3}})$-DP.
\end{fact}

\begin{lemma}
    $\mcP_{\mathrm{DP}}$ respects \Cref{axiom:preprocessing-formal}.

\end{lemma}
\begin{proof}
We observe for any $S, S'$ that differ in one coordinate and permutation $\pi:[n] \to [n]$, that $\pi(S)$ and $\pi(S')$ still differ in one coordinate. Similarly, for any $\sigma:X \to X$, $\sigma(S)$ and $\sigma(S')$ differ in (at most) one coordinate. Therefore, if $\mcM$ is $(\eps,\delta)$-DP, it is still $(\eps,\delta)$-DP after preprocessing.
\end{proof}

To prove that DP satisfies \Cref{axiom:blalant-formal}, we will need the following fact:
\begin{fact}\label{fact:DP-adversary}
    Let $\mcM:X^n \to Y$, $\mcA:Y \to X^n$ be algorithms. Fix $S \in X^n, i \in [n]$ and $x \in X$. If $\mcM$ is $(\epsilon, \delta)$-differentially private then we have:
    $$\Pr\bracket*{S_i \in \mcA(\mcM(S)) } \leq e^{\epsilon} \Pr\bracket*{S_i \in \mcA(\mcM(S_{x \to i})) }+\delta,$$
    where $S_{x \to i}$ denotes $S$ with $i$-th element set to $x$. 
\end{fact}
\begin{proof}
    $S$ and $S_{x \to i}$ differ on at most $1$ coordinate. The result thus follows from $\mcM$ being $(\epsilon,\delta)$-differentially private.
\end{proof}

\begin{lemma}
    $\mcP_{\mathrm{DP}}$ respects \Cref{axiom:blalant-formal}. 
\end{lemma}
\begin{proof}
Let $\mcD$ an arbitrary distribution over $X$ with $||\mcD||_\infty \leq 1/100n^2$. Let $\mcM : X^n \to Y$ be an $(\epsilon,\delta)$-differentially private algorithm and let $\mcA:Y \to X^n$.
    We will show that:
    $$\Ex_{\substack{\bS \iid \mcD^n \\ \bS' \leftarrow \mcA(\mcM(\bS))}}\bracket*{\sum_{x \in \bS}\mathbbm{1}[x \in \bS']} \leq \frac{e^{\epsilon}}{100}+\delta n.$$

    In particular, if $\mcP_{\mathrm{DP}}(\mcM) \leq 1$, which implies $\mcM$ is $(\epsilon=0.1, \delta= \alpha^2/n^3)$-DP, we have $\frac{e^{\epsilon}}{100}+n\delta \leq 0.1$ proving the lemma. We want to bound: 
    $$ \Ex_{\substack{\bS \iid \mcD^n \\ \bS' \leftarrow \mcA(\mcM(\bS))}}\bracket*{\sum_{x \in \bS}\mathbbm{1}[x \in \bS']} = \sum_{i=1}^n \Prx_{\substack{\bS \iid \mcD^n \\ \bS' \leftarrow \mcA(\mcM(\bS))}}[\bS_i \in \bS'] $$

    Since $\mcM$ is $(\eps,\delta)$-DP, for any $i \in [n]$ and $x \in X$ we have, by \Cref{fact:DP-adversary}, that:

$$\Prx_{\substack{\bS \iid \mcD^n }}\bracket*{\bS_i \in \mcA\left(\mcM(\bS)\right)} \leq e^{\epsilon} \Prx_{\substack{\bS \iid \mcD^n }}\bracket*{\bS_i \in \mcA\left(\mcM(\bS_{x \to i})\right)}+\delta.$$
    
     Where $S_{x \to i}$ is the set obtained by setting the $i$-th element of $S$ to $x$. This implies that: 
     \begin{align*}
         \Prx_{\substack{\bS \iid \mcD^n \\ \bS' \leftarrow \mcA(\mcM(\bS))}}[\bS_i \in \bS'] &= \Prx_{\substack{\bS \iid \mcD^n }}[\bS_i \in \mcA(\mcM(\bS))] \\
         &\leq e^{\epsilon} \Prx_{\substack{\bS \iid \mcD^n\\ \bx \sim \mcD }}[\bS_i \in \mcA(\mcM(\bS_{\bx \to i}))]+\delta. \\
         &= e^{\epsilon} \Prx_{\substack{\bS \iid \mcD^n\\ \bx \sim \mcD }}[\bx \in \mcA(\mcM(\bS))]+\delta.\\
     \end{align*}
    Where the last line follows from the symmetry of $\bS_i$ and $\bx$. Hence, we have that: 

     \begin{align*}
         \Ex_{\substack{\bS \iid \mcD^n \\ \bS' \leftarrow \mcA(\mcM(\bS))}}\bracket*{\sum_{x \in \bS}\mathbbm{1}[x \in \bS']} &= \sum_{i=1}^n \left(e^{\epsilon} \Prx_{\substack{\bS \iid \mcD^n\\ \bx \sim \mcD }}[\bx \in \mcA(\mcM(\bS))]+\delta \right)  \\
         &\leq n \cdot e^{\epsilon} \sup_{S \in X^n}\bigg(\Pr_{\bx \sim \mcD}[\bx \in S]\bigg)+\delta n. \\
     \end{align*} 
    
Since $\linf{D} \leq \frac{1}{100n^2}$, we have that for any $S \in X^n$, $\Pr_{\bx \sim \mcD}[\bx \in S] \leq \frac{n}{100n^2}$. From which we can conclude that \Cref{axiom:blalant-formal} holds since:   

\begin{align*}
    \Ex_{\substack{\bS \iid \mcD^n \\ \bS' \leftarrow \mcA(\mcM(\bS))}}\bracket*{\sum_{x \in \bS}\mathbbm{1}[x \in \bS']} \leq \frac{e^{\epsilon}}{100}+\delta n. &\qedhere
\end{align*}
\end{proof}

Before proving \Cref{axiom:strong-composition-formal} holds, we recall advanced composition for $(\epsilon,\delta)$-differential privacy.\medskip

\noindent{\bf Theorem }(DP-Strong-Composition, Theorem 3.20 in \cite{DR14book}).
     
    {\it For all} $\epsilon, \delta \geq 0$, {\it if }$\mcM^1, \ldots, \mcM^\ell$ {\it are} $(\epsilon, \delta)$-{\it differentially private, then the composed algorithm }$(\mcM^1, \ldots, \mcM^\ell)$ {\it is} $(\epsilon', \delta')$-{\it differentially private where} $\delta' \coloneqq 2\ell\delta$ {\it and } $$\epsilon' \coloneqq \epsilon\sqrt{\ell \ln(1/(\ell\delta))}+ \ell\epsilon(e^{\epsilon}-1).$$
\medskip

\begin{lemma}
     $\mcP_{\mathrm{DP}}$ respects \Cref{axiom:strong-composition-formal} with composition constant $c=5/8$.
\end{lemma}

\begin{proof} 
Let $\beta$ be a suitably large constant, we define $\epsilon':=\ell^c \cdot \epsilon \cdot \left({\beta \log^2(n)}\right)^{2c}$. Say we have algorithms ${\calM}^1, \ldots, {\calM}^\ell$ each taking $n$ samples with $\epsilon \geq \mcP_{\mathrm{DP}}(\mcM ^i)$. For $\epsilon'$ to be less than $1$ we need $$\epsilon \leq \ell^{-c \cdot }\left(\beta\log^2(n)\right)^{-2c}.$$

This means that each ${\calM}^i$ is $(\mu, \mu^2/n^{3})$-DP, where $\mu=\alpha \epsilon^{4/5}$. By our choice of $c$ we have $$\mu = \frac{\alpha}{\beta  \log^2(n)} \cdot \ell^{-1/2}.$$
 
 We denote by $\mcC$ the composed algorithm $\left({\calM}^1, \ldots, {\calM}^\ell\right)$, to prove the Axiom holds, it remains to show ${\mcP_{\mathrm{DP}}}(\mcC) \leq 1$. This is equivalent to showing $(0.1, \alpha^2/n^{3})$-differentially private. First note that $\mu \leq 1$ and $\mu^2\ell\leq \alpha <1$, which implies  $$\mu(e^{\mu}-1) \leq 2\mu^2 \text{ and }\mu^2\ell \leq \mu \sqrt{\ell}.$$

Using Strong-Composition, we can conclude $\mcC$ is $(\epsilon^\star,\delta^\star)$-DP where $\delta^\star=2\ell\delta$. First note that \begin{equation}\label{eq:kdelta}
\ell \delta= \ell \mu^2/n^3 = \frac{1}{n^3}\left(\frac{\alpha}{\beta \log^2(n)}\right)^2.
\end{equation}

From the above, we have $\delta^\star \leq \alpha^2/n^3$ by picking $\beta$ to be suitably large.

\begin{align*}
    \epsilon^\star &= \mu\sqrt{2\ell\ln(1/\ell\delta)} +\ell\mu(e^{\mu}-1) \\
    &\leq \mu\sqrt{2\ell\ln(1/\ell\delta)} +2\ell\mu^2 \\
    &\leq \mu\sqrt{2\ell\ln(1/\ell\delta)}+2\mu\sqrt{\ell}\\
    &\leq 3\mu\sqrt{2\ell\ln(1/\ell\delta)} \\
    &\leq {3\sqrt{2} \alpha} \cdot   \frac{\ln(1/\ell\delta)} {\beta \log^2(n)} \\
    &\leq \frac{\ln(1/\ell\delta)} {\beta \log^2(n)}
\end{align*}
Where the last line follow by $\alpha \leq 0.1$. 
By our bound on $\ell\delta$ in \Cref{eq:kdelta}, we can choose $\beta$ to be a suitably large constant to have $\epsilon^\star  \leq 0.1$. Hence, we have that $\mcC$ is $(0.1, \alpha^2/n^3)$-differentially private as needed. 
\end{proof}

\begin{lemma}
    $\mcP_{\mathrm{DP}}$ respects \Cref{axiom:amplification-formal} with the following parameters: $p(n,1/\beta)= \Delta \cdot \frac{n^{10}}{\beta}$ for some large enough constant $\Delta \geq 1$. The resulting algorithm $\mcM'$ is $(\beta, 5\beta)$-equivalent to $\mcM$.
\end{lemma}
\begin{proof}
    let $\calM:X^n \to Y$. Assume $\mcP_{\mathrm{DP}}(\calM) \leq 1$ which implies $\calM$ is $(0.1, {\alpha^2}/{n^{3}})$-DP. 
    Let $k \geq p(n, \frac{1}{\beta})$ and $m:=kn$, we first use \Cref{lem:DP-to-DP} with the following parameters:

$$\epsilon=  n^2\log^3(m)/ m \text{ and }\delta=\epsilon^2/m^{3}$$

The resulting algorithm $\calM'$ is $(\beta, 5\beta)$-equivalent to $\mcM$. Recall that $m \geq k \geq \Delta n^{10}/\beta$ and $\Delta$ is large enough. So $\mcM'$ uses:

\begin{align*}
    O\Big(\frac{\log(1/\beta)^2+\log(1/\beta)\log(1/\delta)}{\epsilon}\Big) \cdot n^2 = O\left(\frac{\log^2(m)}{n^2 \log^3(m)/m} \right)\cdot n^2  \leq m \text{ samples}.
\end{align*}

We also have that $\calM'$ is $(\epsilon, \epsilon^2/m^{3})$-DP. Using $m/k =n$ and $k \geq n^{10}$ we have $$\epsilon \leq n^{2}\log^3(m)/m=  \frac{m\log^3(m)}{k^2}=\frac{n \log^3(kn)}{k} \leq \frac{2\log^3(k)}{k^{9/10}}.$$

Since $k \geq \Delta$ for some large enough constant $\Delta$, we can conclude $\epsilon \leq \frac{1}{k^{4/5}}$. And thus $\mcP_{\mathrm{DP}}(\mcM')=\left(\epsilon/\alpha\right)^{5/4} \leq \alpha^{-5/4}/ k = O(1/k)$. 
\end{proof}

\section{Minimality of our axioms and proof of \Cref{thm:minimal-intro}}
\label{sec:axioms-min}

As discussed in \Cref{subsec:overview-axioms-necessary}, removing any of our axioms would lead to ill-behaved notions of privacy. In this section, we formalize the claims made there. We will furthermore show that these ill-behaved notions of privacy all allow algorithms that solve one of the following two tasks:

\begin{definition}[The task $\FindElement$]
    \label{def:find-element}
    For any domain $X$, the task $\FindElement$ is defined as follows: An algorithm given i.i.d. samples $\bS \sim \mcD^n$ from an unknown $\mcD$ should output some $x$ for which $\mcD(x) > 0$.
\end{definition}
The next task is similar, but only defined on distributions with one heavy element, and the algorithm's task is to output a different element in the support.
\begin{definition}[The task $\FindLightElement$]
    \label{def:find-light-element}
    For any domain $X$, the task $\FindLightElement$ is defined only on distributions $\mcD$ where there is some $x$ satisfying $0.7 \leq \mcD(x) \leq 0.9$. To solve it, an algorithm given i.i.d. samples $\bS \sim \mcD^n$ should output some $y \neq x$ satisfying $\mcD(y) > 0$.
\end{definition}
We will show in \Cref{subsec:DP-hardness-of-finding-elements} that neither $\FindElement$ nor $\FindLightElement$ can be solved by a DP algorithm using less than $O(\sqrt{|X|})$ many samples. Hence, by showing that the removal of any one axiom allows for solving either $\FindElement$ or $\FindLightElement$ with $O(1)$ samples, we complete the proof of \Cref{thm:minimal-intro}.

\subsection{Removing each requirement of \Cref{axiom:preprocessing-formal}}
\label{subsec:remove-axiom-1}

Since \Cref{axiom:preprocessing-formal} has two requirements, we will show that removing each requirement on its own is enough to allow ill-behaved privacy measures. In both cases, the privacy measure we define will have the following special structure.

\begin{definition}[Binary privacy measure]
    \label{def:binary}
    A privacy measure $\mcP$ is \emph{binary} if there some set of ``good" algorithms $\mcG$ and
    \begin{equation*}
        \mcP(\mcM) = \begin{cases}
            0&\text{if }\mcM \in \mcG \\
            2&\text{otherwise.}
        \end{cases}
    \end{equation*}
\end{definition}
The reason binary privacy measures are easy to analyze is because our axioms say very little about how a privacy measure should behave at privacy levels about $1$ (i.e. none of \Cref{axiom:blalant-formal,axiom:strong-composition-formal,axiom:amplification-formal} apply in that regime. Note the choice to not require our axioms to enforce much when the privacy level is high was not artificial: Even DP starts to behave different when $\eps > 1$ (for example, DP only satisfies strong composition when $\eps \leq 1$ and for larger $\eps$, satisfies linear composition).

We now state how binary privacy measures interact with our axioms:
\begin{claim}[\Cref{axiom:preprocessing-formal} for binary privacy measures]
    \label{claim:binary-pre}
    A binary privacy measure with set $\mcG$ satisfies \Cref{axiom:preprocessing-formal} if and only if:
    \begin{enumerate}
        \item \textbf{$\mcG$ is closed under reordering inputs} For any algorithm $\mcM:X^n \to Y$ and permutation $\pi:[n] \to [n]$ $\mcM \in \mcG \iff \mcM \circ \pi \in \mcG$.
        \item \textbf{Remapping the domain maintains $\mcG$:} For any mapping $\sigma:X \to X$ and algorithm $\mcM:X^n \to  Y$, $\mcM \in \mcG \implies \mcM \circ \sigma \in \mcG$.
    \end{enumerate}
\end{claim}
\begin{proof}
    This is immediate from the definition of \Cref{axiom:preprocessing-formal} and \Cref{def:binary}.
\end{proof}

\begin{claim}[\Cref{axiom:blalant-formal} for binary privacy measures]
\label{claim:binary-blatant}
    A binary privacy measure with set $\mcG$ satisfies \Cref{axiom:blalant-formal} if and only if every $\mcM \in \mcG$ is not blatantly non-private.
\end{claim}
\begin{proof}
    This is immediate from the definition of \Cref{axiom:blalant-formal} and \Cref{def:binary}.
\end{proof}

\begin{claim}[\Cref{axiom:strong-composition-formal} for binary privacy measures]
\label{claim:binary-strong-composition}
    A binary privacy measure with set $\mcG$ satisfies \Cref{axiom:strong-composition-formal} iff for any $\mcM_1, \ldots, \mcM_k:X^n \to Y$ in $\mcG$, the algorithm that takes as input $S \in X^n$ and outputs $(\mcM_1(S), \ldots, \mcM_k(S))$ is also in $\mcG$
\end{claim}
\begin{proof}
   This is immediate from the definition of \Cref{axiom:strong-composition-formal} and \Cref{def:binary}.
\end{proof}

\begin{claim}[\Cref{axiom:amplification-formal} for binary privacy measures]
\label{claim:binary-amplification}
    A binary privacy measure with set $\mcG$ satisfies \Cref{axiom:amplification-formal} iff there exists some polynomial $p:\R^2 \to \R$, so that, for any $\mcM:X^n \to Y$ in $\mcG$, failure probability $\beta > 0$, and $k \geq p(n,1/\beta)$,  there exists some $(\beta, \beta' = O(\beta))$-equivalent algorithm $\mcM'$ taking in $m \coloneqq kn$ that is also in $\mcG$.
\end{claim}
\begin{proof}
   This is immediate from the definition of \Cref{axiom:amplification-formal} and \Cref{def:binary}.
\end{proof}

\subsubsection{Removing the requirement that reordering the input maintains privacy}

In this case, we show that a binary measure privacy where the good algorithms are those that only depend on the first half of their dataset satisfies the remaining axioms.
\begin{definition}[First-half-only-privacy measure]
    \label{def:first-half-privacy}
    The \emph{first-half-only-privacy measure} $\mcP_{\text{half}}$ is the binary privacy measure for set $\mcG$ defined as follows: An $\mcM:X^n \to Y$ is in $\mcG$ iff $n \geq 2$ and there exists some (possibly randomized) algorithm $f:X^{\floor{n/2}} \to Y$ for which $\mcM(S)$ and $f(S_{\leq \floor{n/2}})$ are equal in distribution for all $S$.
\end{definition}
\begin{claim}
    \label{claim:remove-preprocessing-1}
    $\mcP_{\text{half}}$ satisfies \Cref{axiom:blalant-formal,axiom:strong-composition-formal,axiom:amplification-formal} and also the ``remapping the domain maintains privacy" part of \Cref{axiom:preprocessing-formal}.
\end{claim}
\begin{proof}
    This is immediate from \Cref{claim:binary-pre,claim:binary-blatant,claim:binary-strong-composition,claim:binary-amplification}. For \Cref{claim:binary-amplification}, given any $\mcM:X^n \to Y$, we take $\mcM':X^m \to Y$ to be the algorithm that runs $\mcM$ on the first $n$ points of its dataset. Since $\mcM$ depends on only the first half of its dataset, the same will be true of $\mcM'$ (indeed $\mcM'$ will depend on \emph{less} than the first half of its dataset).
\end{proof}
We also observe that this privacy measure allows an algorithm which solves $\FindElement$.
\begin{fact}
    \label{fact:remove-pre-1-find-element}
    The algorithm $\mcM:X^n \to Y$ which, on input $S$, outputs $S_1$ both is $\mcP_{\text{half}}$-private and solves $\FindElement$ with failure probability $0$.
\end{fact}

\subsubsection{Removing the requirement that remapping the domain maintains privacy}

We now will remove the second part of \Cref{axiom:preprocessing-formal} and give a different ill-behaved definition of privacy. This definition will allow algorithms to behave arbitrarily when there is one very heavy element in their sample; however, if there is no such heavy element, the algorithm must output a single uninformative symbol $y^{\star}$.

\begin{definition}[Heavy-elements-only-privacy]
    \label{def:heavy-privacy}
    We say a dataset $S \in X^n$ is ``heavy" if there is a single $x \in X$ appearing at least $0.6n$ times in $S$. The \emph{heavy-elements-only-privacy measure} $\mcP_{\text{heavy}}$ is the binary privacy measure with $\mcM:X^n \to Y \in \mcG$ iff\,\footnote{This requirement that $n \geq 40$ is only to make the proof of \Cref{claim:heavy-comparison} easier.} $n\geq 40$ and the following holds: There is some special output $y^{\star} \in Y$ s.t. for any $S \in X^n$ that is not heavy, $\mcM(S)$ always outputs $y^{\star}$.
\end{definition}
\begin{claim}
    \label{claim:remove-preprocessing-2}
    $\mcP_{\text{heavy}}$ satisfies \Cref{axiom:blalant-formal,axiom:strong-composition-formal,axiom:amplification-formal} and also the ``reordering the input maintains privacy" part of \Cref{axiom:preprocessing-formal}.
\end{claim}
\begin{remark}[Permuting vs remapping the domain]
    \label{remark:permute-vs-remap}
    \Cref{claim:remove-preprocessing-2} shows that if we fully remove the part of \Cref{axiom:preprocessing-formal} that requires privacy is maintained whenever the domain is remapped according to all $\sigma:X \to X$, then $\mcP_{\text{heavy}}$ is a valid (and ill-behaved) measure of privacy. We actually observe something stronger: An alternative definition of \Cref{axiom:preprocessing-formal} is to only require privacy is maintained when the domain is \emph{permuted} according to some bijective function $\sigma:X \to X$. In this case, $\mcP_{\text{heavy}}$  would satisfy all the axioms. Therefore, $\mcP_{\text{heavy}}$ justifies why \Cref{axiom:preprocessing-formal} requiring privacy is preserved for all remappings rather than permutations of the domain is essential.
\end{remark}

Our proof of \Cref{claim:remove-preprocessing-2} will use the following result.
\begin{claim}
    \label{claim:heavy-comparison}
    There exists an absolute constant $c \leq 20$ such that for, $n \geq 40$, $m \geq 2n+1$, and distribution $\mcD$,
    \begin{equation*}
        \Prx_{\bS \sim \mcD^n}[\bS\text{ is heavy}] \geq \frac{1}{c} \cdot \Prx_{\bS' \sim \mcD^m}[\bS'\text{ is heavy}].
    \end{equation*}
\end{claim}
Our proof of \Cref{claim:heavy-comparison} will use two (nontrivial) bounds on the behavior of hypergeometric random variables.
\begin{fact}[Median of hypergeometric is close to its mean, \cite{Sie01,CE09}]
    \label{fact:hyper-median}
    Let $\bx$ be any hypergeometric random variable with mean $\mu$. Then, the median of $\bx$ is either $\floor{\mu}$ or $\ceil{\mu}$.
\end{fact}

\begin{fact}[Log-concave distributions are spread out, \cite{MP25,Ara24}]
    \label{fact:log-concave-spread}
    Let $\mcD$ be a distribution supported on $\N$ that is log-concave (which includes all hypergeometric distributions) and has variance $\sigma^2$. Then, for all $x \in \N$,
    \begin{equation*}
        \mcD(x) \leq \frac{1}{\sqrt{1 + \sigma^2}}.
    \end{equation*}
\end{fact}
\begin{proof}[Proof of \Cref{claim:heavy-comparison}]
    One way to draw $\bS \sim \mcD^n$ is to first draw $\bS' \sim \mcD^m$ and then draw $\bS$ uniformly without replacement from $\bS'$. We will show that for any fixed choice of heavy $S'$, if we draw $\bS$ uniformly without replacement from $S'$, then $\bS$ is heavy with probability at least $1/c$. Then, the desired result easily follows from the following series of inequalities:
    \begin{equation*}
        \Pr[\bS\text{ is heavy}] \geq \Pr[\bS\text{ is heavy}\mid \bS'\text{ is heavy}] \geq \frac{1}{c} \cdot \Pr[\bS'\text{ is heavy}].
    \end{equation*}
    
  Since $S'$ is heavy, there exists a single element appearing $k$ times where $k \coloneqq pm$ and $p \geq 0.6$. Let $\bx$ be the random variable indicating the number of times this element appears in $S$. Then, $\bx$ is drawn from a hypergeometric distribution with mean $\mu \coloneqq pn$ and variance
    \begin{equation*}
        \sigma^2 \coloneqq np(1-p)\frac{m-n}{m-1}.
    \end{equation*}
    Since $m \geq 2n+1$, we will have $\sigma^2 \geq np(1-p)/2$.  Then, by \Cref{fact:hyper-median},
    \begin{equation*}
        \Pr[\bx \geq \floor{pn}] \geq 1/2.
    \end{equation*}
    In particular, if $\floor{pn} \geq 0.6n$ we are done. Otherwise, we will have $pn-1 \leq 0.6n$ in which case $p$ must be between $0.6$ and $0.6 + 1/n$. For $n \geq 40$, this, means $0.6 \leq p \leq 0.625$, in which case the variance is at least $\sigma^2 \geq 40 \cdot 0.625 \cdot 0.375 /2 \geq 4$. Then, we can bound,
    \begin{equation*}
        \Pr[\bx \geq 0.6n] \geq \Pr[\bx \geq \floor{pn}] - \Pr[\bx \geq \floor{pn} \text{ and }\bx < 0.6n].
    \end{equation*}
    Since $pn \geq 0.6$, there is at most one value for $x$ satisfying $x \geq \floor{pn}$ and $x < 0.6n$. Using \Cref{fact:log-concave-spread}, the probability $\bx$ takes on this one value is at most $\frac{1}{\sqrt{1+\sigma^2}} \leq \frac{1}{\sqrt{5}}$. Therefore,
    \begin{equation*}
        \Pr[\bx \geq 0.6n] \geq \frac{1}{2} - \frac{1}{\sqrt{5}} \geq \frac{1}{20}. \qedhere
    \end{equation*}

\end{proof}

We are now ready to prove that $\mcP_{\mathrm{heavy}}$ satisfies all of our axioms except for the ``reordering the domain maintains privacy" part of \Cref{axiom:preprocessing-formal}.
\begin{proof}[Proof of \Cref{claim:remove-preprocessing-2}]
    The ``reordering the input maintains privacy" part of \Cref{axiom:preprocessing-formal} is immediate from \Cref{claim:binary-pre}. Similarly, \Cref{axiom:strong-composition-formal} holds immediately from \Cref{claim:binary-strong-composition}. \Cref{axiom:blalant-formal} holds by \Cref{claim:binary-blatant} and the observation that for any $\mcD$ satisfying $\linf{\mcD} \leq 1/(100n^2)$ it is unlikely for $n \geq 40$ that $\bS \sim \mcD^n$ is heavy.

    The proof that \Cref{axiom:amplification-formal} holds requires a bit more care: Let $\mcM:X^n \to Y$ be any algorithm that is $\mcP_{\mathrm{heavy}}$-private. Then, it has an input $y^{\star}$ it outputs whenever it does not have a heavy input. For any $m \geq 2n+1$, let $\mcM':X^m \to Y$ be the algorithm that on input $S' \in X^m$ does the following:
    \begin{enumerate}
        \item It draws a uniform size-$n$ subsample $\bS$ without replacement from $S'$ and defines $
        \by = \mcM(S)$.
        \item If $S'$ is heavy, it outputs $\by$. Otherwise, it outputs $y^{\star}$.
    \end{enumerate}
    Its clear that $\mcP_{\mathrm{heavy}}(\mcM') = 0$. Therefore all that remains is to show that $\mcM'$ is $(\beta, \beta' = O(\beta))$-equivalent to $\mcM$ for all $\beta$. We will prove this with $\beta' = (c+1)\cdot \beta$ where $c$ is the constant in \Cref{claim:heavy-comparison}. For this, we first observe that the definition of $\mcM'$ immediately implies that for any distribution $\mcD$ there exists a coupling of $\by' \coloneqq \mcM'(\bS')$ and $\by \coloneqq \mcM(\bS)$ where $\bS \sim \mcD^n$ and $\bS' \sim \mcD^m$ satisfying that either $\by' = \by$ or $\by' = y^{\star}$. Furthermore, the latter case occurs with probability at most $\Pr[\bS'\text{ is heavy}]$.
    
    Now, let $\mcT$ be any statistical task that $\mcM$ solves with failure probability $\beta$. On any input distribution $\mcD$, if $y^{\star}$ is a valid output for the task $\mcT$ on input distribution $\mcD$, then the failure probability of $\mcM'$ is strictly less than the failure probability of $\mcD$. In this case, we are done.
    
    Therefore, it remains to handle the case where $y^{\star}$ is not a valid output of $\mcT$ on the distribution $\mcD$. Then it must be the case that $\Pr_{\bS \sim \mcD^n}[\bS \text{ is heavy}] \leq \beta$, as otherwise, $\mcM$ would have too high of a failure probability. In this case, the failure probability of $\mcM'$ is at most the failure probability of $\mcM$ plus the probability that $\bS' \sim \mcD^m$ is heavy. The first quantity is at most $\beta$, and the second at most $c \cdot \beta$ by \Cref{claim:heavy-comparison}, giving the desired bound.
\end{proof}

Next, we observe that this notion of privacy allows for solving the $\FindLightElement$ task:
\begin{fact}
    For any $y^{\star} \in Y$ and $n\geq 40$, let $\mcM:X^n \to Y$ be the algorithm that does the following on input $S$.
    \begin{enumerate}
        \item If $S$ is not heavy, it outputs $y^{\star}$.
        \item If $S$ is heavy, an arbitrary element that appears $\leq n/2$ times in $S$ if one exists. Otherwise, also output $y^{\star}$.
    \end{enumerate}
    Then, $\mcM$ is $\mcP_{\text{heavy}}$-private and solves $\FindLightElement$ with failure probability $\exp(-\Omega(n))$.
\end{fact}

\subsection{Removing \Cref{axiom:blalant-formal}}
\label{subsec:remove-axiom-2}
Removing this axiom leads to the easiest analysis, because we can just make everything private.
\begin{definition}[The privacy measure that allows all algorithms]
    We define $\mcP_{\mathrm{all}}$ to be the privacy measure for which $\mcP_{\mathrm{all}}(\mcM) = 0$ for all algorithms $\mcM$.
\end{definition}
\begin{claim}
    $\mcP_{\mathrm{all}}$ satisfies \Cref{axiom:preprocessing-formal,axiom:strong-composition-formal,axiom:amplification-formal}
\end{claim}
\begin{proof}
    This is immediate.
\end{proof}
Furthermore, since $\mcP_{\mathrm{all}}$ allows all algorithms, it trivially allows algorithms solving $\FindElement$ (such as the algorithm from \Cref{fact:remove-pre-1-find-element}).

\subsection{Replacing \Cref{axiom:strong-composition-formal} with linear composition}
\label{subsec:remove-axiom-3}
Here, we will show that even if we don't fully remove \Cref{axiom:strong-composition-formal} but weaken it to \emph{linear composition} (i.e. $c=1$ in \Cref{axiom:strong-composition-formal}), it still allows for an ill-behaved notion of privacy. This measure of privacy will use is a scaling of \emph{junta}-size. 
\begin{definition}[Juntas]
    \label{def:junta}
    For any $k \in [n]$, an algorithm $\mcM:X^n \to Y$ is a $k$-junta if 
    there exists a size-$k$ $I \subseteq [n]$ and (possibly randomized) function $f:X^k \to Y$ for which, on any input $S \in X^n$, the output distribution $\mcM(S)$ is equal to that of $f(S_{I})$ where
    \begin{equation*}
        S_I \coloneqq \paren*{S_{I_1}, \ldots, S_{I_k}}.
    \end{equation*}
\end{definition}
For example, the algorithm $\mcM:X^n \to X^2$ which outputs the first and last element of its dataset is a $2$-junta.
\begin{definition}[Junta privacy]
    \label{def:privacy-junta} We define the \emph{junta privacy measure}, $\mcP_{\mathrm{junta}}$ as follows: For any $\mcM:X^n \to Y$, we set $\mcP_{\mathrm{junta}}(\mcM) = 2k/n$ where $k$ is the minimum value for which $\mcM$ is a $k$-junta.
\end{definition}
\begin{claim}
    \label{claim:junta-satisfies-axioms}
    $\mcP_{\mathrm{junta}}$ satisfies \Cref{axiom:preprocessing-formal,axiom:blalant-formal,axiom:amplification-formal} and also \Cref{axiom:strong-composition-formal} with $c=1$.
\end{claim}
\begin{proof}
    \Cref{axiom:preprocessing-formal} is immediate since the definition of junta size is maintained under reorderings of the sample and remappings of the domain. \Cref{axiom:blalant-formal} holds because $\mcP_{\mathrm{junta}}$-private algorithm can only depend on half of their dataset, so the adversary cannot guess $0.9$-fraction of the points. For \Cref{axiom:amplification-formal} given any $\mcM:X^n \to Y$ and any $m \geq n$, we construct $\mcM':X^m \to Y$ to just run $\mcM$ on the first $n$ points of its size-$m$ sample. If $\mcM$ is a $k$-junta, then $\mcM'$ will still be a $k$-junta, so $\mcP_{\mathrm{junta}}(\mcM') \leq \mcP_{\mathrm{junta}}(\mcM) \cdot \frac{n}{m}$. Furthermore, the output distribution of $\mcM'$ given a draw from $\mcD^m$ is identical to that of $\mcM$ given a draw of $\mcD^n$, so $\mcM'$ is $(\beta, \beta'=\beta)$-equivalent to $\mcM$ for any choice of $\beta$.

    Finally, \Cref{axiom:strong-composition-formal} with $c=1$ holds because the composition of $\ell$ many $k$-juntas is always an $\ell k$ junta.
\end{proof}
Next, we observe this definition of privacy allows for solving $\FindElement$.
\begin{fact}
    The algorithm $\mcM:X^n \to Y$ which, on input $S$, outputs $S_1$ both is $\mcP_{\text{junta}}$-private for any $n \geq 2$ and solves $\FindElement$ with failure probability $0$.
\end{fact}

\subsection{Removing \Cref{axiom:amplification-formal}}
\label{subsec:remove-axiom-4}
If we remove \Cref{axiom:amplification-formal}, we can just use a rescaled version of \Cref{def:privacy-junta}. This rescaling converts linear composition to strong composition, in exchange for losing linear scalability.
\begin{definition}[Square Root Junta privacy]
    \label{def:privacy-junta-square-root} We define the \emph{square root junta privacy measure}, $\mcP_{\sqrt{\mathrm{junta}}}$ as follows: For any $\mcM:X^n \to Y$, we set $\mcP_{\mathrm{junta}}(\mcM) = \sqrt{2k/n}$ where $k$ is the minimum value for which $\mcM$ is a $k$-junta.
\end{definition}
\begin{claim}
    $\mcP_{\sqrt{\mathrm{junta}}}$ satisfies \Cref{axiom:preprocessing-formal,axiom:blalant-formal} and also \Cref{axiom:strong-composition-formal} with $c=1/2$.
\end{claim}
\begin{proof}
    \Cref{axiom:preprocessing-formal,axiom:blalant-formal} hold exactly as they did in \Cref{claim:junta-satisfies-axioms}. For \Cref{axiom:strong-composition-formal}, we once again use that the composition of any $\ell$ many $k$-juntas is a $\ell k$ junta. In this case, if $\mcM'$ is the composed algorithm and $p = \sqrt{2k/n}$ is the original privacy level, then,
    \begin{equation*}
        \mcP_{\sqrt{\mathrm{junta}}}(\mcM') \leq \sqrt{2k\ell/n} = \sqrt{\ell} \cdot p.\qedhere
    \end{equation*}
\end{proof}
Once again, we have that this notion of privacy allows for solving the $\FindElement$ task.
\begin{fact}
    The algorithm $\mcM:X^n \to Y$ which, on input $S$, outputs $S_1$ both is $\mcP_{\sqrt{\text{junta}}}$-private for any $n \geq 2$ and solves $\FindElement$ with failure probability $0$.
\end{fact}

\subsection{DP-hardness of $\FindElement$ and $\FindLightElement$}
\label{subsec:DP-hardness-of-finding-elements}

Lastly, we show that neither $\FindElement$ nor $\FindLightElement$ have DP algorithms using less than $\approx \sqrt{|X|}$ samples. This completes the proof of \Cref{thm:minimal-intro}.

\begin{claim}[DP-hardness of $\FindElement$]
    \label{claim:find-element-hard}
    For any domain $X$ and $n \leq \sqrt{|X|}/100$, there is no $(\eps = 1, \delta = 1/(10n))$-algorithm $\mcM:X^n \to X$ which solves $\FindElement$ with failure probability at most $1/2$.
\end{claim}
\begin{proof}
    For any algorithm $\mcM:X^n \to X$, let us define
    \begin{equation*}
        p(\mcM) \coloneqq \Prx_{\bS \sim  X^n}[\mcM(\bS) \in S].
    \end{equation*}
    We first argue that any $\mcM$ that solves $\FindElement$ with failure probability at most $1/2$ satisfies $p(\mcM) \geq 1/3$.

    For this, draw $\bX'$ to a uniform size $|X|/10$ subset of $|X|$. Then, let us consider the average failure probability of $\mcM$ over a $\bmcD \coloneqq \Unif(\bX')$,
    \begin{equation*}
        \Ex_{\bX' \sim \binom{X}{|X|/10}}\bracket*{\Ex_{\bS \sim \Unif(\bX')^n}[\Pr[\mcM(\bS) \notin \bX']]} \leq 1/2,
    \end{equation*}
    where the $1/2$ upper bound is because, if $\mcM$ has a failure probability of at most $1/2$ on any single distribution, it also has a failure probability of at most $1/2$ on average over any set of distributions. Next, will switch the order of the above expectation:
    \begin{align*}
        \Ex_{\bX' \sim \binom{X}{|X|/10}}\bracket*{\Ex_{\bS \sim \Unif(\bX')^n}[\Pr[\mcM(\bS) \notin \bX']]} = \Ex_{\bS \sim \Unif(X)^n}\bracket*{\Ex_{\bX' \mid \bS}\bracket*{\Pr[\mcM(\bS) \notin \bX']}}.
    \end{align*}
    Observe that conditioning on $\bS =S$ conditions on all the values in $S$ being part of $\bX'$. In particular, for any $x \notin S$, if we draw $\bX' \mid \bS = S$ the probability that $x \in \bX'$ is at most $1/10$. Therefore,
    \begin{equation*}
         \Ex_{\bS \sim \Unif(X)^n}\bracket*{\Ex_{\bX' \mid \bS}\bracket*{\Pr[\mcM(\bS) \notin \bX']}} \geq (1-p(\mcM))\cdot 0.9.
    \end{equation*}
    Hence, we have the desired condition that $p(\mcM) \geq 1/3$.

    Finally, we'll show that no algorithm with $p(\mcM) \geq 1/3$ can be $(\eps,\delta)$-DP. For this, we will argue that for any such $\mcM$, there exists some neighboring datasets $S, S'$ and element $x$ s.t.
    \begin{equation}
        \label{eq:neighbors-far}
        \Pr[\mcM(S) = x] \leq \frac{1}{100n}\quad\quad\text{and}\quad\quad \Pr[\mcM(S') = x] \geq \frac{1}{5n}.
    \end{equation}
    We show the existence of such an $S, S', x$ via the probabilistic method: Let us draw $\bS \sim \Unif(X)^n$, $\bx \sim \Unif(X)$, and $\bi \sim \Unif([n])$ independently. Then, we set $\bS'$ to be identical to $\bS$ except with the $\bi^{\text{th}}$ element replaced with $\bx$. We will show that,
    \begin{equation}
        \label{eq:unlikely-x}
        \Prx_{\bS, \bx}\bracket*{\Prx_{\text{randomness of }\mcM}[\mcM(\bS) = \bx] \geq 1/100n} \leq 1/100n
    \end{equation}
    For this, we use that $\bS$ and $\bx$ are independent, meaning
    \begin{equation*}
         \Ex_{\bS, \bx}\bracket*{\Prx_{\text{randomness of }\mcM}[\mcM(\bS) = \bx]} = \Pr[\mcM(\bS) = \bx] \leq \frac{1}{|X|}.
    \end{equation*}
    As long as $|X| \geq (100n)^2$, we obtain \Cref{eq:unlikely-x} by Markov's inequality.

    Next, we will show that,
    \begin{equation}
        \label{eq:likely-x}
        \Prx_{\bS', \bx}\bracket*{\Prx_{\text{randomness of }\mcM}[\mcM(\bS') = \bx] \geq 1/5n} \geq 1/10n.
    \end{equation}
    For this, observe the distribution of $\bS'$ is simply $\Unif(X)^n$ and that $\bi$ is independent of $\bS'$. Then, by the definition of $p(\mcM)$,
    \begin{equation*}
        \Prx_{\bS' \sim \Unif(X)^n, \bi \sim \Unif([n])}\bracket*{\mcM(\bS') = \bS'_{\bi}} \geq \frac{p(\mcM)}{n} \geq \frac{1}{3n}.
    \end{equation*}
    Using the fact that $\bS'_{\bi} = \bx$, we can rewrite this as
    \begin{equation*}
        \Ex_{\bS', \bx}\bracket*{\Prx_{\text{randomness of }\mcM}[\mcM(\bS') = \bx]} = \Pr[\mcM(\bS') = \bx] \geq \frac{1}{3n}.
    \end{equation*}
    The above inequality implies \Cref{eq:likely-x} by reverse Markov. Combining \Cref{eq:likely-x} and \Cref{eq:unlikely-x}, we have that with nonzero probability over $\bS, \bS', \bx$, that \Cref{eq:neighbors-far} holds. Since $\bS$ and $\bS'$ are guaranteed to be neighbors, this implies $\mcM$ is not $(\eps, \delta)$-DP.
\end{proof}

We will prove a similar bound for $\FindLightElement$. This proof will use a reduction to \Cref{claim:find-element-hard}.
\begin{claim}[DP-hardness of $\FindLightElement$]
    \label{claim:find-light-element-hard}
    For any domain $X$ and $n \leq \frac{\sqrt{|X|-1}}{100}$, there is no $(\eps = 1, \delta = 1/(10n))$-algorithm $\mcM:X^n \to X$ which solves $\FindLightElement$ with failure probability at most $1/2$.
\end{claim}
\begin{figure}[ht]

  \captionsetup{width=.9\linewidth}

    \begin{tcolorbox}[colback = white,arc=1mm, boxrule=0.25mm]
    \vspace{2pt}
   
    \textbf{Input:} An $(\eps,\delta)$-DP algorithm $\mcM:(X)^n \to (X)$, solving $\FindLightElement$ with failure probability at most $\beta$, single element $x^{\star} \in X$, and sample $S' \in (X')^n$ where $X' \coloneqq X \setminus \set{x^{\star}}$.\\
    
    \textbf{Output: } The output of an algorithm $\mcM'(S')$ where $\mcM':(X')^n \to (X')$ is $(\eps,\delta)$-DP and solves $\FindElement$ with failure probability at most $\beta$.\\
    
 \begin{enumerate}[nolistsep,itemsep=2pt]
        \item  Draw $\bb_1, \ldots, \bb_n \iid \Ber(0.7)$.
        \item Construct the sample $\bS \in X^n$, by setting, for all $i \in [n]$,
        \begin{equation*}
            \bS_i \coloneqq \begin{cases}
                x^{\star}&\text{if }\bb_i=1\\
                S_i'&\text{otherwise.}
            \end{cases} 
        \end{equation*}
        \item Let $\bx = \mcM(\bS)$. If $\bx \neq x^{\star}$, output $\bx$. Otherwise, output a uniform element of $X'$.
    \end{enumerate}
    
    \end{tcolorbox}
\caption{Reduction from $\FindElement$ to $\FindLightElement$}
\label{fig:FindElementReduction}
\end{figure}

\begin{proof}

    Suppose we had an algorithm $\mcM:X^n \to X$ that was $(\eps, \delta)$-DP and solved $\FindLightElement$ with failure probability at most $1/2$ over domain $X$. Let $X'$ be created by removing any one element, $x^\star$ from $X$ (i.e. $X' \cup \set{x^\star}$). We will show that the algorithm $\mcM':(X')^n \to X'$ in \Cref{fig:FindElementReduction} is $(\eps,\delta)$-DP and solves $\FindElement$ with failure probability at most $1/2$. The desired result then follows from \Cref{claim:find-element-hard}.

    We first show that $\mcM'$ solves $\FindElement$ with failure probability at most $1/2$. Observe that for any distribution $\mcD'$ over $X'$, if we draw $\bS' \sim (\mcD')^n$ and then create the sample $\bS$ as in \Cref{fig:FindElementReduction}, then the distribution of $\bS$ is $\mcD^n$ where $\mcD$ is the distribution that puts $0.7$ mass on $x^{\star}$ and the remaining $0.3$ mass is distributed according to $\mcD'$. Hence, since $\mcM$ solves $\FindLightElement$ with failure probability at most $1/2$, the output $\bx = \mcM(\bS)$ must be some element satisfying $\mcD'(\bx) > 0$ with probability at least $1/2$. Therefore $\mcM'$ solves $\FindElement$ with failure probability at most $1/2$
    
    Lastly, we show that $\mcM'$ is $(\eps, \delta)$-DP whenever $\mcM$ is $(\eps, \delta)$-DP. For any $Y' \subseteq X'$ and $(S^{(1)})', (S^{(2)})' \in (X')^n$ differing in one coordinate, we wish to show that
    \begin{equation*}
        \Pr[\mcM'((S^{(1)})') \in Y'] \leq e^{\eps}\cdot  \Pr[\mcM'((S^{(2)})') \in Y']  + \delta.
    \end{equation*}
    Since the choices of ${\bb_1}, \ldots, {\bb_n}$ in \Cref{fig:FindElementReduction} are independent of the input sample, it suffices to show that for any fixed choice of $b\coloneqq b_1,\ldots, b_n$,
    \begin{equation}
        \label{eq:dp-fixed-b}
        \Pr[\mcM'((S^{(1)})') \in Y'\mid \bb=b] \leq e^{\eps}\cdot  \Pr[\mcM'((S^{(2)})') \in Y'\mid \bb=b]  + \delta.
    \end{equation}
    Fix any choice of $b$. Then, let $S^{(1)}$ and $S^{(2)}$ respectively be the samples $S$ constructed in \Cref{fig:FindElementReduction} when given inputs $(S^{(1)})'$ and $(S^{(2)})'$ respectively, and when $\bb=b$. Then, we observe that $S^{(1)}$ and $S^{(2)}$ also differ in at most one coordinate. Therefore, since $\mcM$ is $(\eps,\delta)$-DP, for any $Y \subseteq X$,
    \begin{equation*}
        \Pr[\mcM((S^{(1)})) \in Y] \leq e^{\eps}\cdot  \Pr[\mcM((S^{(2)})) \in Y]  + \delta.
    \end{equation*}
    Since the output of $\mcM'(S')$ is formed by postprocessing the output of $\mcM(S)$, the above implies \Cref{eq:dp-fixed-b} holds for all fixed values of $b$. This means that $\mcM'$ is indeed $(\eps,\delta)$-DP, and the desired result follows from the impossibility result of \Cref{claim:find-element-hard}.
\end{proof}

\section{Acknowledgments}
The authors thank the anonymous ITCS reviewers for their helpful feedback. Guy is supported by NSF awards 1942123, 2211237, and 2224246 and a Jane Street Graduate Research Fellowship. William is supported by NSF awards 2106429, 2107187. Toniann is supported by the Simons Foundation Collaboration on the Theory of Algorithmic Fairness. 






\bibliographystyle{alpha}
\bibliography{ref}

@inproceedings{F17,
  title={A general characterization of the statistical query complexity},
  author={Feldman, Vitaly},
  booktitle={Conference on learning theory},
  pages={785--830},
  year={2017},
  organization={PMLR}
}

@inproceedings{ILPS22,
  title={Reproducibility in learning},
  author={Impagliazzo, Russell and Lei, Rex and Pitassi, Toniann and Sorrell, Jessica},
  booktitle={Proceedings of the 54th annual ACM SIGACT symposium on theory of computing},
  pages={818--831},
  year={2022}
}

@inproceedings{KL10,
  title={Towards an axiomatization of statistical privacy and utility},
  author={Kifer, Daniel and Lin, Bing-Rong},
  booktitle={Proceedings of the twenty-ninth ACM SIGMOD-SIGACT-SIGART symposium on Principles of database systems},
  pages={147--158},
  year={2010}
}

@article{DRS22,
  title={Gaussian differential privacy},
  author={Dong, Jinshuo and Roth, Aaron and Su, Weijie J},
  journal={Journal of the Royal Statistical Society: Series B (Statistical Methodology)},
  volume={84},
  number={1},
  pages={3--37},
  year={2022},
  publisher={Wiley Online Library}
}

@article{S24,
  title={A Statistical Viewpoint on Differential Privacy: Hypothesis Testing, Representation, and Blackwell's Theorem},
  author={Su, Weijie J},
  journal={Annual Review of Statistics and Its Application},
  volume={12},
  year={2024},
  publisher={Annual Reviews}
}

@inproceedings{DMNS06,
  title={Calibrating noise to sensitivity in private data analysis},
  author={Dwork, Cynthia and McSherry, Frank and Nissim, Kobbi and Smith, Adam},
  booktitle={Theory of cryptography conference},
  pages={265--284},
  year={2006},
  organization={Springer}
}

@inproceedings{BGHILPSS23,
author = {Bun, Mark and Gaboardi, Marco and Hopkins, Max and Impagliazzo, Russell and Lei, Rex and Pitassi, Toniann and Sivakumar, Satchit and Sorrell, Jessica},
title = {Stability Is Stable: Connections between Replicability, Privacy, and Adaptive Generalization},
year = {2023},
isbn = {9781450399135},
publisher = {Association for Computing Machinery},
address = {New York, NY, USA},
url = {https://doi.org/10.1145/3564246.3585246},
doi = {10.1145/3564246.3585246},
booktitle = {Proceedings of the 55th Annual ACM Symposium on Theory of Computing},
pages = {520–527},
numpages = {8},
location = {Orlando, FL, USA},
series = {STOC 2023}
}

@article{DR14book,
author = {Dwork, Cynthia and Roth, Aaron},
title = {The Algorithmic Foundations of Differential Privacy},
year = {2014},
issue_date = {Aug 2014},
publisher = {Now Publishers Inc.},
address = {Hanover, MA, USA},
volume = {9},
number = {3–4},
issn = {1551-305X},
url = {https://doi.org/10.1561/0400000042},
doi = {10.1561/0400000042},
journal = {Found. Trends Theor. Comput. Sci.},
month = aug,
pages = {211–407},
numpages = {197}
}

@incollection{Vad17book,
  title={The complexity of differential privacy},
  author={Vadhan, Salil},
  booktitle={Tutorials on the Foundations of Cryptography: Dedicated to Oded Goldreich},
  pages={347--450},
  year={2017},
  publisher={Springer}
}

@article{NH21book,
  title={Differential privacy for databases},
  author={Near, Joseph P and He, Xi and others},
  journal={Foundations and Trends{\textregistered} in Databases},
  volume={11},
  number={2},
  pages={109--225},
  year={2021},
  publisher={Now Publishers, Inc.}
}

@article{census,
  title={The 2020 census disclosure avoidance system topdown algorithm},
  author={Abowd, John M and Ashmead, Robert and Cumings-Menon, Ryan and Garfinkel, Simson and Heineck, Micah and Heiss, Christine and Johns, Robert and Kifer, Daniel and Leclerc, Philip and Machanavajjhala, Ashwin and others},
  journal={Harvard Data Science Review},
  volume={2},
  year={2022},
  publisher={The MIT Press}
}

@techreport{AppleDP,
  title        = {Differential Privacy: Technical Overview},
  institution  = {Apple Inc.},
  note         = {White paper; documents Apple’s local DP deployment and budgets. Accessed Aug 19, 2025},
  url          = {https://www.apple.com/privacy/docs/Differential_Privacy_Overview.pdf},
}

@article{DJY17,
  title={Collecting telemetry data privately},
  author={Ding, Bolin and Kulkarni, Janardhan and Yekhanin, Sergey},
  journal={Advances in Neural Information Processing Systems},
  volume={30},
  year={2017}
}

@article{Google23,
  title={Federated learning of gboard language models with differential privacy},
  author={Xu, Zheng and Zhang, Yanxiang and Andrew, Galen and Choquette-Choo, Christopher A and Kairouz, Peter and McMahan, H Brendan and Rosenstock, Jesse and Zhang, Yuanbo},
  journal={arXiv preprint arXiv:2305.18465},
  year={2023}
}

@article{LWX23,
  title={Smoothed Differential Privacy},
  author={Liu, Ao and Wang, Yu-Xiang and Xia, Lirong},
  journal={Transactions on Machine Learning Research},
  year={2023}
}

@inproceedings{DKMMN06,
  title={Our data, ourselves: Privacy via distributed noise generation},
  author={Dwork, Cynthia and Kenthapadi, Krishnaram and McSherry, Frank and Mironov, Ilya and Naor, Moni},
  booktitle={Annual international conference on the theory and applications of cryptographic techniques},
  pages={486--503},
  year={2006},
  organization={Springer}
}

@inproceedings{DRV10,
  title={Boosting and differential privacy},
  author={Dwork, Cynthia and Rothblum, Guy N and Vadhan, Salil},
  booktitle={2010 IEEE 51st annual symposium on foundations of computer science},
  pages={51--60},
  year={2010},
  organization={IEEE}
}

@inproceedings{KOV15,
  title={The composition theorem for differential privacy},
  author={Kairouz, Peter and Oh, Sewoong and Viswanath, Pramod},
  booktitle={International conference on machine learning},
  pages={1376--1385},
  year={2015},
  organization={PMLR}
}

@article{HWR13,
  title={Random Differential Privacy},
  author={Hall, Robert and Wasserman, Larry and Rinaldo, Alessandro},
  journal={Journal of Privacy and Confidentiality},
  volume={4},
  number={2},
  year={2013},
  publisher={Journal of Privacy and Confidentiality}
}

@inproceedings{WLF16,
  title={On-average kl-privacy and its equivalence to generalization for max-entropy mechanisms},
  author={Wang, Yu-Xiang and Lei, Jing and Fienberg, Stephen E},
  booktitle={International Conference on Privacy in Statistical Databases},
  pages={121--134},
  year={2016},
  organization={Springer}
}

@article{BF16,
  title={Typical stability},
  author={Bassily, Raef and Freund, Yoav},
  journal={arXiv preprint arXiv:1604.03336},
  year={2016}
}

@inproceedings{ACGMMTZ16,
  title={Deep learning with differential privacy},
  author={Abadi, Martin and Chu, Andy and Goodfellow, Ian and McMahan, H Brendan and Mironov, Ilya and Talwar, Kunal and Zhang, Li},
  booktitle={Proceedings of the 2016 ACM SIGSAC conference on computer and communications security},
  pages={308--318},
  year={2016}
}

@article{BBG18,
  title={Privacy amplification by subsampling: Tight analyses via couplings and divergences},
  author={Balle, Borja and Barthe, Gilles and Gaboardi, Marco},
  journal={Advances in neural information processing systems},
  volume={31},
  year={2018}
}

@misc{canonne2020shortnotelearningdiscrete,
      title={A short note on learning discrete distributions}, 
      author={Clément L. Canonne},
      year={2020},
      eprint={2002.11457},
      archivePrefix={arXiv},
      primaryClass={math.ST},
      url={https://arxiv.org/abs/2002.11457}, 
}

@inproceedings{M17,
  title={R{\'e}nyi differential privacy},
  author={Mironov, Ilya},
  booktitle={2017 IEEE 30th computer security foundations symposium (CSF)},
  pages={263--275},
  year={2017},
  organization={IEEE}
}

@inproceedings{BS16,
  title={Concentrated differential privacy: Simplifications, extensions, and lower bounds},
  author={Bun, Mark and Steinke, Thomas},
  booktitle={Theory of cryptography conference},
  pages={635--658},
  year={2016},
  organization={Springer}
}

@article{DR16,
  title={Concentrated differential privacy},
  author={Dwork, Cynthia and Rothblum, Guy N},
  journal={arXiv preprint arXiv:1603.01887},
  year={2016}
}

@inproceedings{BDRS18,
  title={Composable and versatile privacy via truncated cdp},
  author={Bun, Mark and Dwork, Cynthia and Rothblum, Guy N and Steinke, Thomas},
  booktitle={Proceedings of the 50th Annual ACM SIGACT Symposium on Theory of Computing},
  pages={74--86},
  year={2018}
}

@article{steinke22,
  title={Composition of differential privacy \& privacy amplification by subsampling},
  author={Steinke, Thomas},
  journal={arXiv preprint arXiv:2210.00597},
  year={2022}
}

@inproceedings{KKMV23,
  title={Statistical indistinguishability of learning algorithms},
  author={Kalavasis, Alkis and Karbasi, Amin and Moran, Shay and Velegkas, Grigoris},
  booktitle={International Conference on Machine Learning},
  pages={15586--15622},
  year={2023},
  organization={PMLR}
}

@book{DLbook01,
  title={Combinatorial methods in density estimation},
  author={Devroye, Luc and Lugosi, G{\'a}bor},
  year={2001},
  publisher={Springer Science \& Business Media}
}

@article{ABS24,
  title={Optimal hypothesis selection in (almost) linear time},
  author={Aliakbarpour, Maryam and Bun, Mark and Smith, Adam},
  journal={Advances in Neural Information Processing Systems},
  volume={37},
  pages={141490--141527},
  year={2024}
}

@article{Sie01,
  title={Median bounds and their application},
  author={Siegel, Alan},
  journal={Journal of Algorithms},
  volume={38},
  number={1},
  pages={184--236},
  year={2001},
  publisher={Elsevier}
}

@article{CE09,
  title={Linearly bounded liars, adaptive covering codes, and deterministic random walks},
  author={Cooper, Joshua N and Ellis, Robert B},
  journal={arXiv preprint arXiv:0909.0029},
  year={2009}
}

@article{MP25,
  title={A note on statistical distances for discrete log-concave measures},
  author={Marsiglietti, Arnaud and Pandey, Puja},
  journal={Statistics \& Probability Letters},
  volume={216},
  pages={110257},
  year={2025},
  publisher={Elsevier}
}

@article{Ara24,
  title={Entropy-variance inequalities for discrete log-concave random variables via degree of freedom},
  author={Aravinda, Heshan},
  journal={Discrete Mathematics},
  volume={347},
  number={1},
  pages={113683},
  year={2024},
  publisher={Elsevier}
}

@InProceedings{pmlr-v49-cummings16,
  title = 	 {Adaptive Learning with Robust Generalization Guarantees},
  author = 	 {Cummings, Rachel and Ligett, Katrina and Nissim, Kobbi and Roth, Aaron and Wu, Zhiwei Steven},
  booktitle = 	 {29th Annual Conference on Learning Theory},
  pages = 	 {772--814},
  year = 	 {2016},
  editor = 	 {Feldman, Vitaly and Rakhlin, Alexander and Shamir, Ohad},
  volume = 	 {49},
  series = 	 {Proceedings of Machine Learning Research},
  address = 	 {Columbia University, New York, New York, USA},
  month = 	 {23--26 Jun},
  publisher =    {PMLR},
  url = 	 {https://proceedings.mlr.press/v49/cummings16.html}
}

@InProceedings{pmlr-v89-wang19b,
  title = 	 {Subsampled Renyi Differential Privacy and Analytical Moments Accountant},
  author =       {Wang, Yu-Xiang and Balle, Borja and Kasiviswanathan, Shiva Prasad},
  booktitle = 	 {Proceedings of the Twenty-Second International Conference on Artificial Intelligence and Statistics},
  pages = 	 {1226--1235},
  year = 	 {2019},
  editor = 	 {Chaudhuri, Kamalika and Sugiyama, Masashi},
  volume = 	 {89},
  series = 	 {Proceedings of Machine Learning Research},
  month = 	 {16--18 Apr},
  publisher =    {PMLR},
  url = 	 {https://proceedings.mlr.press/v89/wang19b.html}
}

@inproceedings{10.1145/1526709.1526733,
author = {Korolova, Aleksandra and Kenthapadi, Krishnaram and Mishra, Nina and Ntoulas, Alexandros},
title = {Releasing search queries and clicks privately},
year = {2009},
isbn = {9781605584874},
publisher = {Association for Computing Machinery},
address = {New York, NY, USA},
url = {https://doi.org/10.1145/1526709.1526733},
doi = {10.1145/1526709.1526733},
booktitle = {Proceedings of the 18th International Conference on World Wide Web},
pages = {171–180},
numpages = {10},
location = {Madrid, Spain},
series = {WWW '09}
}

@inproceedings{10.1145/2840728.2840747,
author = {Bun, Mark and Nissim, Kobbi and Stemmer, Uri},
title = {Simultaneous Private Learning of Multiple Concepts},
year = {2016},
isbn = {9781450340571},
publisher = {Association for Computing Machinery},
address = {New York, NY, USA},
url = {https://doi.org/10.1145/2840728.2840747},
doi = {10.1145/2840728.2840747},
booktitle = {Proceedings of the 2016 ACM Conference on Innovations in Theoretical Computer Science},
pages = {369–380},
numpages = {12},
location = {Cambridge, Massachusetts, USA},
series = {ITCS '16}
}

\appendix
\section{From TV-Stability to DP}\label{sec:TV-to-DP}
In this section, we give the proof of \Cref{lem:TV-to-DP} which we recall below.
\tvToDP*

We firs recall the following notions of sample-perfect generalization and replicability:

\begin{definition}[Sample perfectly generalization, \cite{pmlr-v49-cummings16,BGHILPSS23}]
    An algorithm $\mcM : X^m\to  Y$ is said to be $(\beta,\epsilon,\delta)$-sample perfectly generalizing, if for every
distribution $D$ over $X$ , with probability at least $1-\beta$ over the draw of $\bS, \bS' \iid D^n$ we have that for every $Y' \subseteq Y:$

$$e^{-\epsilon}\Pr[\mcM(\bS') \in Y'] +\delta \leq \Pr[\mcM(S) \in Y] \leq e^{\epsilon}\Pr[\mcM(\bS') \in Y']+\delta.$$
\end{definition}

\begin{definition}[Replicable algorithms \cite{ILPS22}]
    Let $\mcM : X^n \to Y$, be an algorithm taking a set of $S \in X^n$ and using internal randomness $r$. We say that coin tosses $r$ are $\tau$-good for A on distribution $\mcD$ if there exists a “canonical output” $s_r$
such that $$\Prx_{\bS \sim \mcD^n}[\mcM(\bS, r)=s_r] \geq 1-\tau$$

We say that $\mcM$ is $(\rho,\tau)$-replicable if, for every distribution $\mcD$,
with probability at least $1-\rho$, the coin toss $\br$ drawn by $\mcM$ is $\tau$-good on distribution $\mcD$.
\end{definition}

In \cite{BGHILPSS23}, the authors show how to transform a $(\beta, \epsilon, \delta)$-sample-perfectly generalizing  (for small enough $\beta, \epsilon, \delta)$ algorithm into a $(0.1,0.1)$-replicable one. Then then give a transformation from replicability to $(\epsilon,\delta)$-DP. We do the same: \Cref{lem:tv-to-rep} follows almost immediately from their result, though for completeness, we will spell out the steps. \Cref{lem:rep-to-DP} is follows a similar proof strategy as their result, but we obtain $(\beta, O(\beta))$-equivalence rather than $(\beta, O(\beta \log(1/\beta)))$-equivalence.

\begin{lemma}\label{lem:tv-to-rep}
    There exists a constant $1>\rho^*>0$ such that if $\mcM:X^n \to Y$ is $\rho^*$-TV stable then there exists an algorithm $\mcM':X^n \to Y$ which is $(0.1,0.1)$-replicable. {Furthermore, for any set $S$ the distribution $\mcM'(S)$ is the same as that $\mcM(S)$.}
\end{lemma}

\begin{lemma}\label{lem:rep-to-DP}
    Let $1 \geq \beta, \epsilon, \delta \geq 0$ and $\mcM:X^n \to Y$ be $(0.1,0.1)$-replicable. There exists a $(\beta,5\beta)$-equivalent algorithm $\mcM':X^m \to Y$, which is $(\epsilon,\delta)$-differentially private using

    $$m=n \cdot O\Bigg(\log(1/\beta) \cdot \frac{\log(1/\beta)+\log(1/\delta)}{\epsilon}\Bigg) \text{ samples.}$$
\end{lemma}

Using the two above, the proof of \Cref{lem:TV-to-DP} is straightforward.
\begin{proof}[Proof of \Cref{lem:TV-to-DP}]
    Let $\mcM:X^n \to Y$ be $\rho$-TV stable, where $\rho^*$ is the constant of \Cref{lem:tv-to-rep}. Then there exists a $(\beta,\beta)$-equivalent algorithm $\mcM':X^n\to Y$ which is $(0.1, 0.1)$-replicable. The lemma then follows by applying \Cref{lem:rep-to-DP} to $\mcM'$.
\end{proof}

We now prove \Cref{lem:tv-to-rep}, while we prove \Cref{lem:rep-to-DP} in the next subsection.
\begin{proof}[Proof of \Cref{lem:tv-to-rep}]
Let $\mcD$ by an arbitrary distribution over $X$. Since $\mcM$ is $\rho$-TV stable we have by Markov's inequality that:
    $$\Prx_{\bS,\bS' \iid \mcD^n}\bracket*{\dtv(\mcM(\bS), \mcM(\bS')) \geq \sqrt{\rho}} \leq \sqrt{\rho}.$$

    However, if $\dtv(\mcM(\bS), \mcM(\bS')) \leq \sqrt{\rho}$, we clearly have that for every $Y' \subseteq Y$:

$$\Pr[\mcM(\bS') \in Y'] =  \Pr[\mcM(S) \in Y] \leq \Pr[\mcM(\bS') \in Y'] \pm \sqrt{\rho}.$$

As such, it's easy to see that $\mcM$ is $(\sqrt{\rho}, 0, \sqrt{\rho})$-sample perfectly generalizing. From this point, the proof follows from the work of \cite{BGHILPSS23}. We can first show $\mcM$ is $(\rho^{1/4},0,\sqrt{\rho} + \rho^{1/4})$-perfectly generalizing (See Lemma 3.6 of \cite{BGHILPSS23}). This in turns implies that there exists an algorithm $\mcM':X^n \to Y$ which is $(10 \cdot 20(2\rho^{1/4}+\rho^{1/2}), \frac{1}{10})$-replicable (See Theorem 3.17 and Claim 2.21 of   \cite{BGHILPSS23}). 
\end{proof}

\subsection{The proof of \Cref{lem:rep-to-DP}}

Our proof will follow the strategy of Theorem 3.1 of \cite{BGHILPSS23}: the key idea is that if $\mcM$ is replicable, one can draw many random strings $\br_1, \ldots, \br_k$. For each $\br_i$, we can draw multiple fresh sets $\bS_{i,1}, \ldots, \bS_{i,\ell} \iid D^n$ and run $\mcM(\cdot , \br_i)$ on each of these to get answers $\by_{i,j}$. In \cite{BGHILPSS23}, the authors then run DP-selection on $\mcY=\bracket*{\by_{1,1}, \ldots, \by_{1,\ell}, \ldots \by_{k,1}, \ldots, \by_{k,\ell}}$ to output a frequent element from that set. Note that if a $\br_i$ was $0.1$ good, then most of these runs of $\mcM(\circ , \br_i)$ will output its canonical answer $s_{\br_i}$, so there will be ``frequent answer'' in the dataset $\mcY$. Correctness then follows by bounding how many of the elements in $\mcY$ can ``bad'' answers to the statistical task we're trying to solve.  

We proceed slightly differently, for each sets $\bracket*{\by_{1,1}, \ldots, \by_{1,\ell}}, \ldots ,\bracket*{ \by_{k,1}, \ldots, \by_{k,\ell}}$, we run a slightly different version of DP-selection on each of these sets to get replies $\bh_1, \ldots, \bh_\ell$. In particular, our selection algorithm is such that it outputs $\bot$ with high probability when no element appeared many times in the dataset. Indeed, there should be few random strings $r$ which are $0.5$ replicable but with canonical answer $s_r$ being a ``bad'' answer. Indeed, we can show that if $\bh_i \neq \bot$, it's very likely that $\bh_i$ is a ``good'' answer to the statistical task. As such, it suffices to return an arbitrary $\bh_i$ which isn't $\bot$.

Before diving into the proof, we recall the following results:
\begin{lemma}[\cite{DRV10}]\label{lem:simple_composition}
   Let $\epsilon,\delta>0$ if $\mcM_1, \ldots, \mcM_k$ are $(\epsilon,\delta)$-approximate DP. Then their composition $(\mcM_1, \ldots, \mcM_k)$ is $(k\epsilon,k\delta)$-differentially private. 
\end{lemma}

\begin{theorem}[Theorem 3.6 of \cite{DR14book}]\label{lem:L1-and-Lap}
    Let $f:X^n \to \mathbb{R}$ be a function with sensitivity $1$ \footnote{This means, that for any $S,S' \in X^n$ differ in a single entry $|f(S)-f(S')| \leq 1$}. Then for any $O \subseteq \mathbb{R}$ and $S,S'\in X^n $ differing on single entry:
    $$\Prx_{\boldeta \sim Lap(1/\epsilon)}[f(S)+\boldeta \in O] \leq e^{\epsilon} \Prx_{\boldeta \sim Lap(1/\epsilon)}[f(S')+\boldeta \in O].$$
\end{theorem}

\begin{theorem}[DP-Selection \cite{10.1145/1526709.1526733,10.1145/2840728.2840747,BDRS18}]
    \label{thm:select-DP}
    There exists an $(\epsilon,\delta)$-differentially private algorithm $\textsf{DP-Select}$ such that on input $S \in X^n$, the algorithm has the following property: If $n \geq O(\log(1/\delta)/\epsilon)$ and the most common of the input dataset appears at least $0.2n$ more times than any other element, then $\textsf{DP-Select}$ outputs this element with probability $1$. 
\end{theorem}

With the above results, in mind we present our slightly modified DP-selection algorithm. Which as eluded above returns $\bot$ with high probability when no element appears $60\%$ of the time.

\begin{figure}[ht]
  \captionsetup{width=.9\linewidth}

    \begin{tcolorbox}[colback = white,arc=1mm, boxrule=0.25mm]
    \vspace{2pt}
   
    \textbf{Input:} A set of $S$ of $n$ elements from $X$ and parameters $\epsilon,\delta>0$.  
 \begin{enumerate}[nolistsep,itemsep=2pt]
        \item Let $\textsf{max-freq}(S)$ be the maximum frequency of any element in $S$.
        \item Let $$\mathbf{m}=\textsf{max-freq}(S)+\boldeta \text{ and  }\bh=\textsf{DP-Select}(S,\epsilon,\delta)$$
        where $\boldeta \sim Lap(1/\epsilon)$.  
        
        \item Return $\bh$ if $\mathbf{m} > 0.7n$ and $\bot$ otherwise. 
    \end{enumerate}
    \end{tcolorbox}
\caption{\textsf{Pick-Heavy} : A differentially private algorithm to find an very frequent element of a dataset.  }
\label{fig:simple-block}
\end{figure}

\begin{lemma}\label{lem:Pick-Heavy-full}
   The algorithm \textsf{Pick-Heavy} (See \Cref{fig:simple-block}) is $(2\epsilon,2\delta)$-DP. Furthermore if $$n \geq  O\left( \frac{\log(1/\beta)+\log(1/\delta)}{\epsilon}\right)$$
   the we have the following:
   \begin{enumerate}
   \item If an element appears at least $0.6n$ times in $S$, the algorithm outputs $\bot$ or $h$ at least $1-\beta$. 
    \item If an element $h$ appears at least $0.8n$ times in $S$, the algorithm outputs $h$ with probability at least $1-\beta$. 
    \item  If no element appears more than $0.6n$ times in $h$, the algorithm outputs $\bot$ with probability at least $1-\beta$.
      \end{enumerate}
\end{lemma}
\begin{proof}

    First, note that since, counting the maximal frequency of an element in $S$ is a sensitivity $1$ query, we have by \Cref{lem:L1-and-Lap} that line one line $2$ computing $\boldsymbol{m}$ is done in  $(\epsilon,0)$-differentially private way. 
    Furthermore, since $\textsf{DP-Select}$ is an $(\epsilon, \delta)$-differentially private algorithm, we have that line $2$, by \Cref{lem:simple_composition}, corresponds to running a $(2\epsilon, 2\delta)$-DP algorithm to get the pair $(\boldsymbol{m},\bh)$. Finally line $3$ is just post-processing this pair, so we can conclude that the algorithm is $(2\epsilon, 2\delta)$.

We now turn to the correctness of the algorithm. It is well known (See Fact 3.7 in \cite{DR14book}) that:
    $$\Prx_{\boldeta \sim Lap(1/\epsilon)}[|\boldeta| \geq t/\epsilon]=e^{-t}.$$   

Let $h^\star$ be the most frequent element in $S$. If $h^\star$ appears less than $0.6n$ times, the algorithm will return $\bot$ as long as $\boldeta <0.1n$. 
The proof of the Item (3) immediately follows from our choice of $n \geq O(\log(1/\beta)/\epsilon)$. 

Now assume $h^\star$ appears at least $0.6n$ times. Then, $h^\star$ appears at least $0.2n$ more times than every other element in $S$. Whenever this happens, the algorithm either returns $\bot$ or $h^\star$. From this Item (1) immediately follows.

Finally if $h^\star$ appears at least $0.8n$ times.  By our choice of $n \geq O(\log(1/\beta)/\epsilon)$, we have $\boldeta > -0.1n$, with probability $1-\beta$. As such the algorithm will return $h^\star$ with probability at least $1-\beta$, proving Item (2). 
\end{proof}

Using $\textsf{Pick-Heavy}$, we can now prove \Cref{lem:rep-to-DP}.

\begin{figure}[ht]

  \captionsetup{width=.9\linewidth}

    \begin{tcolorbox}[colback = white,arc=1mm, boxrule=0.25mm]
    \vspace{2pt}
   
    \textbf{Input:} A replicable algorithm $\mcM:X^n \to Y$. Sample to an unknown distribution $\mcD$ over $X$. And parameters $\epsilon, \delta, \beta >0$. \\
    \textbf{Initialization: } Let $k=O(\log(1/\beta))$ and  $\ell=O(\frac{\log( \beta^{-1} \cdot \delta^{-1} )}{\epsilon})$\\
 \begin{enumerate}[nolistsep,itemsep=2pt]
        \item  Draw $\br_{1}, \ldots, \br_{k}$ random strings for $\mcM$.
        \item For each $i \in [k] :$
        \begin{itemize}
            \item Draw sets $\bS_{i,1}, \ldots, \bS_{i,\ell}$ each made of $n$ from $\mcD$. 
            \item Let $\bT^i=[\by_{i,1}, \ldots, \by_{i,\ell}]$ where $\by_{i,j}=\mcM(\bS_{i,j}, \br_i)$.
            \item Run $\textsf{Pick-Heavy} (\bT^i, \epsilon/2, \delta/2)$ to get $\bh^i$.
        \end{itemize}
        \item Output $\bh^\star$ as a uniformly random $\bh^i$ which isn't $\bot$. If no such $\bh^i$ exists, output a uniformly random element of $Y$.
    \end{enumerate}
    
    \end{tcolorbox}
\caption{Replicability to DP reduction}
\label{fig:RepToDP}
\end{figure} 

\begin{proof}[Proof of \Cref{lem:rep-to-DP}]
Given $\mcM$, we consider the algorithm $\mcM'$ obtained by using the transformation of figure \Cref{fig:RepToDP}. We first prove $(\epsilon,\delta)$-differential privacy. For a given draw of $\br_1 \ldots, \br_k$, two neighboring input dataset $S,S'$ can differs in at most one of the $\by_{i,j}$. Which means there is a unique $i^\star$ where $S,S'$ can differ on $\bT^{i^\star}$. By \Cref{lem:Pick-Heavy-full}, we have that each call to \textsf{Pick-Heavy} is $(2\epsilon,2\delta)$-differentially private. So by our choice of parameters for any $O \subseteq X \cup \{\bot\}$ we have:

$$\Pr[\bh^{i^\star} \in O \text{ for } S] \leq e^{\epsilon}\Pr[\bh^{i^\star} \in O \text{ for } S']+\delta.$$

On the other hand, all other $\bh^j$, $j \neq i^\star$ are identically distributed. So clearly on line $3$ for any $O \subseteq X$ we have  
$$\Pr[\bh^\star \in O \text{ for }S ] \leq e^{\epsilon}\Pr[\bh^{*} \in O \text{ for }S']+\delta.$$

We now turn the correctness of the algorithm. We assume $\beta \leq 0.2$, as otherwise the bound of $5\beta$ on the correctness probability is trivial. Fix a distribution $\mcD$ over $X$ and let $Y^*$ be the set of ``good answers'' to the task $\mcT$ under the distribution $\mcD$. 

We let $C$ be the set of possible coin draws for $\mcM$ and $\mcC$ denote the distribution of the coin draws. Since $\mcM$ is replicable, we have if we draw a random string $\br \sim \mcC$, with probability at least $0.9$, it is $0.9$ good. Hence let $G$ be the set of $0.9$ good strings in $C$. We denote by $B$ the set of strings $r$ such that $ \Prx_{S \sim \mcD^n}[\mcM(\bS, \br) \not \in Y^*].$ Since $\mcM$'s output is correct with probability at least $1-\beta$, we have $$\Prx_{\br \sim \mcC} \Big[ \br \in B\Big] \leq 2\beta. $$
Finally, we let $P=C \setminus B$, we have $\Prx_{\br \sim \mcC}[\br \in P] \geq 0.9-2\beta \geq 0.5$.  We will consider the distribution of $\bh^i$ depending on wether $\br_i \in B, P$ or $R \setminus P,B$. 
\begin{itemize}
    \item First, assume we have $r \in P$. This implies that $r$ is $0.1$-replicable and the canonical answer $s_r$ is in $Y^*$. Given $\ell$ independently drawn $\bS_{1}, \ldots, \bS_{\ell} \sim \mcD^n$, by a standard Chernoff bound (and making the hidden constant in $\ell$ large enough), 
we have that with probability at least $1-\beta/40$ less than $0.8\ell$ of the elements in $\bT=[\mcM(\bS^1, r), \ldots, \mcM(\bS_\ell, r)]$ are $s_r$. Assuming this holds, by making the hidden constant in $\ell$ large enough, we can conclude using \Cref{lem:Pick-Heavy-full} (Item $(2)$) that $\textsf{Pick-Heavy}(\bT, \epsilon/2, \delta/2)$ returns $s_r$ with probability at least $1-e^{-0.1\ell\epsilon}\geq 1-\beta/40$. This implies that if $\br_i \in P$, line $2$ of \Cref{fig:RepToDP} returns $\bh^i \in Y^*$ with probability at least $1-\beta/20 \geq 0.99$.

\item If $r \in B$, given $\ell$ independently drawn $\bS_{1}, \ldots, \bS_{\ell} \sim \mcD^n$, in the worst case we have that $\textsf{Pick-Heavy}[\mcM(\bS^1, r) \ldots, \mcM(\bS_\ell, r)]$ always returns $\bh \in Y \setminus Y^*$. 
\item If $r \in R \setminus (B \cup G)$. Then, we must have $\Pr_{\bS \sim \mcD^n}[\mcM(\bS, r) \not \in Y^*] \leq 0.5$. Hence, given $\ell$ independently drawn $\bS_{1}, \ldots, \bS_{\ell} \sim \mcD^n$, it follows by a standard Chernoff bound, that the probability there exists an element $y \in Y \setminus Y^*$ which appears at least $0.6n$ times in $\bT=[\mcM(\bS^1, r), \ldots, \mcM(\bS_\ell, r)]$ is at most $\beta/10$. Assuming no bad answer appears more than $0.6n$ times, we consider two cases:
\begin{itemize}
    \item If some element $h$ appears in $\bT$ $0.6n$ times, this element must be in $Y^*$. In this case \textsf{Pick-Heavy} can only return $h$ or $\bot$. (See Item (1) of \Cref{lem:Pick-Heavy-full}). 
    \item If no element appears in $\bT$ more than $0.6n$ times, \textsf{Pick-Heavy}($\bT$) returns $\bot$ with failure probability $e^{-0.1\ell\epsilon} \leq \beta/10$.(See Item (3) of \Cref{lem:Pick-Heavy-full})  
\end{itemize}
This implies that if $\br_i \in R \setminus (B \cup P)$, line $2$ of \Cref{fig:RepToDP} returns $\bh^i \in Y \setminus Y^*$ with probability at most $\beta/5$.

\end{itemize}

From the above discussion, we have that for each $i \in [k]$ on line $3$ of \Cref{fig:RepToDP}: 
$$\Pr[\bh^i \in Y^*] \geq \Pr[\br_i \in P](1-\beta/20) \geq 0.49$$
\text{ and }
$$\Pr[\bh^i \in Y \setminus Y^*]  \leq \Pr[\br_i \in R \setminus (B \cup P) ]\cdot \beta/5+\Pr[\br_i \in B ] \leq \beta/5+2\beta \leq 2.2\beta.$$

As such the probability the algorithm outputs a random value in $Y$, meaning every $\bh^i=\bot$ is at most $(1-0.6)^{O(k)}\leq \beta/10$. Otherwise, the algorithm returns some $\bh^i \neq \bot$ in which case the probability we return a bad answer is at most $2.2\beta/0.49 \leq 4.9\beta$. So the algorithm of \Cref{fig:RepToDP} has a failure probability of $4.9\beta+\beta/10 \leq 5\beta$.
\end{proof}

\subsection{The proof of \Cref{lem:DP-to-DP}}\label{sec:DP-to-DP}
Finally, we sketch the proof of \Cref{lem:DP-to-DP} (recalled below for convenience). As we already we mentioned, the result is quite similar to that Theorem 6.2 of \cite{BGHILPSS23}. The proof of \Cref{lem:DP-to-DP} is identical, except for the  last step where instead of calling the ``Replicability to DP'' theorem used by \cite{BGHILPSS23} (Theorem 3.1 in \cite{BGHILPSS23}) we instead use \Cref{lem:rep-to-DP} presented earlier in this section. 
\DPtoDP*
\begin{proof}[Sketch of the proof of \Cref{lem:DP-to-DP}]
Let $\alpha$  be a enough small constant and $\mcM: X^n \to Y$ be a $(0.1, \alpha^2/n^3)$ differentially private algorithm. Following the first two steps of the proof of Theorem 6.2 of \cite{BGHILPSS23}, we can get an algorithm $\mcM':X^m \to Y$ which is $(\beta,\beta)$-equivalent to $\mcM$ and using  $m:=O(n^2)$ samples. Furthermore, $\mcM'$ is $(0.1, 0.1)$-replicable. We can now apply \Cref{lem:rep-to-DP} to get a $(\beta,5\beta)$-equivalent algorithm $\mcM^\star=X^{m'} \to Y$ which is $(\epsilon,\delta)$-differentially private using
    \begin{equation*}
        m'=m \cdot O\Bigg(\log(1/\beta) \cdot \frac{\log(1/\beta)+\log(1/\delta)}{\epsilon}\Bigg) \text{ samples.} \qedhere
    \end{equation*}
\end{proof}

\section{Rényi DP}
\label{sec:Reyni}

In this section, we show that Rényi DP fits our axioms, that our axioms imply Rényi DP (up to a $\log \log |Y|$ factor), and that this additional factor is necessary.
\begin{definition}[Rényi differential privacy, \cite{M17}]
    For any $\alpha > 1, \eps > 0$, an algorithm $\mcM:X^n \to Y$ is \emph{$(\alpha, \eps)$-RDP} if, for every $S, S' \in X^n$ differing in only one of the $n$ coordinates,
    \begin{equation*}
        \frac{1}{\alpha - 1} \cdot \ln\paren*{\Ex_{\by \sim \mcM(S')}\bracket*{\paren*{\frac{\Pr[\mcM(S)=\by]}{\Pr[\mcM(S')=\by]}}^\alpha}} \leq \eps.
    \end{equation*}
\end{definition}

We recall some facts about Rényi-DP:
\begin{fact}[\cite{M17}]\label{fact:RDP}
For any $\alpha>0$, Réyni-differential privacy has the following properties:

    \begin{itemize}
        \item Let $\alpha>1$. If $\mcM:X^n \to Y$ is $(\alpha,\epsilon)$-RDP, then for any $S,S' \in X^n$ differing in exactly one coordinate and $Y' \subseteq Y$ we have:

        $$\Pr[\mcM(S) \in Y] \leq \left(e^{\epsilon} \Pr[\mcM(S') \in Y'] \right)^{\frac{\alpha-1}{\alpha}} $$
        \item (Post-Processing) If $\mcM: X^n \to Y$ is $(\alpha,\epsilon)$-RDP, then for any $\mcM' : Y  \to Z$, the algorithm $\mcM' \circ \mcM$ is $(\alpha,\epsilon)$-RDP. 
        \item (Non-Adaptive Composition) If $\mcM, \mcM'$ are $(\alpha,\epsilon)$-RDP algorithm then their non-adaptive composition $(\mcM,\mcM')$ is $(\alpha,2\epsilon)$-RDP. 
        \item If $M$ is $(\epsilon,0)$-DP, then it is $(\alpha, \epsilon)$-RDP. 
    \end{itemize}
    
\end{fact}

\subsection{Axioms imply Rényi DP}
We prove our axioms imply Rényi DP, albeit with now a (small) dependency on $|Y|$. In particular we prove the following:

\begin{theorem}
    \label{thm:axioms-to-RDP}
    Let $\mcP$ be any privacy measure satisfying \Cref{axiom:amplification-formal,axiom:blalant-formal,axiom:preprocessing-formal,axiom:strong-composition-formal} and $\mcM : X^n \to Y$ be any $\mcP$-private algorithm. For any $\eps, \delta, \beta > 0$ and, $c$ the constant in \Cref{axiom:strong-composition-formal} and $p$ the polynomial in \Cref{axiom:amplification-formal}, define
    \begin{equation*}
        m' \coloneqq \tilde{O}\left(\frac{r^2 \cdot n^2 \cdot \log(\log(|Y|))}{\beta^2 \cdot \epsilon}\right) \text{ where } r= \max \left(n \cdot p(n,1/\beta), n^{\frac{1}{1-c}} \right).
    \end{equation*}
    there is an $(2,\eps)$-Rényi differentially private algorithm $\mcM'$ using $m'$ samples that is $(\beta, \beta' \coloneqq O(\beta))$-equivalent to $\mcM$. 
\end{theorem}

The proof will rely on the following lemma. 

 \begin{restatable}{lemma}{reptoRDP}\label{lem:rep2-to-RDP}
    Given a $(0.1,0.1)$-replicable algorithm $\mcM : X^n \to Y$ solving a statistical task $\mcT$ with failure probability $\beta$, there exists is $(2,\epsilon)$-Rényi differentially private algorithm $\mcM':X^m \to Y$ solving task $\mcT$ with failure probability $5 \beta$ and using $$m=n \cdot \tilde{O}\Bigg( \frac{\log^2(\beta^{-1}) \log(\log(|Y|))}{\epsilon}\Bigg) \text{ samples. }$$  
\end{restatable}

We delay the proof of the above, from which \Cref{thm:axioms-to-RDP} follows easily.  

\begin{proof}[Proof of \Cref{thm:axioms-to-RDP}]
     Let $\rho^\star$ be the constant of \Cref{lem:tv-to-rep}. We first apply \Cref{thm:axioms-to-tv-overview} to $\mcM$ to get a $(\beta,\beta')$-equivalent, $\beta'=O(\beta)$, and $\rho^\star$-TV-Stable algorithm $\mcM': X ^{m} \to Y$. We have that $\mcM'$ uses $$m=\tilde{O}(r^2 \cdot n^2/\beta^2) \text{ where }r=\max\left( n \cdot p(n, 1/\beta), n^{\frac{1}{1-c}} \right)\text{ samples. }$$

    The proof then follows by applying \Cref{lem:rep2-to-RDP} to $\mcM'$. 
\end{proof}

\begin{theorem}[Stable selection from \cite{BDRS18}]\label{thm:stable_select}
    For every $0<\rho<1$, and $\omega \geq 1/\sqrt{\rho}$, there exists an algorithm $\mcM:X^n \to X$ for the  such that for every $1 < \alpha < \omega$, $\mcM$ is $(\alpha, \rho \alpha)$-Rényi-DP.  The algorithm uses 
    $$n=O\left( \omega \log\left(1+\frac{\log\left(|Y|\right)}{\rho \omega}\right)\right) \text{ samples,}$$

    and if the most common element appears at least $0.2n$ more times than any other element, the algorithm outputs this element with probability at least $2/3$. 
\end{theorem}

We will need the following corollary:

\begin{corollary}\label{thm:RDP-select}
    Let $0<\epsilon<1$, there exists an $(2,\epsilon)$-Rényi differentially private algorithm $\textsf{RDP-Select}$ such that on input $S \in X^n$, the algorithm has the following property: If $$n \geq \tilde{O}\left(\frac{\log(1/\beta)}{\epsilon} \log(\log(|Y|)\right),$$and the most common of the input dataset appears at least $0.2n$ more times than any other element, the algorithm $\textsf{RDP-Select}$  outputs that element with probability at least $1-\beta$.
\end{corollary}
\begin{proof}
  Let $k=O(\log(1/\beta)), \rho=\epsilon/2k, \omega=O(1/\rho)$ and 

  $$n=O\left( \omega \log\left(1+\frac{\log\left(|Y|\right)}{\rho \omega}\right)\right)=\tilde{O}\left(\frac{\log(1/\beta)}{\epsilon} \log(\log(|Y|)\right).$$

  Let $\mcM:X^n \to X$ be the algorithm obtained from \Cref{thm:RDP-select} using  $\rho, \omega$ as above. We have that $\mcM$ is $(2,\epsilon/k)$-Rényi DP. We will consider the following algorithm $\mcM':X^n \to X$, which on input $S \in X^n$ runs $k$ copies of $\mcM(S)$ in parallel to obtain $(\bh_1, \ldots, \bh_k)$. And finally outputs the most frequent element  among $(\bh_1, \ldots, \bh_k)$ (breaking ties arbitrarily). 

  First, by Non-adaptive Composition and Postprocessing (See \Cref{fact:RDP}), it's easy to see that $\mcM'$ is $(2,\epsilon)$-Rényi DP. Now assume, that there exists some element $h^\star$ which appears at least $0.2n$ more times than any other elements in $S$. It follows by a standard Chernoff bound (and setting the hidden constant in $k$ to be large enough) that with probability at least $1-\beta$ strictly more than half of the $k$ copies of $\mcM$ returned $h^\star$, in which case $\mcM'$ returns $h^\star$. 
\end{proof}

With the above algorithm, the proof of \Cref{lem:rep2-to-RDP} follows the same pattern as the proof of \Cref{lem:rep-to-DP}. We thus only sketch the proof. 
\begin{proof}[Proof sketch of \Cref{lem:rep2-to-RDP}]

If we replace \textsf{DP-Selection} in \textsf{Pick-Heavy} by \textsf{RDP-Selection}, using \Cref{fact:RDP} it's easy to see that the resulting algorithm would be $(2,2\epsilon)$-RDP. Furthermore, the resulting algorithm the same correctness guarantees as the one we gave in \Cref{lem:Pick-Heavy-full} (except the error probability is now $1-2\beta$). Finally, given a $(0.1,0.1)$-replicable algorithm $\mcM:X^n \to Y$ we can use the same reduction as the one of \Cref{fig:RepToDP}, but with the Rényi-DP version of \textsf{RDP-Selection} instead, and setting $\ell=\tilde{O}\left(\frac{\log(1/\beta)}{\epsilon} \log(\log(|Y|)\right)$. By appropriately setting the hidden constants in $\ell$ and $k$ the resulting algorithm would be $(\beta,5\beta)$-equivalent to $\mcM$, and it would use 
\begin{equation*}
n \cdot \tilde{O}\left(\frac{\log(1/\beta)^2}{\epsilon} \log(\log(|Y|)\right) \text{ samples}. \qedhere \end{equation*}
\end{proof}

\subsection{Rényi DP satisfies the axioms}
\label{subsec:Reyni-fits-axioms}

We now define our stability measure $\mcP_{\mathrm{RDP}}$ as follows:
\begin{definition}
    For an algorithm $\mcM$ taking $n$ samples let $\epsilon$ be the smallest value such that $\calM$ is $(2,\epsilon)$-RDP. We set $$\mcP_{\mathrm{RDP}}(\mcM)= \sqrt{ \epsilon}.$$ 
\end{definition}   

\begin{lemma}
    $\mcP_{\mathrm{RDP}}$ respects \Cref{axiom:preprocessing-formal}.  
\end{lemma}
\begin{proof}
We observe for any $S, S'$ that differ in one coordinate and permutation $\pi:[n] \to [n]$, that $\pi(S)$ and $\pi(S')$ still differ in one coordinate. Similarly, for any $\sigma:X \to X$, $\sigma(S)$ and $\sigma(S')$ differ in (at most) one coordinate. Therefore, if $\mcM$ is $(2,\epsilon)$-RDP, it is still $(2,\epsilon)$-RDP after preprocessing.
\end{proof}

To prove RDP respects \Cref{axiom:blalant-formal} we will need the following fact:
\begin{fact}\label{fact:RDP-adversary}
    Let $\mcM:X^n \to Y$, $\mcA:Y \to X^n$ be algorithms. Fix $S \in X^n, i \in [n]$ and $x \in X$. If $\mcM$ is $(\alpha, \epsilon)$-Rényi differentially private then we have:
    $$\Pr\bracket*{S_i \in \mcA(\mcM(S))} \leq \big(e^{\epsilon} \Pr\bracket*{S_i \in \mcA(\mcM(S_{x \to i})) }\big)^{\frac{\alpha-1}{\alpha}}$$
    where $S_{x \to i}$ denotes $S$ with $i$-th element set to $x$. 
\end{fact}
\begin{proof}
    $S$ and $S_{x \to i}$ differ on at most $1$ coordinate. By postprocessing (second item of \Cref{fact:RDP}) meaning $\mcA \circ \mcM$ is also $(\alpha,\epsilon)$-RDP. The result then follows from the first item of \Cref{fact:RDP}. 
\end{proof}

\begin{lemma}
    $\mcP_{\mathrm{RDP}}$ respects \Cref{axiom:blalant-formal}. 
\end{lemma}
\begin{proof}
Let $\mcD$ an arbitrary distribution over $X$ with $||\mcD||_\infty \leq 1/100n^2$.
    Let $\mcM : X^n \to Y$ be an $(2,\epsilon)$-Rényi differentially private algorithm and let $\mcA:Y \to X^n$.
    We will show that:
    $$\Ex_{\substack{\bS \iid \mcD^n \\ \bS' \leftarrow \mcA(\mcM(\bS))}}\bracket*{\sum_{x \in \bS}\mathbbm{1}[x \in \bS']} \leq \frac{e^{\epsilon/2}\sqrt{n}}{10}.$$

In particular, if $\mcP_{\mathrm{RDP}}(\mcM) \leq 1$, which implies $\mcM$ is $(2, 1)$-RDP, we have $\frac{e^{\epsilon/2}\sqrt{n}}{10} \leq 0.2n$ proving the lemma. We want to bound 
    $$ \Ex_{\substack{\bS \iid \mcD^n \\ \bS' \leftarrow \mcA(\mcM(\bS))}}\bracket*{\sum_{x \in \bS}\mathbbm{1}[x \in \bS']} = \sum_{i=1}^n \Prx_{\substack{\bS \iid \mcD^n \\ \bS' \leftarrow \mcA(\mcM(\bS))}}[\bS_i \in \bS'] $$

   Since $\mcM$ is $(2, \epsilon)$-RDP, for any $i \in [n]$ and $x \in X$ we have, by \Cref{fact:RDP-adversary}, that:

    $$\Prx_{\substack{\bS \iid \mcD^n }}\bracket*{\bS_i \in \mcA\left(\mcM(\bS)\right)}  \leq e^{\epsilon/2} \left(\Prx_{\substack{\bS \iid \mcD^n }}\bracket*{\bS_i \in \mcA\left(\mcM(\bS_{x \to i})\right)}\right)^{1/2}.$$
     Where $S_{x \to i}$ is the set obtained by setting the $i$-th element of $S$ to $x$. This implies that: 
     \begin{align*}
         \Prx_{\substack{\bS \iid \mcD^n \\ \bS' \leftarrow \mcA(\mcM(\bS))}}[\bS_i \in \bS'] &= \Prx_{\substack{\bS \iid \mcD^n }}[\bS_i \in \mcA(\mcM(\bS))] \\
         &\leq e^{\epsilon/2} \left(\Prx_{\substack{\bS \iid \mcD^n\\ \bx \sim \mcD }}[\bS_i \in \mcA(\mcM(\bS_{\bx \to i}))]\right)^{1/2}. \\
         &= e^{\epsilon/2} \left(\Prx_{\substack{\bS \iid \mcD^n\\ \bx \sim \mcD }}[\bx \in \mcA(\mcM(\bS))]\right)^{1/2}.\\
     \end{align*}
    Where the last line follows from the symmetry of $\bS_i$ and $\bx$. Hence, we have that: 

     \begin{align*}
         \Ex_{\substack{\bS \iid \mcD^n \\ \bS' \leftarrow \mcA(\mcM(\bS))}}\bracket*{\sum_{x \in \bS}\mathbbm{1}[x \in \bS']} &= \sum_{i=1}^n e^{\epsilon/2} \left(\Prx_{\substack{\bS \iid \mcD^n\\ \bx \sim \mcD }}[\bx \in \mcA(\mcM(\bS))] \right)^{1/2}  \\
         &\leq n \cdot e^{\epsilon/2} \left(\sup_{S \in X^n}\bigg(\Pr_{\bx \sim \mcD}[\bx \in S]\bigg)\right)^{1/2}. \\
     \end{align*} 
    
Since $\linf{D} \leq \frac{1}{100n^2}$, we have that for any $S \in X^n$, $\Pr_{\bx \sim \mcD}[\bx \in S] \leq \frac{n}{100n^2}$. From which we can conclude that for all $n \in \mathbb{N}$, \Cref{axiom:blalant-formal} holds since:   
\begin{equation*}
    \Ex_{\substack{\bS \iid \mcD^n \\ \bS' \leftarrow \mcA(\mcM(\bS))}}\bracket*{\sum_{x \in \bS}\mathbbm{1}[x \in \bS']} \leq \frac{e^{\epsilon/2}\sqrt{n}}{10} \leq 0.2n. \qedhere
\end{equation*}
\end{proof}

\begin{lemma}
     $\mcP_{\mathrm{RDP}}$ respects \Cref{axiom:strong-composition-formal} with composition constant $c=1/2$.
\end{lemma}

\begin{proof} 
We define $\epsilon':= \ell^{1/2} \epsilon$. Say we have algorithms ${\calM}^1, \ldots, {\calM}^\ell$ each taking $n$ samples and with $\ell^{-1/2} \geq \mcP_{\mathrm{RDP}}(\mcM ^i)$. By definition, this implies each $\mcM^i$ is $(2, \frac{1}{\ell})$-RDP. Then, by \Cref{fact:RDP} (Non-adaptive composition), we have that the algorithm $\mcM'=(\mcM^1, \ldots, \mcM^\ell)$ is $(2, 1)$-Renyi differentially private. Thus, $\mcP_{\mathrm{RDP}}(\mcM') \leq 1$, meaning $\mcM'$ is $\mcP_{\mathrm{RDP}}$-private.
\end{proof}

Before proving $\mcP_{\mathrm{RDP}}$ fits \Cref{axiom:amplification-formal}, we will need the following result about amplification for Rényi-differential privacy:\\

\medskip\noindent{\bf Theorem} (Rényi Amplification by subsampling, \cite{pmlr-v89-wang19b})
  {\it Given an algorithm} $\mcM:X^n \to Y$.{\it Let }$m \geq n$ {\it and consider the randomized algorithm} $\mcM' : X^m \to Y$ {\it which on input }$S \in X^m$ {\it subsamples without replacement }$n$ {\it elements from $S$ and runs }$\mcM$ {\it the subsampled dataset.  If }$\mcM$ {\it is }$(2, \epsilon)${\it-RDP, then }$\mcM'${\it is }$(2,\epsilon')$-{\it RDP where} 

  $$\epsilon' \leq \log \left(1+4\left( \frac{n}{m} \right)^2 \cdot \left(e^{\epsilon}-1\right)\right).$$

\medskip

We can now prove the following:
\begin{lemma}\label{lem:}
    $\mcP_{\mathrm{RDP}}$ respects \Cref{axiom:amplification-formal} with $p(n,\frac{1}{\beta})=1$. For any choice of $\beta$, the resulting algorithm $\mcM'$ is $(\beta,\beta)$-equivalent to $\mcM$.
\end{lemma}
\begin{proof}
    Let $\mcM:X^n \to Y$ be $\mcP_{\mathrm{RDP}}$-private. We have that $\mcM$ is $(2,1)$-RDP. Let such that Let $k \geq n$, by applying Rényi amplification by subsampling with $m=nk$, we get an algorithm $\mcM':X^{m} \to Y$ which is $(2, \epsilon)$-RDP where 

    $$\epsilon\leq \log \left(1+4\left( \frac{n}{m} \right)^2 \cdot \left(e-1\right)\right) \leq \log\left(1+8\left(\frac{n}{m}\right)^2\right) \leq 16(n/m)^2=\frac{16}{k^2}.$$

    As such we have that $\mcP_{\mathrm{RDP}}(\mcM') \leq \sqrt{\frac{16}{k^2}}=\frac{4}{k}.$

    Furthermore, $\mcM'$ simply runs $\mcM$ on a subset of its input which is subsampled without replacement. So we clearly have that the distribution $\mcM'(\bS')$ where $\bS' \iid \mcD^m$  and $\mcM(\bS)$ $\bS \iid \mcD^n$ are identical for any distribution $\mcD$ over $X$. So $\mcM'$ and $\mcM$ are $(\beta,\beta)$ equivalent for every choice of $\beta$. 
\end{proof}

\subsection{Tightness of \Cref{thm:axioms-to-RDP}}

\Cref{thm:axioms-to-RDP} converts an algorithm $\mcM:X^n \to Y$ that is $\mcP$-private into a new an equivalent RDP algorithm with a new sample size $m'$ depends on $\log \log |Y|$. Here, we show that $\log \log |Y|$ term is necessary.
\begin{lemma}[Separation between $(\eps,\delta)$-DP and RDP]
    \label{lem:sep-DP-RDP}
    For every output domain $Y$, there is a statistical task $\mcT$ satisfying
    \begin{enumerate}
        \item For every $\eps, \delta$ and $n \coloneqq \poly(1/\eps, \log(1/\delta), \log(1/\beta))$, there is an $(\eps,\delta)$-DP algorithm $\mcM:X^n \to Y$ solving $\mcT$ with failure probability $0$.
        \item Any $(\alpha=2,\eps=1)$-RDP algorithm $\mcM':X^m \to Y$ that solves $\mcT$ with failure probability $1/2$ must use $m = \Omega(\log \log |Y|)$ samples.
    \end{enumerate}
\end{lemma}
Since $(\eps,\delta)$-DP fits our axioms (\Cref{thm:DP-fits-axioms-intro}), \Cref{lem:sep-DP-RDP} implies that the $\log \log |Y|$ term cannot be removed from \Cref{thm:axioms-to-RDP}. This separation is based on stable selection and is similar to the lower bound from \cite{BDRS18}, though we specialize to the $\alpha = 2$ case. \Cref{lem:sep-DP-RDP} follows from the following claim.
\begin{claim}[Weak group privacy for RDP]
    \label{claim:group-privacy-RDP}
    Let $\mcM:X^m \to Y$ be any $(\alpha=2,\eps)$-RDP algorithm. Then, for any $S, S'$ differing in $k$ coordinates and $y\in Y$,
    \begin{equation*}
        \Pr[\mcM(S') = y] \geq (\Pr[\mcM(S) = y])^{2^k} \cdot e^{-(2^k-1)\eps}.
    \end{equation*}
\end{claim}
\begin{proof}
    We prove this by induction on $k$. If $k = 0$, we have the $ \Pr[\mcM(S') = y] \geq  \Pr[\mcM(S) = y]$ which clearly holds because $S$ and $S'$ must be the same. 
    
    For $k \geq 1$, let $S^{\mathrm{(mid)}}$ be a sample that differs from $S'$ in $1$ coordinate and from $S$ in $k-1$ coordinates. Then, by rearranging the first item of \Cref{fact:RDP},
    \begin{equation*}
        \Pr[\mcM(S')=y] \geq e^{-\eps} \cdot \Pr[\mcM(S^{(mid)})=y]^2.
    \end{equation*}
    We apply the inductive hypothesis and obtain
    \begin{align*}
        \Pr[\mcM(S')=y] &\geq e^{-\eps} \cdot \Pr[\mcM(S^{(mid)})=y]^2 \\
        &\geq e^{-\eps} \cdot \paren*{(\Pr[\mcM(S) = y])^{2^{k-1}} \cdot e^{-(2^{k-1}-1)\eps}}^2 \\
        &= (\Pr[\mcM(S) = y])^{2^k} \cdot e^{-(2^k-1)\eps}. \qedhere
    \end{align*}
\end{proof}

\begin{proof}[Proof of \Cref{lem:sep-DP-RDP}]
    Fix any output domain $Y$ and let $X = Y$. Consider the statistical task $\mcT$ that is defined for any distribution $\mcD$ where there is some $x^\star$ with all the mass (i.e. $\mcD(x)  = 1$). For such $\mcD$, the only valid response of the task $\mcD$ is this $x^{\star}$. Put differently, for an algorithm $\mcM:X^n \to Y$ to solve $\mcT$ with failure probability $\beta$ is equivalent to $\mcM(\set{x,x,\ldots, x}) = x$ with probability at least $1-\beta$ for all choices of $x$.

    The first requirement of \Cref{lem:sep-DP-RDP} is satisfied by the selection algorithm of \Cref{thm:select-DP}. Therefore, all that remains is to prove a lower bound against RDP algorithms. Fix any $(\alpha=2,\eps=1)$-RDP algorithm $\mcM':X^m \to Y$ that solves the task $\mcT$. Then, for each $y \in Y$, we have that $\mcM'(\set{y,y,\ldots, y}) = y$ with probability at least $1/2$. Then, since any set $S \in X^m$ differs from the set $\set{y,y,\ldots, y}$ in at most $m$ coordinates, we can apply \Cref{claim:group-privacy-RDP} to obtain for any set $S$ and $y \in Y$,
    \begin{equation*}
        \Pr[\mcM(S) = y] \geq (e^{-1}/2)^{2^m}.
    \end{equation*}
    Since it must be the case that $\sum_{y \in Y}\Pr[\mcM(S) = y] = 1$, we have that,
    \begin{equation*}
        |Y| \cdot  (e^{-1}/2)^{2^m} \leq 1.
    \end{equation*}
    This implies that $m \geq \Omega(\log \log |Y|)$, as desired.
\end{proof}

\end{document}